\documentclass[a4paper,11pt]{article}

\usepackage{microtype,amssymb,amsmath,amsthm,mathrsfs,enumitem,hyperref,braket,bbm,graphicx,xcolor}
\usepackage[export]{adjustbox}
\usepackage[sorting=none, style=numeric, doi=true,isbn=false,url=false,backend=biber,hyperref=true,eprint=true]{biblatex}
\hypersetup{
    colorlinks,
    citecolor=blue,
    linkcolor=blue,
    linktocpage=true,
    urlcolor=blue,
}
\usepackage{vmargin}
\usepackage{latexsym}
\setmarginsrb{2.5cm}{1cm}{2.5cm}{2cm}{12pt}{11mm}{0pt}{13mm}

\newcommand{\diff}{\text{d}}
\newcommand{\norm}[1]{\left\lVert#1\right\rVert}

\newcommand\restr[2]{{% we make the whole thing an ordinary symbol
			\left.\kern-\nulldelimiterspace % automatically resize the bar with \right
			#1 % the function
			\vphantom{\big|} % pretend it's a little taller at normal size
			\right|_{#2} % this is the delimiter
			}}
\newtheorem{theorem}{Theorem}
\newtheorem{lemma}[theorem]{Lemma}
\newtheorem{corollary}[theorem]{Corollary}

 \title{Asymptotic analysis of spin-foams with time-like faces in a new parametrization}
\author{Jos\'{e} Diogo Sim\~{a}o\footnote{j.d.simao@uni-jena.de},   \hspace{2mm} Sebastian Steinhaus\footnote{sebastian.steinhaus@uni-jena.de}
\vspace{3mm}
\\ Theoretisch-Physikalisches Institut \\ Friedrich-Schiller-Universität Jena \\ Max-Wien-Platz 1 \\ 07743 Jena, Germany}

\date{\today}
\numberwithin{equation}{section}

\begin{document}

\maketitle

\begin{abstract}

In this article we study the Conrady-Hnybida extension of the Lorentzian Engle-Pereira-Rovelli-Livine spin-foam model, which admits time-like cells rather than just space-like ones. Our focus is on the asymptotic analysis of the model's vertex amplitude. 

We propose a new parametrization for states associated to time-like 3-cells, from which we derive a closed-form expression for their amplitudes. This allows us to revisit the conditions under which critical points of the amplitudes occur, and we find Regge-like geometrical critical points in agreement with the literature. However, we find also evidence for non-geometrical points which are not dynamically suppressed without further assumptions; the model then does not strictly asymptote to the Regge action, contrary to what one would expect.   
We moreover prove Minkowski and rigidity theorems for Minkowskian polyhedra, extending the asymptotic analysis to non-simplicial spin-foams.
\\

\end{abstract}

\tableofcontents

\section{Introduction}

The existence of a causal structure of space-time, i.e. a degree to which physical objects at different locations may or may not affect one-another, is one of the most profound insights of special and general relativity.  While Riemmanian theories and their quantum counterparts may be useful in exploring some of gravity's features, it is clear that a proper quantum theory of gravity has to unavoidably include a form of causality. Path integral approaches, representing one possible route to quantizing such a theory, should thus be defined with a Lorentzian signature. %for path-integral approaches, this may amount to considering a Lorentzian signature. 
A heuristic argument can be made: if space-time geometry is expected to fluctuate at small scales, and if its quanta (whatever they may be) encode to some level a notion of causality, then  the  causal structure of space-time itself should also fluctuate.   

Spin-foam models are a path integral approach to quantum gravity closely related to Loop Quantum Gravity (LQG) \cite{thomasbook}. While several different models have been proposed over the last decades \cite{Barrett2000,Perez2001,Han:2010pz}, all such models share a common structure: spin-foams are defined on a 2-dimensional cell complex, to which one assigns group-theoretic data (usually unitary irreducible representations of some symmetry group, as well as intertwining maps between them), and from which an amplitude is constructed \cite{Perez:2012wv,carlobook}. While the 2-complex is usually induced from an underlying triangulation of the space-time manifold, generalisations to more generic 2-complexes exist \cite{Kaminski:2009fm}.

Among the theories that have been defined in the literature, the Lorentzian Engle-Pereira-Rovelli-Livine (EPRL) simplicial model \cite{Engle:2007wy}\footnote{See also \cite{Engle:2007uq,Engle:2007qf,Livine:2007vk} for the definition of the model for Riemannian signature.} is one of the most well-studied ones. Its particularity essentially resides in which group-theoretic data is assigned to the 2-complex dual to a triangulation% by 4-simplices, and in which manner one does so
. The initial symmetry group of the theory is taken to be the  $\text{SL}(2,\mathbb{C})$ double cover of the Lorentz group; it then turns out that the derivation of the model requires picking a normal vector to every tetrahedron, and the EPRL prescription is to take these vectors to be time-like. Such a choice of ``time-like gauge'' reduces the symmetry group to an $\text{SU}(2)$ subgroup, thereby obtaining structures similar to the spin-networks of LQG (in this context called projective spin-networks) \cite{Livine:2002ak,Dupuis:2010jn}. The restriction to $\text{SU}(2)$ has moreover shown to have the added benefit of allowing simplifications of the amplitude formula by making use of its well-understood representation theory \cite{Speziale:2016axj}, which opened the door to numerical analyses of the model and its amplitudes in detail \cite{Dona:2019dkf,Dona:2020tvv}.

The necessity of specifying a ``causal character'' for each tetrahedron in the EPRL approach (by selecting a normal vector) poses an interesting possibility. Rather than insisting on assuming all tetrahedra to be space-like, one could rather allow for every possible assignment of causal character. Indeed, choosing a time-like normal for each tetrahedron is a strong restriction: the entire 4-dimensional geometry encoded in the spin-foam is then constituted entirely by space-like building blocks. It is at first sight not clear why more general triangulations of Lorentzian space-times should not appear in the full gravitational path integral. The freedom to independently assign causal characters to tetrahedra and triangles instead suggests that the model could accommodate different causal relations at the quantum level, a possibility that would be welcomed in relating spin-foams to theories such as causal dynamical triangulations \cite{Ambjorn:2012jv,Loll:2019rdj} and causal set theory \cite{Surya:2019ndm}.

In \cite{Conrady2010,Conrady2010a} Conrady and Hnybida (CH) extended the Lorentzian EPRL model to allow for time-like tetrahedra with space- and time-like triangles. In the CH extension, the normal vector is allowed to be space-like, in which case the symmetry group of a tetrahedron is reduced to $\text{SU}(1,1)$ (the double cover of the Lorentz group in Minkowski 3-dimensional space). This prescription considerably complicates the model, since, in addition to the fact that $\text{SU}(1,1)$ is non-compact, its representation theory is more intricate. The causal characters of the triangles in a time-like tetrahedron (i.e. whether a given triangle is space- or time-like) then dictate which of the unitary irreducible representations of $\text{SU}(1,1)$ are assigned to the dual complex.  All possible types of interfaces are allowed: a space-like triangle can now be shared between two space-like tetrahedra, one space-like and one time-like one or two time-like ones. 

Several years after the inception of the Lorentzian EPRL-CH model, the asymptotic expansion of its vertex amplitudes was studied, first for the case of space-like triangles \cite{Kaminski2018a} and then for the case of time-like ones \cite{Liu2019}. The asymptotic expansion has shown to be an ideal tool for understanding the dominant contributions to the path integral when all representation labels are large and the model is expected to behave semi-classically. To this end, the amplitude is expressed in a coherent state basis \cite{Livine:2007vk} as a multidimensional and highly oscillatory integral, which is approximated using a generalised stationary phase approximation. Generically, for different spin-foam models, it was found that the models exhibit critical and stationary points corresponding to piece-wise linear geometries \cite{Barrett:1998gs,Baez:2002rx,Conrady:2008mk,Barrett2009a} weighted by the Regge action \cite{Regge:1961px}. The same was observed particularly for the Lorentzian EPRL model \cite{Barrett2010,Barrett2011,Dona:2020yao}. Such Regge-like critical points were also found for CH extension for both space-like \cite{Kaminski2018a} and time-like triangles \cite{Liu2019}. We note in passing that these results play important role for the definition of the so-called effective spin-foam models \cite{Asante:2020qpa,Asante:2020iwm,Asante:2021zzh}, which use the semi-classical vertex amplitude as a starting point to efficiently explore the dynamics of large triangulations.

\vspace{2ex}

Our work revisits the asymptotics derived in \cite{Kaminski2018a,Liu2019}, our main interest being the analysis of time-like triangles. The coherent state prescription for such triangles proposed in \cite{Conrady2010} leads to an intricate expression for representation matrix elements that cannot be straightforwardly simplified. In \cite{Liu2019} these were asymptotically expanded for large representations in an involved calculation, which allowed the authors to continue the study of the vertex amplitude. Here we use a different parametrization in terms of generalised eigenstates with complex eigenvalues, from which we obtain exact closed formulae for the coherent states, and ultimately for the asymptotic action, removing the need for such an approximation. We confirm most of the results of \cite{Liu2019}, in particular the existence of Regge-like critical points and the presence of branch cuts in the integrand. Unfortunately, we also find potentially many more critical points that do not correspond to the desired Regge-like geometries and which are not dynamically suppressed unless one makes further assumptions. We thus conclude that the model is not constrained enough, a claim which strongly resonates with 1) a conjecture in \cite{Kaminski2018a} that simplicity might fail for time-like triangles/polygons and 2) the observation by Conrady \cite{Conrady2010a} that the master constraint for such polygons is not classically equivalent to the simplicity constraints. We present furthermore a proof of Minkowski's theorem for 3-dimensional polyhedra in Minkowski space, as well as a result on the rigidity of such polyhedra, allowing us to extend previous results of \cite{Kaminski2018a,Liu2019} by generalising to spin-foams built from arbitrary polytopes rather than just simplices. Finally, we show it is possible to relax an essential assumption of \cite{Liu2019} demanding tetrahedra to contain both space- and time-like faces, which would otherwise complicate an eventual contact with other approaches to quantum gravity, e.g. causal dynamical triangulations.

The text is organised as follows: in section \ref{section:general_form} we briefly outline the general construction of the Lorentzian EPRL spin-foam model and introduce the Conrady-Hnybida extension. Section \ref{sec:asympotit_expansion} describes the setup of the asymptotic expansion as well as the derivation of the critical point equations for heterochronal and parachronal interfaces (terminology which we define below). Solutions to these equations are derived in section \ref{sec:crit_solutions} for both types of interfaces, both in vectorial and bivectorial form. In section \ref{sec:induced_geometry} the geometry induced by geometric (Regge-like) critical points is constructed. We conclude in section \ref{sec:Discussion} with a thorough discussion of the results and an outlook for future research. Technical details for a more structured understanding of the presented calculations are compiled in the appendix: appendix \ref{appendix:generalized} describes generalised eigenstates of the non-compact generators of $\text{SU}(1,1)$ and their resolution of the identity. Appendix \ref{appendix:principal_series} summarises the unitary irreducible representations of $\text{SL}(2,\mathbb{C})$ and appendix \ref{appendix:geometrygroups} briefly recalls the geometry encoded in $\text{SU}(2)$ and $\text{SU}(1,1)$. The extension of Minkowski's theorem for convex polyhedra to Lorentzian signature can be found in appendix \ref{appendix:geometry}.

\vspace{2ex}

We would like to remark that an earlier version of this paper contained the wrong claim that the so-called Cosine Problem \cite{Barrett2010, Kaminski2018a, Liu2019} would be absent for spin-foams containing both space- and time-like polyhedra. This claim has been redacted thanks to input from Hongguang Liu \cite{Hongguang_git}, and a discussion of the Cosine Problem in the CH framework has been included.

\section{The extended EPRL spin-foam model} 
\label{section:general_form}

The spin-foam formulation we consider is the Lorentzian EPRL model proposed in \cite{Engle2008}, together with a later extension by Conrady and Hnybida \cite{Conrady2010, Conrady2010a}. The foundational idea of the model is to construct a quantum theory for gravity by first defining a path-integral for a related but simpler theory, usually known as $BF$-theory in the literature \cite{Baez2000, Cattaneo1995}; this is a purely topological theory, defined over an $\text{SL}(2,\mathbb{C})$ principal bundle $P \rightarrow M$ with a local connection 1-form $A$ and action $S_{BF}[B,A]=\int_M  \text{Tr } B\wedge F(A)$, where $B$ is a section of some associated vector bundle and $F$ denotes the curvature with respect to the connection. Given that, at the level of the classical theories, the simplicity constraints impose $B=\star(e\wedge e)$ and bring $BF$ into the tetrad formulation of general relativity \cite{Baez2000, Freidel1999}, one may then proceed by prescribing a reasonable form of these constraints to be applied after quantisation of the theory. Allowing for a formal integration of the $B$ field, we may write the Lorentzian-signature path integral over a four-dimensional manifold $M$ with no boundary as
\begin{equation}
Z_{BF}(M)=\int\mathcal{D} A \; \delta(F[A])\,.
\end{equation}
A rigorous meaning for this expression can be obtained by passing from the base manifold to a given choice of a 4-dimensional cellular-decomposition $\Delta$ (see \cite{Oeckl_2005} for a thorough description of the mathematical details), inducing by Poincaré duality an associated 2-complex $\Delta^*$ on which the connection and curvature forms can be integrated\footnote{Polygons, polyhedra and polytopes are taken to be 2-, 3- and 4-dimensional, respectively.}
\begin{equation}
    \begin{gathered}
        \Delta \rightarrow \Delta^* \\
        4\text{-cell (polytope)} \mapsto 0\text{-cell (vertex)} \\
        3\text{-cell (polyhedron)} \mapsto 1\text{-cell (link)} \\
        2\text{-cell (polygon)} \mapsto 2\text{-cell (plaquette)}\,.
    \end{gathered}
\end{equation}
An amplitude can thus be formulated on this dual complex as
\begin{equation}
   Z_{BF}(\Delta^*) =\int_{\text{SL}(2,\mathbb{C})}\prod_{l\in \mathcal{L}} \diff g_l  \prod_{p\in \mathcal{P}} \delta \left(\prod^\rightarrow_{l \in \partial p } g_p\right)\,,
\end{equation}
where $\mathcal{L}$ denotes the subset of $\Delta^*$ of all links (1-dimensional cells) and $\mathcal{P}$ the subset of all plaquettes (2-dimensional cells)\footnote{We reserve the usual terminology of "edges" and "faces" for the cells of the initial decomposition $\Delta$.}. The arrow over the product emphasises that it is ordered. Thought of as a generalised function on the Lie group manifold, the Dirac delta may be expanded in terms of a trace of unitary irreducible representations of the special linear group. Further manipulation then takes us to a more familiar form for a spin-foam amplitude
\begin{equation}
\label{eq:pre_amp}
    %Z_{BF}(\Delta^*)=\SumInt_{\{\chi\} \rightarrow \mathcal{P}} \left[\prod_{p\in \mathcal{P}} (n^2+\rho^2)_p \right] \text{Tr}_{p\in \mathcal{P}} \left[\prod_{l \in \mathcal{L}} \left(\int \diff g_l \prod_{p\, |\, l \in  \partial p} D^{\chi}_p(g_l) \right) \right]\,,
    Z_{BF}(\Delta^*)=\sum_{\{n\} \rightarrow \mathcal{P}} \int_{\{\rho\}\rightarrow \mathcal{P}}\diff \rho  \left[\prod_{p\in \mathcal{P}} (n^2+\rho^2)_p \right] \text{Tr}_{p\in \mathcal{P}} \left[\prod_{l \in \mathcal{L}} \left(\int \diff g_l \prod_{p\, |\, l \in  \partial p} D^{\chi}_p(g_l) \right) \right]\,,
\end{equation}
such that $\chi=(n,\rho)$ labels the unitary and irreducible representations $\mathcal{D}_\chi$ of $\text{SL}(2,\mathbb{C})$ (see appendix \ref{appendix:principal_series}) and $\{\chi\} \rightarrow \mathcal{P}$ denotes an assignment of representations to every plaquette\footnote{Note that, while we use $(n^2+\rho^2)$ for the product over plaquettes in accordance with the original EPRL proposal \cite{Engle:2007wy}, other choices of face amplitude have been proposed in the literature. In \cite{Bianchi:2010fj} it has been argued that one should rather use the dimension of the representations of the relevant subgroup of $\text{SL}(2,\mathbb{C})$.}. The subscript in $\text{Tr}_{p\in \mathcal{P}}$ signifies that one carries out a trace for every product of representation functions $D^\chi(g)$ associated to the same plaquette $p$.

A partition function for gravity can now be derived by first discretizing and then imposing the aforementioned simplicity constraints. We shall not review the construction of the constraints, but simply state the results found in the literature \cite{Engle2008, Conrady2010}. It turns out that the constraints dictate a restriction of the unitary irreducible representations one sums over in \eqref{eq:pre_amp} to particular subsets, depending on the causal character the polyhedra and polygons in the cellular decomposition are assumed to have:
\begin{enumerate}
\item If a polygon in $\Delta$ associated to a plaquette $p$ is time-like, then we restrict to $n_p=-\gamma \rho_p$;
\begin{itemize}
    \item Any two polyhedra associated to links $l,l'$ sharing the polygon are necessarily time-like, and both $D^{\chi}_p(g_l)$ and $D^{\chi}_p(g_{l'})$ are expanded in the continuous series of the pseudo-basis of $\mathcal{D}_\chi$ obtained from eigenstates of a non-compact generator of the $\text{SU}(1,1)$ subgroup (see appendix \ref{appendix:generalized}). We then have $\rho_p=-2 s_p$, and we term these interfaces \textit{parachronal}.
\end{itemize}
\item If a polygon in $\Delta$ associated to a plaquette $p$ is instead space-like, we restrict to $\rho_p=\gamma n_p$;
\begin{itemize}
    \item If the polygon is shared between two space-like polyhedra associated to links $l,l'$, both $D^{\chi}_p(g_l)$ and $D^{\chi}_p(g_{l'})$ are expanded in the $\text{SU}(2)$ canonical basis of $\mathcal{D}_\chi$. We then have $n_p=2j_p$, and we call these interfaces \textit{achronal}.
    \item If the polygon is shared between two time-like polyhedra associated to links $l,l'$, both $D^{\chi}_p(g_l)$ and $D^{\chi}_p(g_{l'})$ are expanded in the discrete series of the $\text{SU}(1,1)$ pseudo-basis of $\mathcal{D}_\chi$. We then have $n_p=2k_p$, and we call these interfaces \textit{orthochronal}.
    \item If the polygon is shared between a space-like polyhedron associated to $l$, and a time-like one associated to $l'$, then $D^{\chi}_p(g_l)$ is expanded in the $\text{SU}(2)$ canonical basis, while $D^{\chi}_p(g_{l'})$ is expanded in the discrete series of the $\text{SU}(1,1)$ pseudo-basis of $\mathcal{D}_\chi$. We still have $n_p=2k_p=2j_p$, and we call these interfaces \textit{heterochronal}.
\end{itemize}
\end{enumerate}
The proportionality factor $\gamma$ above is commonly called the Immirzi parameter. The spins $j$, $k$ and $s$ label irreducible unitary representations of the $\text{SU}(2)$ and $\text{SU}(1,1)$ subgroups in the representation spaces $\mathcal{D}^j$, $\mathcal{D}^{k,\alpha}$ and $\mathcal{C}^s_\epsilon$ respectively; we refer the reader once more to appendices \ref{appendix:generalized} and \ref{appendix:principal_series} for a review of the necessary representation theory, as well as for the respective notation we shall use throughout this paper.

It is helpful, for clarity of exposition, to consider the above expression for the simpler case when the cellular decomposition of $M$ is a triangulation by 4-simplices. A single 4-simplex has the $\Delta^*$ combinatorics
\begin{equation}
     \includegraphics[valign=c,scale=0.6]{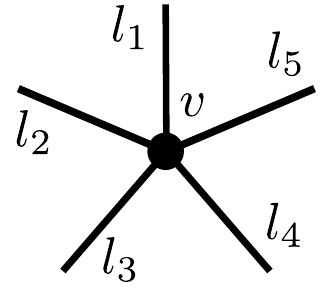}\,,
\end{equation}
there being five links $l_{1,...,5}$ associated to the five boundary tetrahedra in a 4-simplex. In this circumstance the proposed amplitude for gravity reads
\begin{equation}
\label{eq:BF_amp}
    %Z_{\text{grav.}}(\Delta^*_{\text{triang.}})={\SumInt_{\{\chi\} \rightarrow \mathcal{P}}}{\big|_{\text{simpl.}}} \left[\prod_{p\in \mathcal{P}} (n^2+\rho^2)_p \right] \left[\prod_{v \in \mathcal{V}} \includegraphics[valign=c,scale=0.5]{vertex.pdf} \right] \,,
    Z_{\text{grav.}}(\Delta^*_{\text{triang.}})={\left(\sum_{\{n\} \rightarrow \mathcal{P}} \int_{\{\rho\}\rightarrow \mathcal{P}}\diff \rho\right)}{\bigg|_{\text{simpl.}}} \left[\prod_{p\in \mathcal{P}} (n^2+\rho^2)_p \right] \left[\prod_{v \in \mathcal{V}} \includegraphics[valign=c,scale=0.5]{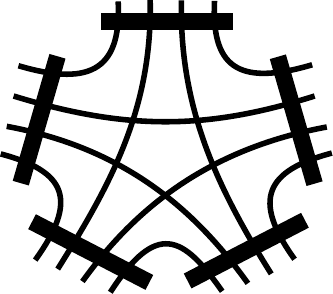} \right] \,,
\end{equation}
where the summation is now restricted to representations that fulfil the simplicity constraints. We use a standard graphical notation: lines (we have omitted their orientations) represent products of representation functions, while boxes represent group integrations. The product over the set of vertices $\mathcal{V}\subset \Delta^*$ is meant as a tilling, following the combinatorics of the underlying triangulation, from where one gets traces as loops of lines. Notice that the structure of each diagram at a vertex of $\Delta^*$ follows the structure of a 4-simplex of $\Delta$, having 5 tetrahedra at the boundary associated to the group integrations, each of which having 4 triangles associated to the lines. Finally, the group integrations in the previous equation can explicitly be carried out, resulting in the usual amplitude formula
\begin{equation}
    Z_{\text{grav.}}(\Delta^*_{\text{triang.}})={\left(\sum_{\{n\} \rightarrow \mathcal{P}} \int_{\{\rho\}\rightarrow \mathcal{P}}\diff \rho\right)}{\bigg|_{\text{simpl.}}} \sum_{\{\iota\} \rightarrow \mathcal{L}}\left[\prod_{p\in \mathcal{P}} (n^2+\rho^2)_p \right] \left[\prod_{v \in \mathcal{V}} \includegraphics[valign=c,scale=0.6]{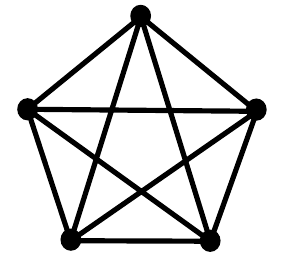} \right]\,,
    %Z_{\text{grav.}}(\Delta^*_{\text{triang.}})={\SumInt_{\{\chi\} \rightarrow \mathcal{P}}}{\big|_{\text{simpl.}}} \sum_{\{\iota\} \rightarrow \mathcal{L}}\left[\prod_{p\in \mathcal{P}} (n^2+\rho^2)_p \right] \left[\prod_{v \in \mathcal{V}} \includegraphics[valign=c,scale=0.6]{15j.pdf} \right]\,,
\end{equation}
where now there is a sum over assignments of intertwiner labels $\iota$ to every link $l\subset \mathcal{L}$. The pentagrammic symbol denotes a contraction of intertwiners at every vertex of the symbol, commonly called the vertex amplitude in the spin-foam literature.%i.e. the $10j$ symbol common in the spin-foam literature.

\subsection{Parameterization via coherent states}

An alternative parameterization of the model, which has shown to be useful in asymptotic analyses, can be given by transitioning from the orthonormal canonical and pseudo- bases to a coherent state basis. By coherent states we mean those which can be obtained by acting with a representation of a group element on a certain choice of reference state \cite{Perelomov_1986}. Letting again the spins $s$, $k$ and $j$ label representations of the continuous and discrete series of $\text{SU(1,1)}$ and representations of $\text{SU}(2)$, respectively, we have the completeness relations
\begin{equation}
\label{eq:comp_k1}
    \mathbbm{1}_{s,\epsilon}= \mu_\epsilon(s)\, \int_{\text{SU}(1,1)}\diff g \; D^{s,\epsilon}(g)\ket{j,\lambda,\sigma}\bra{j,\overline{\lambda},\sigma}D^{s,\epsilon\, \dagger}(g)\,,
\end{equation}
\begin{equation}
    \mathbbm{1}_{k,\alpha}=(2k-1)\, \int_{\text{SU}(1,1)}\diff g \; D^{k,\alpha}(g)\ket{k,m}\bra{k,m}D^{k,\alpha \, \dagger}(g)\,,
\end{equation}
\begin{equation}
    \mathbbm{1}_j= (2j+1)\,\int_{\text{SU}(2)}\diff g \; D^j(g)\ket{j,m}\bra{j,m}D^{j \, \dagger}(g)\,.
\end{equation}
When expressing the spin-foam amplitude in these bases, different choices of reference states will a priori lead to different parameterizations of the model, which nevertheless should agree at the level of the complete amplitude. Naturally, some reference states may be more useful than others in the analysis of the amplitude's asymptotic behaviour. A reasonable choice of reference states can be made by following the algorithm proposed in \cite{Conrady2010}: one looks for ``semi-classical'' states by minimising the variance of the relevant Casimir operator. For $\text{SU}(2)$ unitary irreducible representations and the discrete series of $\text{SU}(1,1)$, one finds in this manner the states
\begin{equation}
\ket{g, j}=D^j(g)\ket{j,j}\,, \quad \ket{g, \alpha k}=D^{k,\alpha}(g)\ket{k,\alpha k}\,,
\end{equation}
which we will use in the following. Regarding the continuous series, we shall propose a different choice than the one made in \cite{Conrady2010}. As we show in appendix \ref{appendix:generalized}, the states $\ket{j,\lambda,\sigma}$ admit complex eigenvalues $\lambda$. Using the notation $F_i F^i$ for the Casimir of $\text{SU}(1,1)$ obtained from the generators $F^i$, defined in the appendix just mentioned, we demand of the quantity
\begin{align}
    \braket{\Delta \,  F_i F^i}:&= \braket{F_i F^i}-\overline{\braket{F_i}}\braket{F^i} \nonumber \\
    &=-s^2-\frac{1}{4}+|\lambda|^2\,,
\end{align}
to vanish at the states of interest. Our choice is to take $\lambda=-\frac{i}{2}-s=ij$ as the preferred state among the circle of possible solutions. It will become apparent that this choice considerably simplifies the analysis\footnote{Compare this with the choice made in \cite{Conrady2010}, where the authors take $\lambda=\sqrt{s^2 +\frac{1}{4}}$.}. The coherent states appearing in the completeness relation \eqref{eq:comp_k1} will therefore be taken to be
\begin{equation}
    \ket{g,ij,1}=D^{s,\epsilon}(g)\ket{j,ij,1}\,, \quad \ket{g,\overline{ij},1}=D^{s,\epsilon}(g)\ket{j,\overline{ij},1}\,,
\end{equation}
where we picked $\sigma=1$. Note the structural similarity to the other coherent states above.

We are now able to rewrite the amplitude of equation \eqref{eq:pre_amp} in terms of the chosen coherent state bases. The trace portion of that equation can be expanded in factors associated to each vertex of $\Delta^*$, and this defines the usual vertex amplitude
\begin{equation}
\label{eq:vertex_amp}
A_v=\int_{\text{SL}(2,\mathbb{C})} \prod_{a=1}^n \diff g_a \delta(g_n) \prod_{b<a} \braket{D^\chi(g_a) \Psi_{ab}, D^\chi(g_b) \Psi_{ba}}\,,
\end{equation}
where we denote polyhedra in $\Delta$ by $a,b$ and their polygonal interfaces by an ordered pair $ab$. Here $n$ is the number of polyhedra contained in the polytope associated to $v\in\mathcal{V} \subset \Delta^*$, and $\Psi_{ab}$ denotes the coherent state associated to the polygon $ab$. The inner product $\braket{\cdot,\cdot}$ refers to the inner product in the $\text{SL}(2,\mathbb{C})$ representation Hilbert space $\mathcal{D}_\chi$. Finally, it is conventional to include a gauge-fixing term $\delta(g_n)$ as a regulator of the integral, necessary due to the non-compactness of the special linear group.

\subsection{Explicit expressions for coherent states}
\label{section:coherent}

We will now obtain an explicit form for the coherent states in the Hilbert space $\mathcal{D}_\chi$ of the principal series of unitary irreducible representations of $\text{SL}(2,\mathbb{C})$. As discussed above, these are constructed with respect to states induced from $\text{SU}(2)$ and $\text{SU}(1,1)$ representations (as in appendix \ref{appendix:principal_series}), with concrete choices of $\chi=(n,\rho)$ determined by the model \cite{Engle2008, Conrady2010, Conrady2010a}. So as to have a more unified notation, we define first two inner products in $\mathbb{C}^2$,
\begin{equation}
\braket{u, \, v}_\mathbbm{1} =u^\dagger \mathbbm{1}  \, v\,, \quad \quad \quad \braket{u, \, v}_{\sigma_3}=u^\dagger {\sigma_3} \,  v\,, \; {\sigma_3}=\text{diag}(1,-1)\,,
\end{equation}
which are invariant under the natural actions of $\text{SU}(2)$ and $\text{SU}(1,1)$, respectively. Moreover, we write the canonical basis of $\mathbb{C}^2$ as
\begin{equation}
\ket{+}:=\begin{pmatrix} 1 \\ 0 \end{pmatrix}\,, \quad \ket{-}:=\begin{pmatrix} 0 \\ 1 \end{pmatrix}\,,
\end{equation}
and further define the states
 \begin{equation}
\ket{l^+}:=\frac{1}{\sqrt{2}}\left(\ket{+}+\ket{-}\right)\,, \quad  \ket{l^-}:=\frac{1}{\sqrt{2}}\left(\ket{+}-\ket{-}\right)\,, \quad \ket{z}:=\begin{pmatrix} z_1 \\ z_2 \end{pmatrix}\,.
\end{equation}
We first construct the reference states in $\mathcal{D}_\chi$ with respect to which all coherent states will be defined.  Referring to equation \eqref{eq:func2} of the appendix and the well-known formula for $\text{SU}(2)$ matrix elements \cite{MartinDussaud2019}, the $m=j$ maximal-weight functions in $\mathcal{D}_\chi$ induced from this subgroup have the form
\begin{equation}
F_{j,j}^\chi(\mathbf{z})=\sqrt{2j+1} \sqrt{\frac{(2j)!}{(j+\frac{n}{2})!(j-\frac{n}{2})!}} \braket{z, \, z}_\mathbbm{1} ^{i \frac{\rho}{2}-1-j} \braket{z,\, +}_\mathbbm{1} ^{j+n/2} \braket{-,\, z}_\mathbbm{1} ^{j-n/2}\,.
\end{equation}
For the $\text{SU}(1,1)$ discrete series and the choice $m=\tau k$, equation \eqref{eq:func11} and the expression for matrix elements derived in \cite{Bargmann1947} give
\begin{align}
F_{k,\tau k}^{\chi, \tau}(\mathbf{z})=\frac{\sqrt{2k-1} \tau^{k-\frac{n}{2}}}{(\frac{n}{2}-k)!}&\sqrt{\frac{\Gamma(\frac{n}{2}-k+1)\Gamma(\frac{n}{2}+k)}{\Gamma(2k)}}\, \Theta\left(\tau\braket{z, \, z}_{\sigma_3}\right) \cdot \nonumber\\
& \cdot \left(\tau \braket{z, \, z}_{\sigma_3}\right)^{i\frac{\rho}{2}-1+k} (\tau\braket{\tau,z}_{\sigma_3})^{-\frac{n}{2}-k} (-\tau\braket{z,-\tau}_{\sigma_3})^{\frac{n}{2}-k}\,, \quad \frac{n}{2}\geq k\,,
\end{align}
\begin{align}
F_{k,\tau k}^{\chi, \tau}(\mathbf{z})=\frac{\sqrt{2k-1} (-\tau)^{k-\frac{n}{2}}}{(k-\frac{n}{2})!}&\sqrt{\frac{\Gamma(2k)}{\Gamma(\frac{n}{2}-k+1)\Gamma(\frac{n}{2}+k)}}\, \Theta\left(\tau\braket{z, \, z}_{\sigma_3}\right) \cdot \nonumber\\
& \cdot \left(\tau \braket{z, \, z}_{\sigma_3}\right)^{i\frac{\rho}{2}-1+\frac{n}{2}} (\tau\braket{\tau,z}_{\sigma_3})^{-\frac{n}{2}-k} (-\tau\braket{-\tau,z}_{\sigma_3})^{-\frac{n}{2}+k}\,, \quad \frac{n}{2}\leq k\,,
\end{align}
where $\ket{\tau}$ is to be understood as $\ket{+}, \, \ket{-}$ for $\tau=(+,-)$, respectively.
Regarding the continuous series, the discussion of sections \ref{appendix:generalized_2} and \ref{appendix:generalized_3} allows us to write, for both relevant states $\lambda=ij$ and $\lambda=\overline{ij}$,
\begin{equation}
     F^{\chi,\tau}_{s, \epsilon, ij,1}(\mathbf{z})=\sqrt{\mu_\epsilon(s)}\Theta(\tau\braket{z,z}_{\sigma_3})(\tau\braket{z,z}_{\sigma_3})^{i\rho/2-1-j}A_{\tau \frac{n}{2},1}^j \braket{z,l^-}_{\sigma_3}^{j+\frac{n}{2}}\braket{l^-,z}_{\sigma_3}^{j-\frac{n}{2}}\, ,
\end{equation}
\begin{align}
    F^{\chi,\tau}_{s, \epsilon, \overline{ij},1}(\mathbf{z})=\sqrt{\mu_\epsilon(s)}\Theta(\tau\braket{z,z}_{\sigma_3})&(\tau\braket{z,z}_{\sigma_3})^{i\rho/2-1-j} A_{\tau \frac{n}{2},0}^j  \braket{z,l^-}_{\sigma_3}^{j+\frac{n}{2}}\braket{l^-,z}_{\sigma_3}^{j-\frac{n}{2}} \cdot\nonumber \\
    & \cdot \frac{i}{2j+1}\left[\left(j+\frac{n}{2}\right)\frac{\braket{z,l^+}_{\sigma_3}}{\braket{z,l^-}_{\sigma_3}}-\left(j-\frac{n}{2}\right)\frac{\braket{l^+,z}_{\sigma_3}}{\braket{l^-,z}_{\sigma_3}} \right]\,.
\end{align}

Finally, to define the coherent states for each family of functions, it suffices to consider the group action \eqref{rep} on $\mathcal{D}_\chi$. A general coherent state will then be defined as $\Psi_{ab}=h_{ab} \triangleright F(\mathbf{z})$, where $h_{ab}$ is taken to be an element of the relevant subgroup of $\text{SL}(2,\mathbb{C})$. We also write $(n,\rho)$ in terms of the spins, as outlined in section \ref{section:general_form}. The states then read
\begin{equation}
\Psi^j_{ab}(\mathbf{z})=\sqrt{2j+1} \braket{z, \, z}_\mathbbm{1} ^{j(i\gamma-1)-1} \braket{z, \, +_{ab}}_\mathbbm{1} ^{2j}\,,
\end{equation}
\begin{equation}
\Psi^{k,\tau}_{ab}(\mathbf{z})=\sqrt{2k-1}\, \Theta\left(\tau\braket{z, \, z}_{\sigma_3}\right)\left(\tau \braket{z, \, z}_{\sigma_3}\right)^{k(i\gamma+1)-1}\left(\tau\braket{\tau_{ab}, \, z}_{\sigma_3} \right)^{-2k}\,,
\end{equation}
\begin{align}
\Psi^{s, \epsilon, \tau}_{ab}(\mathbf{z})=\sqrt{\mu_\epsilon(s)}\Theta\left(\tau\braket{z,z}_{\sigma_3}\right)\left(\tau\braket{z,z}_{\sigma_3}\right)^{-2is-\frac{1}{2}}A_{\tau \frac{n}{2},1}^j \braket{z,l^-_{ab}}_{\sigma_3}^{(i+\gamma)s-\frac{1}{2}}\braket{l^-_{ab},z}_{\sigma_3}^{(i-\gamma)s-\frac{1}{2}}\,,
\end{align}
\begin{align}
\tilde{\Psi}^{s, \epsilon, \tau}_{ab}(\mathbf{z})=\sqrt{\mu_\epsilon(s)}&\Theta\left(\tau\braket{z,z}_{\sigma_3}\right)\left(\tau\braket{z,z}_{\sigma_3}\right)^{-2is-\frac{1}{2}}A_{\tau \frac{n}{2},0}^j \braket{z,l^-_{ab}}_{\sigma_3}^{(i+\gamma)s-\frac{1}{2}}\braket{l^-_{ab},z}_{\sigma_3}^{(i-\gamma)s-\frac{1}{2}} \cdot \nonumber \\
& \cdot \frac{1}{2s}\left[\left((i+\gamma)s-\frac{1}{2}\right)\frac{\braket{z,l^+_{ab}}_{\sigma_3}}{\braket{z,l^-_{ab}}_{\sigma_3}}-\left((i-\gamma)s-\frac{1}{2}\right)\frac{\braket{l^+_{ab},z}_{\sigma_3}}{\braket{l^-_{ab},z}_{\sigma_3}} \right]\,,
\end{align}
where the general notation $\ket{\,\cdot_{ab}}$ stands for a rotated state $\ket{{h_{ab}^T}^{-1} \cdot\,}$. We differentiate between the $\text{SU}(1,1)$-continuous states at $\lambda=ij$ and its dual at $\lambda=\overline{ij}$ by the symbols $\Psi^{s, \epsilon, \tau}_{ab}$ and $\tilde{\Psi}^{s, \epsilon, \tau}_{ab}$, respectively. These are the coherent states with which we will obtain the spin-foam amplitudes.

\section{Asymptotic expansion and analysis} \label{sec:asympotit_expansion}

We shall follow the general procedure applied in \cite{Barrett2009a,Barrett2010,Kaminski2018a}. Consider the 2-complex $\Delta^*$ dual to a 4-dimensional polytope containing $n$ polyhedral 3-cells. Our goal will be to accommodate the expression \eqref{eq:vertex_amp} for the amplitude of $\Delta^*$ in the generically useful form for a stationary phase approximation
\begin{equation}
\int \diff x  \,f(x) e^{\Lambda S(x)} \,, \quad \Lambda \rightarrow \infty\,,
\end{equation}
where $\Lambda$ stands for a uniform scaling of the spins that characterise the boundary data.
To do so, we shall focus on the inner product in equation \eqref{eq:vertex_amp}. Defining $\Omega_{ab}=\overline{g_{a}\triangleright \Psi_{ab}}\cdot g_{b}\triangleright \Psi_{ba}$ for some special linear matrices $g_{a,b}$, we explicitly write the inner product as an integral,
\begin{equation}
\label{eq:asympt}
A_v=\int_{\text{SL}(2,\mathbb{C})}\prod_{a=1}^n \diff g_a \delta(g_n) \prod_{a<b} \int_{\mathbb{C}P} \omega(\mathbf{z}_{ab})  \Omega_{ab}(\mathbf{z}_{ab},g_a,g_b) \,,
\end{equation}
where $\omega(\mathbf{z}_{ab})$ is the integration measure defined in appendix \ref{appendix:principal_series}, and subsequently bring $\Omega_{ab}$ into the generic form of the exponential of an ``action", $\Omega_{ab}=f_{ab} \, e^{ \Lambda S_{ab}}$. We are then interested in the critical points of $S=\sum_{a<b} S^\nu_{ab}$, and these are characterised firstly by a reality condition,
\begin{equation}
 \Re\,  S_{ab}(\mathbf{z}_{ab},g_a,g_b)=0, \,\forall \,a,b\,,
\end{equation}
and secondly by the critical point conditions
\begin{equation}
\label{eq:critic}
\begin{gathered}
\begin{cases}
\delta_{\mathbf{z}_{ab}} S_{ab}(\mathbf{z}_{ab},g_a,g_b)=0, \,\forall \,a,b\,, \\
\sum_{b>a} \delta_{g_a} S_{ab}(\mathbf{z}_{ab},g_a,g_b)=0, \,\forall \,a\,.
\end{cases}
\end{gathered}
\end{equation}
Note that one need only vary the action with respect to the holomorphic spinor and group variables $\mathbf{z}_{ab}$ and $g_a$, since the action is constrained to be purely imaginary \cite{Kaminski2018a}.

Because we allow for a generic causal structure (excluding light-like polygons and polyhedra), we have to consider every possible type of interface between two polyhedra, as per the prescription in section \ref{section:general_form}.  The calculation for achronal interfaces has already been carried out in great detail in the literature \cite{Barrett2010}, so we will refrain from repeating it here. The case of an orthochronal polygon has also been thoroughly discussed in \cite{Kaminski2018a}. While the remaining two possibilities were analysed separately in \cite{Kaminski2018a} and \cite{Liu2019}, respectively, we found that revisiting these cases proved useful in clarifying some previous assumptions and further understanding the structure of the model. There are hence two cases which we will mostly focus on below: the case of a time-like interface between two time-like polyhedra, which we termed parachronal, and the one of a space-like interface between time- and space-like polyhedra, denoted heterochronal. We shall simply state the results for the remaining cases when necessary.

\subsection{Heterochronal interfaces}

We consider the product $\Omega_{ab}=\overline{g_{a}\triangleright \Psi^j_{ab}}\cdot g_{b}\triangleright \Psi^{k,\tau}_{ba}$ where $\Psi^j_{ab}$ is an $\text{SU}(2)$ state and $\Psi^{k,\tau}_{ba}$ is an $\text{SU}(1,1)$ one in the discrete series. The model requires that the spins agree, $j_{ab}=k_{ab}$, so we find
\begin{equation}
\Omega_{ab}= f_{ab}(\mathbf{z}_{ab})\left( p_{ab}(\mathbf{z}_{ab},g_a,g_b)\right)^{\Lambda j_{ab}}\,,
\end{equation}
having explicitly included the uniform scaling factor $\Lambda$, and where we define
\begin{equation}
\label{eq:mspace}
\begin{gathered}
f_{ab}(\mathbf{z}_{ab},g_a,g_b)=\frac{\sqrt{4j_{ab}^2-1}\; \Theta\left(\tau_{ba}\braket{g_b^T z_{ab}, \,g_b^T z_{ab}}_{\sigma_3}\right)}{\tau_{ba}\braket{g_a^T z_{ab}, \, g_a^Tz_{ab}}_\mathbbm{1} \braket{g_b^Tz_{ab}, \, g_b^Tz_{ab}}_{\sigma_3}}\,, \\
p_{ab}^{j_{ab}}(\mathbf{z}_{ab},g_a,g_b)=\left( \frac{ \tau_{ba} \braket{g_b^Tz_{ab}, \, g_b^Tz_{ab}}_{\sigma_3}}{\braket{g_a^T z_{ab}, \, g_a^Tz_{ab}}_\mathbbm{1} }\right)^{(i\gamma+1)j_{ab}} \left(\frac{\braket{ +_{ab} , \, g_a^Tz_{ab}}_\mathbbm{1} }{\tau_{ba}\braket{\tau_{ba}, \, g_b^T z_{ab}}_{\sigma_3} }  \right)^{2j_{ab}}\,.
\end{gathered}
\end{equation}
We can then write the action as the logarithm of $p_{ab}^{j_{ab}}$, i.e.,
\begin{equation}
\label{eq:actionspace}
S_{ab}=(i\gamma+1)j_{ab}\ln  \tau_{ba} \frac{\braket{g_b^Tz_{ab}, \, g_b^Tz_{ab}}_{\sigma_3}}{\braket{g_a^T z_{ab}, \, g_a^Tz_{ab}}_\mathbbm{1} } +2j_{ab}\ln\frac{\braket{+_{ab} , \, g_a^Tz_{ab}}_\mathbbm{1} }{\tau_{ba}\braket{\tau_{ba}, \, g_b^T z_{ab}}_{\sigma_3} } \,,
\end{equation}
and the reality condition $\Re S_{ab}=0$ implies
\begin{equation}
\label{realityspace}
\begin{gathered}
\ln \left| \frac{\braket{g_b^Tz_{ab}, \, g_b^Tz_{ab}}_{\sigma_3}}{\braket{g_a^T z_{ab}, \, g_a^Tz_{ab}}_\mathbbm{1} }\left(\frac{\braket{+_{ab} , \, g_a^Tz_{ab}}_\mathbbm{1} }{\tau_{ba}\braket{ \tau_{ba},g_b^T z_{ab}}_{\sigma_3} }\right)^2 \right|-\gamma\, \text{arg}\left(  \tau_{ba} \frac{\braket{g_b^Tz_{ab}, \, g_b^Tz_{ab}}_{\sigma_3}}{\braket{g_a^T z_{ab}, \, g_a^Tz_{ab}}_\mathbbm{1} }\right)=0  \\
\Leftrightarrow
\begin{cases}
\ket{g_a^T z_{ab}}=\lambda_{ab} \ket{+_{ab}} \\
\ket{g_b^T z_{ba}}=\lambda_{ba} \ket{\tau_{ba}}
\end{cases}, \; \lambda_{ab}, \lambda_{ba}\in \mathbb{C} \,.
\end{gathered}
\end{equation}
Generically, the variation of the action reads
\begin{align}
\delta S_{ab}=(i\gamma+1)j_{ab}&\left(\frac{\delta \braket{g_b^Tz_{ab}, \, g_b^Tz_{ab}}_{\sigma_3}}{\braket{g_b^Tz_{ab}, \, g_b^Tz_{ab}}_{\sigma_3}} - \frac{\delta \braket{g_a^T z_{ab}, \, g_a^Tz_{ab}}_\mathbbm{1} }{\braket{g_a^T z_{ab}, \, g_a^Tz_{ab}}_\mathbbm{1} } \right) + \nonumber \\
&+2j_{ab} \left(\frac{\delta \braket{+_{ab} , \, g_a^Tz_{ab}}_\mathbbm{1} }{\braket{ +_{ab} , \, g_a^Tz_{ab}}_\mathbbm{1} } - \frac{\delta \braket{\tau_{ba}, \, g_b^T z_{ab}}_{\sigma_3}}{\braket{\tau_{ba}, \, g_b^T z_{ab}}_{\sigma_3}}\right)\,,
\end{align}
and therefore we have for the variation with respect to $\mathbf{z}_{ab}$ that
\begin{equation}
\delta_{\mathbf{z}_{ab}} S_{ab}=\left( i\gamma-1 \right)\left(\frac{ \overline{\lambda}_{ba}}{|\lambda_{ba}|^2}\tau_{ba}\bra{\tau_{ba}} {\sigma_3} g_b^T  - \frac{ \overline{\lambda}_{ab}}{|\lambda_{ab}|^2} \bra{+_{ab}} g_a^T \right) \delta\mathbf{z}_{ab} \,,
\end{equation}
where we used the solution to the reality condition. Regarding the variation with respect to SL$(2,\mathbb{C})$ group elements, we will consider two cases: for $g_a$, we choose the expansion
\begin{equation}
g_a(\epsilon)=g _a e^{i \epsilon_l G^l}\,, \quad G=(\vec{J},\, i \vec{J}),\,
\end{equation}
using the $\vec{J}$ generators of $\text{SU}(2)$, while for $g_b$ we take
\begin{equation}
\label{eq:var11}
g_b(\epsilon)=g_be^{i \epsilon_l H^l}\,, \quad H=(\vec{F},\, i \vec{F})\,,
\end{equation}
using the $\vec{F}$ generators of $\text{SU}(1,1)$. The variations with respect to the subgroups will therefore read
\begin{equation}
\label{varg_sl}
\delta_{g_a} S_{ab}=i \epsilon_l(1-i\gamma)j_{ab}\braket{+_{ab} , \, G^l(+_{ab})}_\mathbbm{1} \,,
\end{equation}
\begin{equation}
\label{varg11}
\delta_{g_b} S_{ab}=-i\epsilon_l (1-i\gamma)j_{ab} \tau_{ba}\braket{\tau_{ba}, \,  H^l \tau_{ba}}_{\sigma_3}\,,
\end{equation}
after applying once more the reality condition.

\subsection{Parachronal interfaces}
We now take the boundary states $\tilde{\Psi}^{s,\epsilon,\tau_{ab}}_{ab}$ and $\Psi^{s, \epsilon, \tau_{ba}}_{ba}$ to be those induced from the $\text{SU}(1,1)$ continuous series. Recall that one has to use both the state $\Psi$ and its dual $\tilde{\Psi}$, since they are needed together in the completeness relation of equation \eqref{eq:comp_k1}. Now, rather than directly studying $\Omega_{ab}$ as defined above, we will simplify our analysis by making use of the intertwiner map $\mathcal{A}: \mathcal{D}_\chi \rightarrow \mathcal{D}_{-\chi}$ defined in \eqref{eq:intertwiner}. It is straightforward to check that $\mathcal{A}$ acts on the relevant states as $\mathcal{A} F^{\chi,\tau}_{s, \epsilon, \lambda,\sigma}(\mathbf{z})=\varphi^\tau_{s, \epsilon} F^{-\chi,\tau}_{s, \epsilon, \lambda,\sigma}(\mathbf{z})$, for some phase $\varphi^\tau_{s, \epsilon}$. Indeed, the defining property of the map implies that, at the level of the $\mathfrak{sl}(2,\mathbb{C})$ algebra, we must have for the $F^i$ generators
\begin{equation}
    \mathcal{A} D'^{-\chi}(F^i)=D'^{\chi}(F^i)\mathcal{A}\,,
\end{equation}
and thus also
\begin{equation}
    \mathcal{A} P^{-\chi}= P^\chi\mathcal{A}\,.
\end{equation}
Here $D'^\chi$ denotes the induced representation of the algebra, while $P^\chi$ is the representation of the parity operator \eqref{eq:K1} determined by the requirement $P^\chi D'^\chi(F^i) {P^\chi}^{-1}=D'^\chi(P F^i P^{-1})$. This is enough to establish that $\mathcal{A}$ maps eigenstates $F^{\chi,\tau}_{s, \epsilon, \lambda,\sigma}$ to eigenstates $F^{-\chi,\tau}_{s, \epsilon, \lambda,\sigma}$, up to a phase. That $\varphi^\tau_{s, \epsilon}$ does not depend on $\lambda$ and $\sigma$ can then be seen e.g. by considering the action of the intertwiner on the $F^+$ ladder operator \eqref{eq:ladder}; it will moreover turn out that the explicit form of these phases will not be needed.

Referring back to the inner product in $\mathcal{D}_\chi$, it holds that
\begin{align}
\braket{D^j(g_a) F^{\chi,\tau_{ab}}_{s, \epsilon,\overline{ij},1} ,\; D^{j}(g_b) F^{\chi,\tau_{ba}}_{s, \epsilon, ij,1}}&=\braket{\mathcal{A} \cdot \, D^j(g_a) F^{\chi,\tau_{ab}}_{s, \epsilon,\overline{ij},1} ,\; \mathcal{A} \cdot \, D^{j}(g_b) F^{\chi,\tau_{ab}}_{s, \epsilon,ij,1}} \nonumber \\
&=\overline{\varphi^{\tau_{ab}}_{s, \epsilon}}\varphi^{\tau_{ba}}_{s, \epsilon}  \braket{D^j(g_a)F^{-\chi,\tau_{ab}}_{s, \epsilon,\overline{ij},1}, \cdot D^{j}(g_b) F^{-\chi,\tau_{ab}}_{s, \epsilon,ij,1}}\,,
\end{align}
so one may equivalently study the object $\Omega_{ab}=\overline{g_a \triangleright \mathcal{A} \tilde{\Psi}^{s,\epsilon,\tau}_{ab}}\, \cdot g_b \triangleright \mathcal{A} \Psi^{s,\epsilon,\tau}_{ba}$. As before the spins $s_{ab}=s_{ba}$ must agree, and expanding
\begin{equation}
\label{eq:omega_p}
\Omega_{ab}= f_{ab}(\mathbf{z}_{ab})\left( p_{ab}(\mathbf{z}_{ab},g_a,g_b)\right)^{\Lambda s_{ab}}\,,
\end{equation}
we find
\begin{align}
\label{eq:pref_p}
    f_{ab}(\mathbf{z}_{ab},g_a,g_b)=&N^j_{\epsilon,\tau_{ab},\tau_{ba}} \Theta(\tau_{ab} \braket{g_a^T z, g_a^T z}_{\sigma_3}) \Theta(\tau_{ba} \braket{g_b^T z, g_b^T z}_{\sigma_3}) \cdot \nonumber \\
    &\cdot \left|\braket{l^-_{ab},g_a^T z}_{\sigma_3} \right|^{-1}\left|\braket{l^-_{ba},g_b^T z}_{\sigma_3} \right|^{-1} (\tau_{ab} \braket{g_a^T z, g_a^T z}_{\sigma_3})^{-\frac{1}{2}}(\tau_{ba} \braket{g_b^T z, g_b^T z}_{\sigma_3})^{-\frac{1}{2}} \cdot \nonumber \\
    &\cdot \left[\left(-\frac{1}{2}-(i+\gamma)s_{ab}\right)\frac{\braket{l^+_{ab},g_a^T z}_{\sigma_3}}{\braket{l^-_{ab},g_a^T z}_{\sigma_3}}-\left(-\frac{1}{2}-(i-\gamma)s_{ab}\right)\frac{\braket{g_a^Tz, l^+_{ab}}_{\sigma_3}}{\braket{g_a^Tz, l^-_{ab}}_{\sigma_3}} \right]\,,
\end{align}
\begin{equation}
    p_{ab}^{s_{ab}}(\mathbf{z}_{ab},g_a,g_b)=\left(\frac{\braket{l^-_{ba},g_b^T z}_{\sigma_3}}{\braket{l^-_{ab},g_a^T z}_{\sigma_3}}\right)^{(i+\gamma){s_{ab}}} + \left(\frac{\braket{g_b^Tz, l^-_{ba}}_{\sigma_3}}{\braket{g_a^Tz, l^-_{ab}}_{\sigma_3}}\right)^{(i-\gamma){s_{ab}}}\,,
\end{equation}
with a prefactor
\begin{equation}
\label{eq:prefactor}
    N^j_{\epsilon,\tau_{ab},\tau_{ba}}=\frac{\mu_{\epsilon}(s)}{2s} \overline{\varphi^{\tau_{ab}}_{s, \epsilon}} \varphi^{\tau_{ba}}_{s, \epsilon}  \overline{ A_{-\tau_{ab} \gamma s,0}^j}A_{-\tau_{ba} \gamma s,1}^j\,.
\end{equation}
The associated action, obtained from the logarithm of $p_{ab}^{s_{ab}}$, has thus the simple form\footnote{Much of the intricate structure of the parachronal amplitude is relegated to the measure factor $f_{ab}$.}
\begin{equation}
\label{eq:para_action}
    S_{ab}=(i+\gamma){s_{ab}}\ln \frac{\braket{l^-_{ba},g_b^T z}_{\sigma_3}}{\braket{l^-_{ab},g_a^T z}_{\sigma_3}} +(i-\gamma){s_{ab}} \ln \frac{\braket{g_b^Tz, l^-_{ba}}_{\sigma_3}}{\braket{g_a^Tz, l^-_{ab}}_{\sigma_3}}\,.
\end{equation}
Note that, unlike for achronal, orthochronal and heterochronal actions, the parachronal action is purely imaginary, in agreement with the approximation found in \cite{Liu2019}. The variations are now straightforwardly obtained
\begin{equation}
   \delta_{\mathbf{z}_{ab}} S_{ab}=\left( \frac{\bra{l^-_{ba}}\sigma_3 g_b^T}{\braket{l^-_{ba},g_b^T z}_{\sigma_3}}- \frac{\bra{l^-_{ab}}\sigma_3 g_a^T}{\braket{l^-_{ab},g_a^T z}_{\sigma_3}}\right) \delta \mathbf{z}_{ab} \,,
\end{equation}
\begin{equation}
\delta_{g_a}S_{ab}=-(i+\gamma){s_{ab}}\, i \epsilon_l \frac{\braket{l^-_{ab},H^l g_a^T z}_{\sigma_3}}{\braket{l^-_{ab},g_a^T z}_{\sigma_3}}\,,
\end{equation}
\begin{equation}
\delta_{g_b}S_{ab}=(i+\gamma){s_{ab}}\, i \epsilon_l \frac{\braket{l^-_{ba},H^l g_b^T z}_{\sigma_3}}{\braket{l^-_{ba},g_b^T z}_{\sigma_3}}\,,
\end{equation}
where we made use of the parameterization \eqref{eq:var11} for the group elements. We proceed by noting that $\{\ket{l^\pm_{ab}}\}$ constitutes a basis for $\mathbb{C}^2$, so that we may expand
\begin{equation}
\label{eq:expand}
    \ket{g_a^T z}=\alpha_{ab}\left(\ket{l^+_{ab}} + \beta_{ab} \ket{l^-_{ab}} \right)\,,
\end{equation}
and mutatis mutandis for $ \ket{g_b^T z}$. Under these expansions the previous equations take the form
\begin{equation}
   \delta_{\mathbf{z}_{ab}} S_{ab}=\left( \frac{\bra{l^-_{ba}}\sigma_3 g_b^T}{\alpha_{ba}}- \frac{\bra{l^-_{ab}}\sigma_3 g_a^T}{\alpha_{ab}}\right) \delta \mathbf{z}_{ab} \,,
\end{equation}
\begin{equation}
\label{eq:ga_p}
\delta_{g_a}S_{ab}=-(i+\gamma){s_{ab}}\, i \epsilon_l \left(\braket{l^-_{ab},H^l l^+_{ab}}_{\sigma_3}+\beta_{ab} \braket{l^-_{ab},H^l l^-_{ab}}_{\sigma_3}\right)\,,
\end{equation}
\begin{equation}
\label{eq:gb_p}
\delta_{g_b}S_{ab}=(i+\gamma){s_{ab}}\, i \epsilon_l \left(\braket{l^-_{ba},H^l l^+_{ba}}_{\sigma_3}+\beta_{ba} \braket{l^-_{ba},H^l l^-_{ab}}_{\sigma_3}\right)\,.
\end{equation}
For completeness we also expand equation \eqref{eq:pref_p} in terms of $\ket{l^\pm_{ab}}$, finding
\begin{align}
\label{eq:pref_p_exp}
    f_{ab}(\mathbf{z}_{ab},g_a,g_b)=&\frac{N^j_{\epsilon,\tau_{ab},\tau_{ba}}}{2} \Theta(\tau_{ab} \,\Re \beta_{ab}) \Theta(\tau_{ba} \,\Re \beta_{ba}) \left| \alpha_{ab} \alpha_{ba}\right|^{-2} \cdot \nonumber \\
    &\cdot \left(\tau_{ab} \tau_{ba} \,\Re \beta_{ab} \,\Re \beta_{ba} \right)^{-\frac{1}{2}} \left[(2s_{ab}-i)\,\Im\beta_{ab} - 2\gamma s_{ab} \,\Re \beta_{ab} \right].
\end{align}

\section{Vectorial formulation of the critical points} \label{sec:crit_solutions}

We are now in possession of the algebraic critical point equations with which we can study the dominant behavior of the spin-foam amplitude. These equations can, however, be brought into a more amenable form, which will later on be useful to assign a geometric meaning to them; this is the main subject of this section.
\subsection{Closure relations}
\label{section:closure}

Recall that part of the stationarity condition is obtained from the variation with respect to group elements, and for each polyhedron there is the constraint
\begin{equation}
\label{eq:groupvar}
    \sum_{b \neq a} \delta_{g_a} S_{ab}=0\,,\,\forall \,a\,,
\end{equation}
on which we will focus throughout this subsection.
Consider once more the general form of the spin-foam amplitude of equation \eqref{eq:vertex_amp}. For each choice of a polyhedron $a$ one has a product of $\Omega_{ab}$ functions, one for each polygon labeled by $ab$. Since there is one group integration for each polyhedron, the concrete form of the critical point equations obtained from \eqref{eq:groupvar} will depend on the causal character of each interface $ab$.
\begin{figure}
    \centering
    	{\includegraphics[valign=c,scale=0.8]{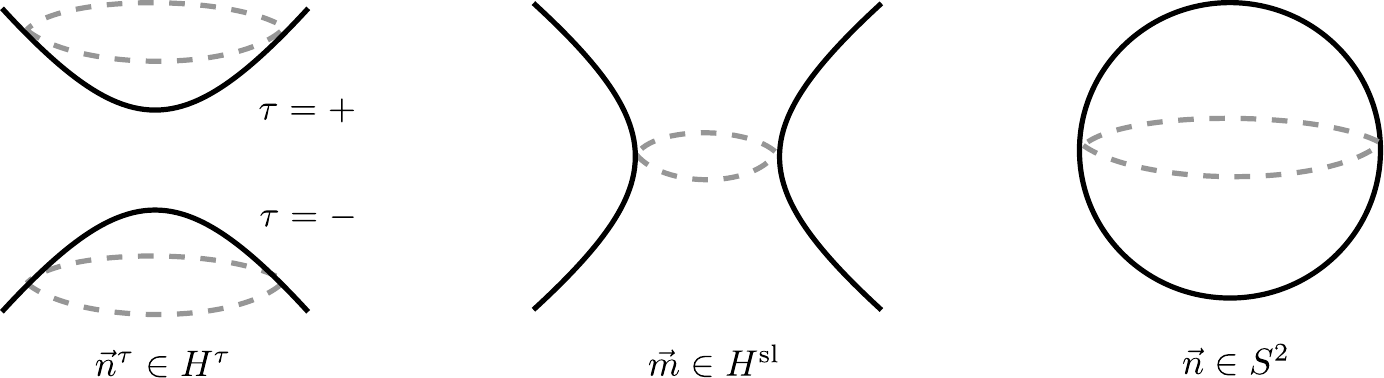}}
    	\caption{2-dimensional subspaces of $\mathbb{R}^{3,1}$ associated with coherent state boundary data. The sheets of the two-sheeted hyperboloid $H^\tau$ are labelled by $\tau$, as discussed in appendix \ref{appendix:geometrygroups}.}
\end{figure}

\subsubsection{Space-like polyhedra}

If we take $a$ to be a space-like polyhedron, then every face $ab$ must be space-like. The interfaces of interest are then heterochronal or achronal. The critical point equations \eqref{varg_sl} for the first type were obtained in the previous section, and the result of \cite{Barrett2010} for the second type has the exact same form (this is not surprising, since the variations with respect to $g_a$ should only depend on the states $\Psi_{ab}$ and not $\Psi_{ba}$). Equation \eqref{eq:groupvar} then implies
\begin{equation}
\forall \text{ s.l. } a\,, \quad \sum_{b} j_{ab} \braket{+_{ab} , \, G^l(+_{ab})}_\mathbbm{1} =0\,.
\end{equation}
Now, by noting that $G=(\vec{J}, i\vec{J})$, and taking into consideration equation \eqref{su2vec} of the appendix, we may rewrite this condition as
\begin{equation}
\label{closurespace}
\forall \text{ s.l. } a\,, \quad \sum_{b}j_{ab} \vec{n}_{ab}=0\,,
\end{equation}
where $\vec{n}_{ab}=2\braket{+_{ab}, \vec{J} (+_{ab})}$ denotes a vector in $S^2$ associated with the coherent state $\ket{+_{ab}}$ for the face $ab$. We therefore see that the variation with respect to the group elements induces a closure condition associated to space-like polyhedra.

\subsubsection{Time-like polyhedra}
\label{section:closure_t}
This time we take $a$ to label a time-like polyhedron. Consequently, every interface $ab$ will either be space-like, and thus orthochronal or heterochronal, or time-like, and thus parachronal. Again, as before, the equations for the variation of the action with respect to the group for heterochronal interfaces obtained in \cite{Kaminski2018a} have the same form as \eqref{varg11}. Collecting equations \eqref{varg11} and \eqref{eq:ga_p}, and again noting that $H=(\vec{F},i\vec{F})$, we find, for every term in the sum of \eqref{eq:groupvar},
\begin{equation}
\forall \text{ t.l. } a\,, \quad  \sum_{b:\,\text{s.l.}ab}j_{ab} \tau_{ab}\braket{\tau_{ab}, \,  \vec{F} \tau_{ab}}_{\sigma_3} + \sum_{b:\,\text{t.l.}ab}-i {s_{ab}} \left(\braket{l^-_{ab},\vec{F} l^+_{ab}}_{\sigma_3}+\beta_{ab} \braket{l^-_{ab},\vec{F} l^-_{ab}}_{\sigma_3} \right)=0\,.
\end{equation}
We may again rewrite this equation in terms of the vectors determined by the coherent states by referring to equations \eqref{eq:su11discvec} and \eqref{eq:su11contvec}. Let us denote by $\vec{n}^\tau_{ab}=2\tau_{ab}\braket{\tau_{ab}, \,  \vec{F} \tau_{ab}}_{\sigma_3}$ and $\vec{m}_{ab}=2i\braket{l^+_{ab},\vec{F} l^-_{ab}}_{\sigma_3}$ the vectors in the two-sheeted time-like hyperboloid $H^\pm$ and the one-sheeted space-like hyperboloid $H^{\text{sp}}$, respectively. One finds then the two equations
\begin{equation}
\label{eq:closure1_p}
\forall \text{ t.l. } a\,, \quad \sum_{b:\,\text{s.l.}ab} j_{ab} \vec{n}_{ab}^\tau + \sum_{b:\,\text{t.l.}ab} s_{ab} \left( \vec{m}_{ab} + \Im \beta_{ab} \tilde{m}_{ab}\right)=0\,,
\end{equation}
\begin{equation}
\label{eq:closure2_p}
    \forall \text{ t.l. } a\,, \quad \sum_{b:\,\text{t.l.}ab} s_{ab} \Re \beta_{ab} \tilde{m}_{ab}=0\,,
\end{equation}
having introduced the vector $\tilde{m}_{ab}=2\braket{l^-_{ab},\vec{F} l^-_{ab}}_{\sigma_3}$; it can be shown to be null and future-directed by explicit computation, and entirely independent of $\tau_{ab}$.

This is a good opportunity to draw a comparison between our analysis and the one carried out in \cite{Liu2019}. As we already mentioned in passing, there the authors performed an asymptotic approximation of the $\Omega_{ab}$ function, constructed with the Conrady-Hnybida prescription $\lambda=\sqrt{s^2+1/4}$ for the relevant coherent states \cite{Conrady2010}. In order to avoid having to resort to such an approximation, we have chosen to parameterize the theory using complex values of $\lambda$, leading to an exact and concise expression for $\Omega_{ab}$ as in equations \eqref{eq:omega_p}, \eqref{eq:pref_p} and \eqref{eq:para_action}. However, despite the apparent differences in the construction of the states, the expression we found for $\Omega_{ab}$ carries a similar structure as the one in \cite{Liu2019}: in particular, in both analyses one finds half-order branch cuts in \eqref{eq:pref_p}, and the action \eqref{eq:para_action} agrees with the dominant action of \cite{Liu2019}, i.e. the action permitting critical points.

Now, it was understood in \cite{Liu2019} that, in order for the critical point equations \eqref{eq:closure1_p} and \eqref{eq:closure2_p} to yield configurations that match our semiclassical expectations (Regge-like piece-wise linear geometries described by the Regge action), one must require all $\beta_{ab} \in \mathbb{C}$ to vanish. There it was shown, for the particular case of a triangulation by simplices, that $\Re \beta_{ab}$ must indeed be zero if one makes certain assumptions on the causal character of the triangles in a time-like tetrahedron: in case a time-like tetrahedron contains at least one space-like triangle, the critical point equations determine that a sum of three light-like vectors or less proportional to $\Re \beta_{ab}$ must vanish; this is only possible if every vector vanishes identically, thus imposing $\Re \beta_{ab} = 0$. To ensure this on the entire triangulation and simultaneously avoid degenerate 4-simplices/vector geometries, the authors proposed that each 4-simplex should consist of at least one space-like and one time-like tetrahedron, where thus the latter has both space- and time-like triangles. This condition, however, explicitly excludes triangulations used in Causal Dynamical Triangulations (CDT) \cite{Ambjorn:2012jv}, which necessarily contain 4-simplices consisting entirely of time-like tetrahedra\footnote{ These 4-simplices are called $(3,2)$-simplices (or $(2,3)$-simplices), denoting a simplex with three vertices in a ``past'' spatial slice and two vertices in the ``future'' one. Such a simplex contains a spatial triangle in the past slice and a spatial edge in the future one. No tetrahedron is entirely contained in one spatial slice and thus all tetrahedra are time-like.}. Such an assumption would then constitute an obstacle to a putative correspondence between the spin-foam and CDT approaches.

We propose alternatively that we relax this assumption, and extend the validity of the argument to any cellular decomposition, by imposing a condition on the $\tau_{ab}$ parameters of parachronal faces. Note that the critical point equations for parachronal interfaces, unlike for heterochronal ones, do not depend on $\tau_{ab}$, and hence neither does the physical content of the dominant contributions\footnote{We remark that the vectors associated with orthochronal faces live in one of the sheets of the two-sheeted hyperboloid, depending on $\tau_{ab}$. Vectors associated to parachronal faces, on the other hand, live on the one-sheeted hyperboloid, and there $\tau_{ab}$ plays no geometrical role.}. It is therefore conceivable that one may demand of all $\tau_{ab}$ of all time-like polyhedra to agree; if one does so, then the signs of all $\Re \beta_{ab}$ must agree too, since the integrand \eqref{eq:pref_p} contains a factor
\begin{equation}
\Theta(\tau_{ab} \braket{g_a^T z, g_a^T z}_{\sigma_3})=\Theta(\tau_{ab} \, \Re \beta_{ab})\,.
\end{equation}
Under this requirement, equation \eqref{eq:closure2_p} determines that a collection of similarly-oriented null vectors must vanish, and this may only be the case if $\Re \beta_{ab}=0$. Note furthermore that in this circumstance the phases $\varphi^{\tau}_{s,\epsilon}$ in equation \eqref{eq:prefactor} mutually cancel.

It finally remains to deal with the variables $\Im \beta_{ab}$ appearing in equation \eqref{eq:closure1_p}. Unfortunately, there does not seem to be an asymptotic condition demanding these terms to vanish. Still we would like to point out that such an obstruction also appeared in \cite{Liu2019}, where the authors required $\Im \beta_{ab}=0$ from the observation that their integrand $\Omega_{ab}$ did not depend on this variable; such a choice would then constitute a ``gauge fixing''. However, from this fact alone it is not clear that configurations with $\Im \beta_{ab} \neq 0$ should be exponentially suppressed or at least subdominant relative to Regge-like configurations. Different choices of $\Im \beta_{ab}$ should indeed lead to different vector solutions of equation \eqref{eq:closure1_p}, and it is conceivable that Regge-like boundary data might have additional contributions when non-vanishing $\Im \beta_{ab}$ are chosen, such that all light-like vectors in the staionary phase equations cancel. We would thus argue that at this stage there does not seem to be any reason to believe that configurations characterised by $\Im \beta_{ab}=0$ should dominate in the expansion. This seems to suggest that the model, at least in its extension to time-like interfaces, is not sufficiently constrained.

Since one would still benefit from an analysis of the special subset of critical points with vanishing $\beta_{ab}$, we shall proceed by restricting our attention to them, keeping in mind the already mentioned caveat that they are but a part of the  dominant contributions to the spin-foam amplitude. For simplicity we will call these special critical points \textit{geometrical}, and for them it holds that
\begin{equation}
\label{eq:closure_p}
\forall \text{ t.l. } a\,, \quad \sum_{b:\,\text{s.l.}ab} j_{ab} \vec{n}_{ab}^\tau + \sum_{b:\,\text{t.l.}ab}  s_{ab} \vec{m}_{ab}=0\,.
\end{equation}

\subsection{Bivector constraints}
We now turn to the algebraic equations coming from the $\mathbf{z}_{ab}$ variations,
\begin{equation}
    \delta_{\mathbf{z}_{ab}} S_{ab}(\mathbf{z}_{ab},g_a,g_b)=0, \,\forall \,a,b\,,
\end{equation}
and we shall argue they can be rephrased as constraints on certain bivectors of $\Lambda^2 \mathbb{R}^{3,1}$.
\subsubsection{Heterochronal interfaces}
Recall that, from the reality condition \eqref{realityspace}, the spinor variables $\ket{z_{ab}}$ must satisfy
\begin{align}
&\ket{g_a^T z_{ab}}=\lambda_{ab}\ket{+_{ab}}\,, \quad a \text{ s.l., } ab \text{ s.l.} \nonumber\\
&\ket{g_b^T z_{ba}}=\lambda_{ba}\ket{\tau_{ba}}\,,  \quad b \text{ t.l., } ba \text{ s.l.} \nonumber
\end{align}
Since $\mathbf{z}_{ab}=\mathbf{z}_{ba}$ by construction, we may group these previous equations into one, finding
\begin{equation}
\label{realvar}
(g_a^T)^{-1} \ket{+_{ab}}=\frac{\lambda_{ba}}{\lambda_{ab}}(g_b^T)^{-1} \ket{\tau_{ba}}\,.
\end{equation}
The variations with respect to $\mathbf{z}_{ab}$, on the other hand, imply
\begin{equation}
\label{zvar}
\overline{g}_a \ket{+_{ab}}=\frac{\overline{\lambda}_{ab}}{\overline{\lambda}_{ba}}\overline{g}_b \tau_{ba} \sigma_3 \ket{\tau_{ba}}\,.
\end{equation}
We intend to write the previous algebraic equations in more geometrical terms by making their dependence on vectors in $S^2, H^\pm$ explicit. Referring to section \ref{appendix:geometrygroups} of the appendix, coherent states determine such vectors as follows:
\begin{align}
&\vec{n}_{ab}=2\braket{+_{ab}\, | \, J^i  (+_{ab})}_\mathbbm{1}  \hat{e}_i \in S^2\,, \\
&\vec{n}^\tau_{ab}=2\tau_{ba}\braket{\tau_{ab}\, | \, F^i  \tau_{ab}}_{\sigma_3} \hat{e}_i \in H^\tau\,.
\end{align}
If, instead of writing these vectors in the canonical $\mathbb{R}^3$ basis, we choose to express them using the bases for $\mathfrak{su}(2)$ and $\mathfrak{su}(1,1)$, we find
\begin{align}
\label{exp1}
&n_{ab}^i J_i=\ket{+_{ab}}\bra{+_{ab}}-\frac{\mathbbm{1} }{2} \,, \\
\label{exp2}
&n_{ab}^{\tau\, i} F_i =\tau_{ba}{\sigma_3} \ket{\tau_{ab}}\bra{\tau_{ab}}-\frac{\mathbbm{1} }{2}\,,
%&m_{ab}^{i}F_i=-i\left( \sigma_3\ket{l^+_{ab}}\bra{l^-_{ab}} -\frac{\mathbbm{1} }{2}\right)\,,
\end{align}
where we used ${\sigma_3} F_i {\sigma_3}=F_i^\dagger$ and the completeness relation for Pauli matrices. The structure of these vectors suggests then the following construction: tensoring equations \eqref{zvar} with the conjugate transpose of \eqref{realvar}, and plugging in the algebra representation of our vectors, we get
\begin{equation}
\label{A2}
\vec{n}_{ab} \cdot \overline{g}_a \vec{J} \, \overline{g}_a^{-1}=\vec{n}_{ba}^{\tau} \cdot \overline{g}_b \vec{F} \, \overline{g}_b^{-1}\,.
\end{equation}

\subsubsection{Parachronal interfaces}

Analogously, for the special geometrical subset of critical points we consider at $\beta_{ab}=0$, equation \eqref{eq:expand} determines that
\begin{align}
&\ket{g_a^T z_{ab}}=\alpha_{ab}\ket{l^+_{ab}}\,,  \quad a \text{ t.l., } ab \text{ t.l.} \nonumber\\
&\ket{g_b^T z_{ba}}=\alpha_{ba}\ket{l^+_{ba}}\,,  \quad b \text{ t.l., } ba \text{ t.l.} \nonumber
\end{align}
from where we find
\begin{equation}
\label{timealgebra1}
(g_a^T)^{-1} \ket{l^+_{ab}}=\frac{\alpha_{ba}}{\alpha_{ab}}(g_b^T)^{-1} \ket{l^+_{ba}}\,,
\end{equation}
while the $z_{ab}$ variation implies
\begin{equation}
\label{timealgebra2}
\overline{g}_a \sigma_ 3 \ket{l^-_{ab}}=\frac{\overline{\alpha}_{ab}}{\overline{\alpha}_{ba}}\overline{g}_b \sigma_3 \ket{l^-_{ba}}\,.
\end{equation}
The coherent state associated to time-like faces reads
\begin{equation}
\vec{m}_{ab}=2i\braket{l^+_{ab}\, | \, F^i  l^-_{ab}}_{\sigma_3} \hat{e}_i \in H^\text{sp}\,,
\end{equation}
as per the discussion in section \ref{appendix:geometrygroups}, and its expansion in terms of the algebra is simply
\begin{equation}
\label{exp3}
m_{ab}^{i}F_i=i\left( \sigma_3\ket{l^-_{ab}}\bra{l^+_{ab}} -\frac{\mathbbm{1} }{2}\right)\,.
\end{equation}
A similar argument to the one above then allows us to find an equation in terms of $\vec{m}_{ab}$,
\begin{equation}
\label{B12}
\vec{m}_{ab} \cdot \overline{g}_a \vec{F} \,\overline{g}_a^{-1}=\vec{m}_{ba} \cdot \overline{g}_b \vec{F} \,\overline{g}_b^{-1}\,.
\end{equation}

\subsubsection{Embedding into \texorpdfstring{$\Lambda^2 \mathbb{R}^{3,1}$}{bivectors}}
In order to proceed it will prove useful to recall some facts regarding the well-known canonical spin homomorphism
\begin{equation}
\begin{gathered}
\pi: \text{SL}(2,\mathbb{C}) \rightarrow \text{SO}^+(3,1) \\
g \mapsto \frac{1}{2}\text{Tr}\left(g\sigma_\mu g^\dagger \sigma_\nu\right)\hat{e}^\mu \otimes \hat{e}^\nu \,,
\end{gathered}
\end{equation}
mapping the special linear group to the orthochronous Lorentz group. It arises from the observation that one may define an action of $\text{SL}(2,\mathbb{C})$ on $\mathbb{R}^{3,1}$ as
\begin{equation}
    x^\mu \sigma_\mu \mapsto g x^\mu \sigma_\mu g^\dagger= \left[\pi(g)x\right]^\mu \sigma_\mu\,,
\end{equation}
which we will use extensively moving forward.

Consider first an infinitesimal SL$(2,\mathbb{C})$ transformation, which we write as $g=\mathbbm{1} +\frac{i}{2} \omega_{\alpha \beta} J^{\alpha \beta}$, with $J^{\alpha \beta}$ the generators of the real $\mathfrak{sl}(2,\mathbb{C})$ algebra\footnote{It is well-known that given any SL$(2,\mathbb{C})$ group element $g$, either $g$ or $-g$ can be written as the exponentiation of some algebra element \cite{Buchbinder1998}.}. Using the previous equation and Pauli matrix identities, one may show that under the spin homomorphism it projects to $\pi(g)^\mu_{\;\;\nu}=\delta^\mu_{\;\;\nu}+\omega^\mu_{\;\;\nu}$ \cite{Buchbinder1998}. Now, by virtue of the Lie bracket of the $J^{\mu \nu}$ generators \cite{weinberg2012the},
\begin{equation}
i\left[J^{\mu \nu}, J^{\rho \sigma} \right]=\eta^{\nu \rho} J^{\mu \sigma} - \eta^{\mu \rho} J^{\nu \sigma} - \eta^{\sigma \mu} J^{\rho \nu} + \eta^{\sigma \nu} J^{\rho \mu}\,,
\end{equation}
one may see that
\begin{align}
(\mathbbm{1} +\frac{i}{2}  \omega_{\alpha \beta} J^{\alpha \beta})J_{\mu \nu} (\mathbbm{1} -\frac{i}{2}\omega_{\gamma \rho} J^{\gamma \rho})&=J_{\mu \nu}+\omega^{\beta}_{\;\;\nu}J_{\mu \beta}+\omega^{\alpha}_{\;\;\mu}J_{\alpha \nu}+\mathcal{O}(\omega^2) \nonumber \\
&= \pi (\mathbbm{1} +\frac{i}{2} \omega_{\gamma \rho} J^{\gamma \rho})^{\alpha}_{\;\;\mu} \pi (\mathbbm{1} +\frac{i}{2} \omega_{\gamma \rho} J^{\gamma \rho})^{\beta}_{\;\;\nu} J_{\alpha \beta} \nonumber \,,
\end{align}
implying the transformation law
\begin{equation}
x^{\mu \nu}g J_{\mu \nu} g^{-1}=\left[\pi(g)^{\wedge 2} x \right]^{\mu \nu} J_{\mu \nu}\,.
\end{equation}
In order to make use of the above identity, we will rewrite the vectors appearing in equations \eqref{A2} and \eqref{B12} in terms of $J^{\mu \nu}$. Making use of the relations\footnote{We use $\epsilon^{0123}=-\epsilon_{0123}=1$.} $K^i=J^{0i}$ and $J^i=\frac{1}{2}\epsilon^{0i}_{\quad jk}J^{jk}$, and defining $n_iJ^i=n_{\mu \nu}J^{\mu \nu}$, $n^\tau_iF^i=n^\tau_{\mu \nu}J^{\mu \nu}$ and $m_iF^i=m_{\mu \nu}J^{\mu \nu}$, we find
\begin{equation}
\vec{n} \mapsto n = \begin{pmatrix}
0 & 0 & 0 & 0 \\
0 & 0 & -n^3 & n^2 \\
0 & n^3 & 0 & -n^1 \\
0 & -n^2 & n^1 & 0
\end{pmatrix}=\star\left[(0,\vec{n})\wedge (1,\vec{0})\right]\,,
\end{equation}
\begin{equation}
\vec{n}^\tau \mapsto n^\tau = \begin{pmatrix}
0 & n^{\tau 2} & n^{\tau 3} & 0 \\
-n^{\tau 2} & 0 & -n^{\tau 1} & 0 \\
-n^{\tau 3} & n^{\tau 1} & 0 & 0 \\
0 & 0 & 0 & 0
\end{pmatrix}=\star\left[(\xi\vec{n}^\tau,0)\wedge (\vec{0},1)\right]\,,
\end{equation}
\begin{equation}
\vec{m} \mapsto m = \begin{pmatrix}
0 & m^{ 2} & m^{3} & 0 \\
-m^{2} & 0 & -m^{ 1} & 0 \\
-m^{ 3} & m^{1} & 0 & 0 \\
0 & 0 & 0 & 0
\end{pmatrix}=\star\left[(\xi\vec{m},0)\wedge (\vec{0},1)\right]\,,
\end{equation}
where $\xi$ is simply an $\text{SO}(3)$ rotation matrix
\begin{equation}
\label{rotation}
\xi=\begin{pmatrix}
1 & 0 & 0 \\
0 & 0 & -1 \\
0 & 1 & 0
\end{pmatrix}\,.
\end{equation}
Because each bivector has a fixed structure, in that the second vector in the wedge product is independent of the particular coherent state that induces it, one may further establish an injection between those states and elements of $\mathbb{R}^{3,1}$,
\begin{equation}
\label{eq:4vec}
    \ket{+_{ab}} \mapsto (0,\vec{n}_{ab})\,, \quad \ket{\tau_{ab}} \mapsto (\xi\vec{n}^\tau_{ab},0)\,, \quad\ket{l^\pm_{ab}} \mapsto (\xi\vec{m}_{ab},0)\,.
\end{equation}

\subsubsection{Bivector equations}
Returning to equations \eqref{A2} and \eqref{B12}, the discussion above allows us to rewrite them in the particularly simple form
\begin{equation}
\label{bivspace}
\pi(\overline{g}_a)^{\wedge 2}  \, n_{ab} = \pi(\overline{g}_b)^{\wedge 2}  \, n^\tau_{ba}\,,
\end{equation}
\begin{equation}
\label{bivtime}
\pi(\overline{g}_a)^{\wedge 2}  \, m_{ab} = \pi(\overline{g}_b)^{\wedge 2}  \, m_{ba}\,,
\end{equation}
where $\pi(\overline{g}_a)^{\wedge 2}$ denotes the induced linear map on the exterior space $\Lambda^2 \mathbb{R}^{3,1}$, and these are the bivector constraints that the section title alludes to.
Since it will be useful for a later discussion in Section \ref{section:parameters}, we will now bring these equations in a similar but different form. Recall that the Hodge star establishes a bilinear correspondence $\star: \, \Lambda^j V \rightarrow \Lambda^{n-j} V$ on an $n$-dimensional vector space $V$, and the inner product in $\Lambda^j V$ may be written as
\begin{equation}
    \braket{u,v}=\star^{-1}\left( \star u \wedge v\right)\,.
\end{equation}
Then the Hodge operator commutes with $\text{SO}(n)$ linear maps $g^{\wedge}$ on the exterior algebra. Indeed, for $a \in V$ and $\beta \in \Lambda^{n-1} V$, it is straightforward to see from the equation above that $a \wedge\left[ g^{\wedge (n-1)} \beta \right]= \left[\star (g^T)^{\wedge (n-1)} \star^{-1} a \right]\wedge \beta$. It follows that
\begin{equation}
\begin{gathered}
\left[ g^{\wedge (n-1)} \beta \right]\wedge \left[ g \,  a\right] = g^{\wedge n} \left[\beta \wedge a \right]  \\
\Leftrightarrow  \; \beta \wedge \left[\star (g^T)^{\wedge (n-1)} \star^{-1} \cdot g \, a \right]= \beta \wedge a \, \text{det}(g)\,,
\end{gathered}
\end{equation}
and thus
\begin{equation}
\label{eq:hodge}
g^{-1}= \star (g^T)^{\wedge (n-1)} \star^{-1} \text{det}(g)^{-1}\,.
\end{equation}
For $g$ in the special orthogonal group this implies that $g \star = \star g^{\wedge (n-1)}$, and the generalization to $g^{\wedge j}\star=\star g^{\wedge (n-j)}$ is immediate.
We may therefore write for our $\text{SO}(3,1)$ case that $\star g^{\wedge 2} = g^{\wedge 2} \star$, and the bivector constraints can be equivalently formulated as
\begin{equation}
\label{nostar}
\pi(\overline{g}_a)^{\wedge 2} \, \left((1,\vec{0})\wedge (0,\vec{n}_{ab})\right) = \pi(\overline{g}_b)^{\wedge 2} \, \left((\vec{0},1)\wedge (\xi \vec{n}_{ba}^\tau,0)\right)\,,
\end{equation}
\begin{equation}
\label{nostar2}
\pi(\overline{g}_a)^{\wedge 2}  \, \left((\vec{0},1)\wedge (\xi \vec{m}_{ab},0)\right) = \pi(\overline{g}_b)^{\wedge 2} \, \left((\vec{0},1)\wedge ( \xi \vec{m}_{ba},0)\right)\,.
\end{equation}

\subsection{Complex parameters}
\label{section:parameters}
One can also assign a more geometric significance to the complex parameters $\lambda, \alpha$ appearing in all the equations above. To do so, we have to make use of the ``partner" states $\ket{-_{ab}},\ket{-\tau_{ab}},\ket{l^+_{ab}}$ to the selected boundary states $\ket{+_{ab}}, \ket{\tau_{ab}}, \ket{l^-_{ab}}$, i.e. those defined by the action of the stabiliser group element on the complementary basis state. To that end define the structure map (as it is called in \cite{Barrett2010})
\begin{equation}
\begin{gathered}
	J: \mathbb{C}^2 \rightarrow \mathbb{C}^2 \\
	(a,b)^T \mapsto (-\overline{b},\overline{a})^T=(-i \sigma_2)\overline{(a,b)}^T \,,
\end{gathered}
\end{equation}
which can be used to invert and conjugate-transpose any GL$(2,\mathbb{C})$ matrix $g$ by $JgJ^{-1}=\det g \, (g^\dagger)^{-1}$, as one may readily show. Acting with this map on states $\ket{+_{ab}}, \ket{\tau_{ab}}, \ket{l^-_{ab}}$ yields the following equations:
\begin{equation}
\begin{gathered}
J\ket{+_{ab}}=J h_{ab} \ket{+}=(h_{ab}^{\dagger})^{-1} J \ket{+}= \ket{-_{ab}}\,, \quad h_{ab} \in \text{SU}(2) \\
J \sigma_3  \ket{\tau_{ab}}=- \sigma_3 J h_{ab}\ket{\tau}=-\sigma_3 (h_{ab}^{\dagger})^{-1}J\ket{\tau}=h_{ab}\sigma_3 J\ket{\tau}= \ket{-\tau_{ab}}\,, \quad h_{ab} \in \text{SU}(1,1) \\
J \sigma_3  \ket{l^-_{ab}}=-h_{ab} \sigma_3 J \ket{l^-}=-\ket{l^-_{ab}}\,, \quad h_{ab} \in \text{SU}(1,1)\,. \\
\end{gathered}
\end{equation}
Let us now go back to equations \eqref{realvar} and \eqref{zvar}. Using the previous relations between the partner states (and $J^{-1}=-J$), those equations can be shown to respectively imply
\begin{equation}
\label{realvarcomp}
\begin{gathered}
\overline{g}_a \ket{-_{ab}}=-\frac{\overline{\lambda}_{ba}}{\overline{\lambda}_{ab}}\overline{g}_b \sigma_3 \ket{-\tau_{ba}} \\
\overline{g}_a \sigma_ 3 \ket{l^+_{ab}}=\frac{\overline{\alpha}_{ba}}{\overline{\alpha}_{ab}}\overline{g}_b \sigma_3 \ket{l^+_{ba}}\,,
\end{gathered}
\end{equation}
\begin{equation}
\label{zvarcomp}
\begin{gathered}
(g_a^T)^{-1} \ket{-_{ab}}=\frac{\lambda_{ab}}{\lambda_{ba}}(g_b^T)^{-1}\tau_{ba}\ket{-\tau_{ba}} \\
(g_a^T)^{-1} \ket{l^-_{ab}}=\frac{\alpha_{ab}}{\alpha_{ba}}(g_b^T)^{-1} \ket{l^-_{ba}}\,.
\end{gathered}
\end{equation}
We will use these equations, as well as the ones for the complementary states, to elucidate the meaning of the complex parameters $\lambda_{ab}$ and $\alpha_{ab}$. For clarity we separate the case that refers to both types of interfaces we consider.

\subsubsection{Heterochronal interfaces}
Looking back at the first lines of equations \eqref{realvar} and \eqref{zvar}, notice that, by factoring one of the states, one finds an eigenvalue equation for $\ket{+_{ab}}$,
\begin{equation}
	g_a^T (g_b^T)^{-1}\sigma_3 \overline{g}_b^{-1} \overline{g}_a \ket{+_{ab}}=\tau_{ba} \left|{\frac{\lambda_{ab}}{\lambda_{ba}}}\right|^2 \ket{+_{ab}}.
\end{equation}
The same procedure applied to \eqref{realvarcomp} and \eqref{zvarcomp} analogously yields
\begin{equation}
	g_a^T (g_b^T)^{-1}\sigma_3 \overline{g}_b^{-1} \overline{g}_a \ket{-_{ab}}=-\tau_{ba} \left|{\frac{\lambda_{ab}}{\lambda_{ba}}}\right|^{-2} \ket{-_{ab}},
\end{equation}
and this is enough to characterize the matrix $g_a^T (g_b^T)^{-1}\sigma_3 \overline{g}_b^{-1} \overline{g}_a$: since its eigenvectors form a basis for the vector space, it is diagonalizable. We may then write, by referring to \eqref{exp1},
\begin{equation}
	g_a^T (g_b^T)^{-1}\sigma_3 \overline{g}_b^{-1} \overline{g}_a= \tau_{ba} e^{2 \ln \left|{\frac{\lambda_{ab}}{\lambda_{ba}}}\right|^2  \vec{n}_{ab} \cdot \vec{J}} 2\vec{n}_{ab} \cdot \vec{J}\,,
\end{equation}
as the matrix on the right-hand side has the same eigenvectors and eigenvalues. Now, massaging both sides, the previous equation implies
\begin{equation}
	\left[\pi\left( \bar{g}_a^{-1}\right) \pi\left( \bar{g}_b\right) (\vec{0},1)\right]^\mu \sigma_\mu = \left[\tau_{ba}\, \sigma_0 \cosh \ln  \left|\frac{\lambda_{ba}}{\lambda_{ab}}\right|^2 + \tau_{ba}\, \vec{n}_{ab} \cdot \vec{\sigma} \sinh \ln  \left|\frac{\lambda_{ba}}{\lambda_{ab}}\right|^2 \right] \vec{n}_{ab} \cdot \vec{\sigma}\,,
\end{equation}
from where we infer
\begin{align}
	\sinh  \ln  \left|\frac{\lambda_{ba}}{\lambda_{ab}}\right|^2 &= \tau_{ba}\text{Proj}_{e_0}\left[\pi\left( \bar{g}_a^{-1}\right) \pi\left( \bar{g}_b\right) (\vec{0},1)\right] \nonumber\\
	\label{eq:ln}
	&=\tau_{ba}\braket{\pi(\overline{g}_a)(1,\vec{0}),\, \pi(\overline{g}_b) (\vec{0},1)}\,.
\end{align}
The orthogonal projection map used above is defined in appendix \ref{appendix:geometry}.

\subsubsection{Parachronal interfaces}
\label{section:timelikeparameters}
The procedure for the time-like faces follows a similar line. We again use the available equations for the states $\ket{l^+_{ab}},\ket{l^-_{ab}}$, finding
\begin{equation}
\sigma_3 \overline{g}_a^{-1} \overline{g}_b \sigma_3 g_b^T (g_a^T)^{-1}\ket{l_{ab}^-}=\frac{\overline{\alpha}_{ba}}{\overline{\alpha}_{ab}}\frac{\alpha_{ab}}{\alpha_{ba}}\ket{l_{ab}^-}\,,
\end{equation}
\begin{equation}
\sigma_3 \overline{g}_a^{-1} \overline{g}_b \sigma_3 g_b^T (g_a^T)^{-1}\ket{l_{ab}^+}=\frac{\alpha_{ba}}{\alpha_{ab}}\frac{\overline{\alpha}_{ab}}{\overline{\alpha}_{ba}}\ket{l_{ab}^+}\,.
\end{equation}
Again, since $\ket{l^+_{ab}},\ket{l^-_{ab}}$ span $\mathbb{C}^2$, and by equation \eqref{exp3}, we have the equality
\begin{equation}
    \sigma_3 \overline{g}_a^{-1} \overline{g}_b \sigma_3 g_b^T (g_a^T)^{-1}=e^{2 i \, \left[2i \text{arg}\left(\frac{\alpha_{ba}}{\alpha_{ab}}\right)\right] \vec{m}_{ab}\cdot \vec{F}^\dagger}\,.
\end{equation}
Using the canonical spinor map on the left-hand side, and Taylor-expanding the right\footnote{It is easy to see that ($\vec{m}\cdot \vec{F}^\dagger)^2=-1$.}, the above equation reads
\begin{equation}
\left[\pi\left( \bar{g}_a^{-1}\right) \pi\left( \bar{g}_b\right) (\vec{0},1)\right]^\mu \sigma_\mu = \sigma_3\left[\sigma_0 \cos \left(2\,\text{arg}\, \frac{\alpha_{ba}}{\alpha_{ab}}\right) - 2\vec{m}_{ab}\cdot \vec{F}^\dagger \sin \left(2\,\text{arg}\,\frac{\alpha_{ba}}{\alpha_{ab}}\right)\right]
\end{equation}
implying
\begin{align}
\cos \left(2\,\text{arg}\, \frac{\alpha_{ba}}{\alpha_{ab}}\right)&=\text{Proj}_{e_3}\left[\pi\left( \bar{g}_a^{-1}\right) \pi\left( \bar{g}_b\right) (\vec{0},1)\right] \nonumber \\
&=-\braket{\pi\left( \bar{g}_a\right) (\vec{0},1),\, \pi\left( \bar{g}_b\right)(\vec{0},1)}\,.
\end{align}
In order to fix the sign ambiguity in the argument of the cosine, we would like to relate the quantity $2\,\text{arg}\, \alpha_{ba}/\alpha_{ab}$ to the parameter of the Lorentz transformation taking $\pi\left( \bar{g}_a\right)(\vec{0},1)$ to $ \pi\left( \bar{g}_b\right)(\vec{0},1)$ and stabilising the plane defined by these vectors. We do this as follows: there must exist a special linear transformation $D_{ab}$ such that
\begin{equation}
    \pi(D_{ab})\pi\left( \bar{g}_a\right)(\vec{0},1)=\pi\left( \bar{g}_b\right)(\vec{0},1)\,,
\end{equation}
which, in terms of $\text{SL}(2,\mathbb{C})$ elements, implies
\begin{equation}
\label{aux}
D_{ab}\overline{g}_a \sigma_3 \overline{g}_a^\dagger D_{ab}^\dagger=\overline{g}_b \sigma_3 \overline{g}_b^\dagger\,.
\end{equation}
Now, by our requirement the linear transformation must preserve the bivector associated to the face, which reads
\begin{equation}
\label{planebiv}
     \star\left[\pi\left( \bar{g}_a\right)(\vec{0},1) \wedge \pi\left( \bar{g}_b \right)(\vec{0},1)\right]\propto  \pi\left( \bar{g}_a\right)^{\wedge 2} m_{ab}\,,
\end{equation}
as one may check by referring to equations \eqref{bivtime} and \eqref{nostar2}. Just as $g\in$ $\text{SL}(2,\mathbb{C})$ induces a transformation of space-time vectors $x^\mu\sigma_\mu \mapsto g x^\mu\sigma_\mu g^\dagger=(\pi(g)x)^\mu \sigma_\mu$, so too it induces a transformation of space-time bivectors
\begin{equation}
    x^{\mu \nu}\sigma_{\mu}\otimes \sigma_\nu \mapsto (g \otimes g) \left[x^{\mu \nu}\sigma_{\mu}\otimes \sigma_\nu\right] (g^\dagger \otimes g^\dagger) = [\pi(g)^{\wedge 2}x]^{\mu \nu} \sigma_{\mu}\otimes \sigma_\nu\,.
\end{equation}
Clearly, then, the group element stabilizing the bivector of equation \eqref{planebiv} must be given by one of either
\begin{equation}
    D_{ab}=\pm\bar{g}_a e^{i \theta_{ab} \left[\star m_{ab}\right]_{\mu \nu} J^{\mu\nu}} \bar{g}_a^{-1}\,,
\end{equation}
since $ e^{i \theta_{ab} \left[\star m_{ab}\right]_{\mu \nu} J^{\mu\nu}}$ stabilizes $ m_{ab}$ and $\pi$ is injective up to a $\mathbb{Z}_2$ factor. It is also straightforward to check that this element takes $a$ to $b$, and not vice-versa. Together with equation \eqref{aux}, this implies
\begin{equation}
\label{eq:arg}
\sigma_3 \overline{g}_a^{-1} \overline{g}_b \sigma_3 g_b^T (g_a^T)^{-1}= e^{-2 \theta_{ab} \vec{m}_{ab}\cdot \vec{F}^\dagger}\,,
\end{equation}
and therefore $\theta_{ab}=2\, \text{arg}\, \frac{\alpha_{ba}}{\alpha_{ab}}$.

\section{Induced geometry at the critical points} \label{sec:induced_geometry}

It remains to bring the previous analysis together into a geometrical picture of the critical points we focus on. For the reader's convenience, we reproduce here the relevant equations. Under an additional multiplication by the matrix $\xi$ of equation \eqref{rotation} for time-like polyhedra, equations \eqref{closurespace} and \eqref{eq:closure_p} read:
\begin{equation}
\label{closure}
\forall \text{ s.l. } a\,, \; \sum_{b\text{: s.l.}ab}j_{ab} \vec{n}_{ab}=0\,, \quad \; \forall \text{ t.l. } a\,, \; \sum_{b:\,\text{s.l.}ab} j_{ab} \xi\vec{n}_{ab}^\tau + \sum_{b:\,\text{t.l.}ab} s_{ab} \xi\vec{m}_{ab}=0\,.
\end{equation}
In addition to these closure conditions, we have as well the bivector constraints of equations \eqref{bivspace} and \eqref{bivtime}:
\begin{equation}
\label{bivectors}
\begin{gathered}
\pi(\overline{g}_a)^{\wedge 2}  \star\left((1,\vec{0})\wedge (0,\vec{n}_{ab})\right) = \pi(\overline{g}_b)^{\wedge 2} \star\left((\vec{0},1)\wedge (\xi\vec{n}_{ba}^\tau,0)\right)\,, \; a\text{ s.l., }b\text{ t.l., }ab\text{ s.l.}\,, \\
\pi(\overline{g}_a)^{\wedge 2}  \star\left((\vec{0},1)\wedge(\xi\vec{m}_{ab},0)\right) = \pi(\overline{g}_b)^{\wedge 2}  \star\left((\vec{0},1)\wedge (\xi\vec{m}_{ba},0)\right)\,, \; a\text{ t.l., }b\text{ t.l., }ab\text{ t.l.}\,.
\end{gathered}
\end{equation}
We reiterate that we are considering only parachronal and heterochronal interfaces for simplicity, and the equations above refer to those cases.

We shall assume that the boundary data is not degenerate, i.e. we require that the set of vectors $\{\vec{n}_{ab}, \vec{n}_{ab}^\tau, \vec{m}_{ab} \}$ at every $a$ spans a 3-dimensional subspace. If non-degeneracy of the boundary data is not assumed the geometrical meaning of the critical point equations is diluted, and a characterisation of the solutions is substantially more involved. Among those solutions characterised by degenerate data one finds in particular the so-called vector geometries \cite{Barrett2010}. Since our new parametrization is such that the geometrical critical points of heterochronal and parachronal interfaces agree with the ones of \cite{Kaminski2018a} and \cite{Liu2019}, respectively, we refrain from an explicit discussion of these more exotic geometries, and refer the reader to the existing literature. A rather exhaustive study of the most general solutions for the entirely space-like case with graphs of general combinatorial structure was carried out in \cite{Dona2020}.

Now, under the requirement of non-degenerate boundary data, Minkowski's theorem for Minkowskian polyhedra (proven in Theorem \ref{theorem:minkowski} of appendix \ref{appendix:geometry}) establishes that there exists, up to translations, a unique polyhedron satisfying the closure constraints of equation \eqref{closure}. If $a$ is a space-like label, such a polyhedron $P_a^s$ will be a subset of $\mathbb{R}^3$, while if $a$ is time-like $P^t_a$ will be a subset of $\mathbb{R}^{2,1}$. In both cases the areas and normals to the faces of the polyhedra will be given by the spins $j_{ab}$, $s_{ab}$ and vectors $\vec{n}_{ab}$, $\vec{n}_{ab}^\tau$, $\vec{m}_{ab}$ appearing in equation \eqref{closure}. Turning to the bivector constraints \eqref{bivectors}, we may then define the vectors $N_a^{(1,\vec{0})}= \pi(\overline{g}_a)(1,\vec{0})$, $N_a^{(\vec{0},1)}=\pi(\overline{g}_a)(\vec{0},1)$, and consider the family of affine embeddings (we disregard translations)
    \begin{equation}
	\begin{gathered}
	    P_a^s \subset \mathbb{R}^3 \hookrightarrow  \hat{P}_a^{s} \subset \mathbb{R}^{3,1} \\
	     \text{ s.t. } \hat{P}_a^{s} \perp N_a^{(1,\vec{0})}\,, \; \hat{P}_{ab}^{s} \perp \pi(\overline{g}_a)(0,\vec{n}_{ab})
	\end{gathered}
	\end{equation}
	\begin{equation}
	\begin{gathered}
	    P_a^t \subset \mathbb{R}^{2,1} \hookrightarrow  \hat{P}_a^{t} \subset \mathbb{R}^{3,1} \\
	    \text{ s.t. } \hat{P}_a^{t} \perp N_a^{(\vec{0},1)}\,, \; \hat{P}_{ab}^{t} \perp
	    \begin{cases} \pi(\overline{g}_a)(\xi\vec{n}^\tau_{ab},0) \text{ if } P_{ab}^t \perp \vec{n}^\tau_{ab}\\ \pi(\overline{g}_a)(\xi\vec{m}_{ab},0) \text{ if } P_{ab}^t \perp \vec{m}_{ab}\end{cases}
	\end{gathered}
	\end{equation}
where $P_{ab}$ denotes the face of the polyhedron $P_a$ shared with $P_b$. The maps are such that the three-dimensional polyhedra are embedded in $\mathbb{R}^{3,1}$ perpendicularly to the four-normals $N_a^{(1,\vec{0})}$ and $N_a^{(\vec{0},1)}$.
The constraints in \eqref{bivectors} thus imply, respectively,
	\begin{equation}
	\begin{gathered}
	\hat{P}_{ab}^{s} \parallel \hat{P}_{ba}^{t}\,, \; A_{ab}=A_{ba}\,, \; \text{ for }ab\text{ s.l.} \\
	\hat{P}_{ab}^{t}\parallel \hat{P}_{ba}^{t}\,, \; A_{ab}=A_{ba}\,,\; \text{ for }ab\text{ t.l.}\,,
	\end{gathered}
	\end{equation}
where we denote by $A_{ab}$ the area of the polygonal face $P_{ab}$. Under this construction,  solutions to the critical point equations of non-degenerate boundary data describe different ways of gluing 3-dimensional polyhedra along their faces according to some fixed combinatorial structure, prescribed by the boundary data, such that the areas of glued faces agree - but not necessarily their polygonal shape. One obtains in this manner a polyhedral pseudo-complex\footnote{Polyhedral complexes are such that, by definition, the intersection of two polyhedra $P_1\cap P_2$ is a face of both. For a polyhedral pseudo-complex we require only that their intersection lies in the same plane as both faces.} with matching areas, frequently called a twisted geometry in the literature \cite{Freidel:2010aq,Dona2020}. Note that, for  the  particular  case  when  the  cellular  decomposition  $\Delta$  is  a triangulation by simplices, the critical point equations are enough to uniquely reconstruct a 4-simplex with shape matching at every polygonal face. This is because a 4-simplex is completely characterised up to isometries by its ten edge lengths, and these are fixed when tetrahedra are reconstructed from triangle areas through Minkowski's theorem. 

The above construction also endows the complex parameters $\lambda_{ab}$ and $\alpha_{ab}$ of section \ref{section:parameters} with a geometrical meaning. Equation \eqref{eq:ln} for heterochronal interfaces can be rewritten as
	\begin{equation}
	    \cosh{\left(\ln  \left|\frac{\lambda_{ba}}{\lambda_{ab}}\right|^2 - i\tau_{ab} \frac{\pi}{2}\right)}=\frac{\braket{N_a^{(1,\vec{0})}, N_b^{(\vec{0},1)}}}{||N_a^{(1,\vec{0})}|| \, ||N_b^{(\vec{0},1)}||}\,,
	\end{equation}
where we define $||\cdot||=\sqrt{||\cdot||^2}\in \mathbb{R}^+_0 \cup i \mathbb{R}^+_0$. Hence, according to our definition of Minkowski angles from appendix \ref{appendix:geometry}, we have that $\ln  \left|\frac{\lambda_{ba}}{\lambda_{ab}}\right|^2=\Re \theta_{N_a^{(1,\vec{0})}, N_b^{(\vec{0},1)}}$ (note that the real part of the angle is independent of the sign of the normals). In this case the complex parameters are associated to the real part of the dihedral angle between a space-like polyhedron $a$ and a time-like one $b$. In its turn, equation \eqref{eq:arg} for parachronal interfaces contains an inner product of vectors laying in a space-like plane, where the induced metric is Riemannian,
    \begin{equation}
        \cos\left(2\,\text{arg}\, \frac{\alpha_{ba}}{\alpha_{ab}} \right)= \frac{\braket{N_a^{(\vec{0},1)}, N_b^{(\vec{0},1)}}}{||N_a^{(\vec{0},1)}|| \, ||N_b^{(\vec{0},1)}||}\,,
    \end{equation}
and the discussion of subsection \ref{section:timelikeparameters} allows us to write $2\,\text{arg}\, \frac{\alpha_{ba}}{\alpha_{ab}}=\theta_{N_a^{(\vec{0},1)}, N_b^{(\vec{0},1)}}$. However, because normals to time-like 3-cells define an Euclidean plane, a key difference to the situation above and the cases discussed in \cite{Barrett2010, Kaminski2018a} arises: here the angle does depend on the sign of the normal vectors to the polyhedra, which, for a given solution to the critical points, may not all be outward-pointing. One can thus not interpret the angles $\theta_{N_a^{(\vec{0},1)}, N_b^{(\vec{0},1)}}$ as strictly dihedral.

\subsection{Symmetries of the solutions and the Cosine Problem}

In order to identify the possible symmetries of solutions to the critical point equations we first need to state the symmetries of the spin-foam actions \eqref{eq:actionspace},\eqref{eq:para_action} themselves. Just as for the actions of the remaining two kinds of interfaces analysed in \cite{Barrett2010,Kaminski2018a}, they satisfy the following symmetries:
\begin{itemize}
    \item Lorentz: a global action of $g \in \text{SL}(2,\mathbb{C})$ at every $g_a \mapsto g\, g_a$ and $\mathbf{z}_{ab} \mapsto {g^\dagger}^{-1}\mathbf{z}_{ab}$. This symmetry is gauge-fixed by the Dirac delta in \eqref{eq:vertex_amp}.
    \item Spin lift: a local transformation of $g_a \mapsto -g_a$.
    \item Spinor rescaling: a local transformation at each face $ab$ taking $\mathbf{z}_{ab} \mapsto \kappa \mathbf{z}_{ab}$, with $\kappa \in \mathbb{C}^*$.
\end{itemize}
These symmetries carry over to the critical point equations, as can be immediately seen (note that the special linear group elements $\pm g$ map to the same $\pi(g)$ under the canonical spin homomorphism). There is an apparent additional local symmetry of the bivector equations \eqref{bivectors} in taking $\pi(g_a) \mapsto -\pi(g_a)$, but since $-\pi(g_a) \notin \text{SO}^+(3,1)$ this Lorentz transformation does not lift to an $\text{SL}(2,\mathbb{C})$ element.

The easiest way to find possible additional symmetries of the equations which are not symmetries of the action is to make use of our geometrical reconstruction. As we have argued, a solution to those equations will determine a gluing of 3-dimensional polyhedra along their faces according to a given combinatorial prescription, such that the areas of glued faces match (such a gluing might be a polytope itself if the polygonal shapes of the faces happen to agree, but in general this need not be the case). One may then wonder how this object behaves under the the isometries of Minkowski space-time: these are translations, rotations, boosts, reflections and inversions. Among them, the only isometries not contained in the symmetries of the action are reflections with respect to some 3-dimensional subspace. The question arises of whether such transformations are indeed capable of inducing a different solution for the same boundary data, comprised of areas and 4-vectors \eqref{eq:4vec} determined by the coherent states. 
%We answer this by examining different possibilities:

To answer this, first denote by ${R_v}^\mu_{\;\nu}=\mathbbm{1} ^\mu_{\;\nu}-\frac{2v^\mu v_\nu}{\braket{v,v}}$ the reflection with respect to the subspace orthogonal to $v$. Now, under a global reflection of the polytope, a given bivector (which describes a face of a polyhedron) will transform as
\begin{equation}
\pi(g_a)^{\wedge 2} (N_a \wedge v_{ab}) \mapsto \left(R_v \, \pi(g_a) R_{N_a}\right)^{\wedge 2} (N_a \wedge v_{ab})\,,
\end{equation}
where $N_a$ and $v_{ab}$ denote any of the normals and vectors appearing in the bivector equations \eqref{bivectors}. That this is the case can be seen from the transformation property of the normal $\pi(g_a) N_a$, which in turn may be derived from how the edges $e$ of a polytope transform under a reflection $e \mapsto R_v e$:
\begin{align*}
    \pi(g_a) N_a &= \pi(g_a) \star(e_1\wedge e_2\wedge e_3)\\
    &=\star \,\pi(g_a)^{\wedge 3}(e_1\wedge e_2\wedge e_3) \\
    & \mapsto \star \,\left(R_v \,\pi(g_a)\right)^{\wedge 3}(e_1\wedge e_2\wedge e_3) \\
    &=\text{det}(R_v) \, R_v  \pi(g_a)\star(e_1\wedge e_2\wedge e_3) \\
    &= \, R_v  \, \pi(g_a) R_{N_a} N_a\,,
\end{align*}
where $e_1, e_2, e_3$ are three linearly independent edge vectors of the polyhedron associated to the bivector, and in commuting the Hodge star we have used equation \eqref{eq:hodge}. Furthermore, the resulting $\text{SO}^+(3,1)$ transformation admits a pre-image in $\text{SL}(2,\mathbb{C})$, as we now show. Let $v$ be some element of $\mathbb{R}^{3,1}$, and $g, \, i^{1-\delta_{\mu 0}}\sigma_\mu \in \text{SL(2},\mathbb{C})$. Then it holds that
\begin{align}
\label{eq:full_reflection}
\left[\pi\left(i^{1-\delta_{\mu 0}}\sigma_\mu(g^{-1})^\dagger i^{1-\delta_{\nu 0}}\sigma_\nu\right)v\right]^\alpha \sigma_\alpha &=\, \sigma_\mu(g^{-1})^\dagger\sigma_\nu v^\alpha \sigma_\alpha \sigma_\nu g^{-1}\sigma_\mu  \nonumber \\
&=\,\sigma_\mu(g^{-1})^\dagger (\Sigma_\nu v)^\alpha \sigma_\alpha  g^{-1}\sigma_\mu  \nonumber \\
&=\,\sigma_\mu J g (P \Sigma_\nu v)^\alpha \sigma_\alpha g^\dagger J^{-1} \sigma_\mu  \nonumber \\
&=\, \left[\pi(g)P\Sigma_\nu v\right]^\alpha \sigma_\mu J \sigma_{\alpha} J^{-1} \sigma_\mu  \nonumber\\
&=\,\left[P \Sigma_\mu \pi(g) \Sigma_\nu P v\right]^\alpha \sigma_\alpha  \,,
\end{align}
where $P=\text{diag}(1,-\vec{1})$ is the parity relation, $\Sigma_\nu$ is a $4\times4$ diagonal matrix of the following form,
\begin{equation}
    \Sigma_0=\mathbbm{1}\,, \quad (\Sigma_i)_{00}=(\Sigma_i)_{ii}=1=-(\Sigma_i)_{jj}\,,j\neq i\,,
\end{equation}
arising from commuting $\sigma_\nu$ and $\sigma_\alpha$, and we have used that $J \sigma_\mu J^{-1}=\text{det}(\sigma_\mu)\sigma_\mu$. Conjugation and inversion of $g$, followed by the adjoint action of the Pauli matrices, thus amounts to a reflection of the associated $\text{SO}^+(3,1)$ group element,
\begin{equation}
\label{eq:reflection}
i^{1-\delta_{\mu 0}}\sigma_\mu (g^{-1})^\dagger i^{1-\delta_{\nu 0}} \sigma_\nu \mapsto R_{e_\mu} \pi(g) R_{e_\nu}\,.
\end{equation}
Having understood how a reflection of the polyhedral pseudo-complex induces a transformation on the special linear matrices that define it, one can now study individually some of the different interfaces figuring in the model:
\begin{enumerate}
    \item Every polyhedron is space-like. The achronal bivector algebraic equations of \cite{Barrett2010}, adapted to our notation, read
    \begin{equation}
    \begin{gathered}
    (g_a^T)^{-1} \ket{+_{ab}}=\frac{\varsigma_{ba}}{\varsigma_{ab}}(g_b^T)^{-1} \ket{+_{ba}}\,, \\
    \overline{g}_a \ket{+_{ab}}=\frac{\overline{\varsigma}_{ab}}{\overline{\varsigma}_{ba}}\overline{g}_b \ket{+_{ba}}\,,
    \end{gathered}
    \end{equation}
    where $\varsigma_{ab}$ is, as in our analysis, the proportionality parameter in $\ket{g_a^T z_{ab}}=\varsigma_{ab}\ket{+_{ab}}$. Because every polygon and polyhedron is space-like, every normal is of the type $(1,0)$. Any reflection $R_{e_\mu}$ of the polytope induces a transformation of the group elements as $\pi(g_a) \mapsto R_{e_\mu} \pi(g_a) R_{e_0}$, corresponding in $\text{SL}(2,\mathbb{C})$ to $g_a \mapsto i^{1-\delta_{\mu 0}}\sigma_\mu (g_a^\dagger)^{-1}$. By making use of the global Lorentz gauge we may simply take $\mu=0$, and the resulting transformation can be understood as a parity operation (this is precisely the transformation found in \cite{Barrett2010}). One may then check that if $\{g_a, \varsigma_{ab}\}$ is a solution for given states $\{\ket{+_{ab}}\}$ then so is $\{(g_a^\dagger)^{-1}, \overline{\varsigma}_{ba}\}$.
    At the level of the bivector equations, geometrically more explicit, the existence of the second solution follows from the particular structure of the vectors associated with the boundary data, which remains invariant under a $R_{e_0}$ reflection:
    \begin{align}
        &\pi(\overline{g}_a)^{\wedge 2}  \left((1,\vec{0})\wedge (0,\vec{n}_{ab})\right) = \pi(\overline{g}_b)^{\wedge 2} \left((1,\vec{0})\wedge (0,\vec{n}_{ba})\right) \nonumber \\
        &\Leftrightarrow  \left[R_{e_0} \pi(\overline{g}_a) R_{e_0}\right]^{\wedge 2} \left((1,\vec{0})\wedge (0,\vec{n}_{ab})\right) = \left[R_{e_0} \pi(\overline{g}_b) R_{e_0}\right]^{\wedge 2}\left((1,\vec{0})\wedge (0,\vec{n}_{ba})\right)  \nonumber \\
        &\Leftrightarrow \pi\left((\overline{g}_a^{-1})^\dagger \right)^{\wedge 2}  \left((1,\vec{0})\wedge (0,\vec{n}_{ab})\right) = \pi\left((\overline{g}_b^{-1})^\dagger \right)^{\wedge 2}  \left((1,\vec{0})\wedge (0,\vec{n}_{ba})\right)\,, \nonumber
    \end{align}
where in the last line we used equation \eqref{eq:full_reflection}.
    \item Every polygon is time-like. In this case, the parachronal bivector algebraic equations are
    \begin{equation}
    \begin{gathered}
    (g_a^T)^{-1} \ket{l^+_{ab}}=\frac{\alpha_{ba}}{\alpha_{ab}}(g_b^T)^{-1} \ket{l^+_{ba}}\,, \\
    \overline{g}_a \sigma_ 3 \ket{l^+_{ab}}=\frac{\overline{\alpha}_{ba}}{\overline{\alpha}_{ab}}\overline{g}_b \sigma_3 \ket{l^+_{ba}}\,,
    \end{gathered}
    \end{equation}
    and all polyhedral normals are space-like and of the form $(\vec{0},1)$. A reflection of the polytope then acts as $\pi(g_a) \mapsto R_{e_\mu} \pi(g_a) R_{e_3}$, corresponding to $g_a \mapsto i^{1-\delta_{\mu 0}}\sigma_\mu (g_a^\dagger)^{-1}i \sigma_3$. Just as before, one can check than any solution $\{g_a ,\alpha_{ba}\}$ induces a second one $\{ (\overline{g}_a^{-1})^\dagger i\sigma_3 ,\overline{\alpha}_{ab}\}$.
    Precisely as in the space-like case, the second solution can be obtained from the bivector equations by considering an $R_{e_0}$ reflection, in agreement with  \cite{Kaminski2018a}:
    \begin{align}
        &\pi(\overline{g}_a)^{\wedge 2}  \left((\vec{0},1)\wedge (\xi\vec{m}_{ab},0)\right) = \pi(\overline{g}_b)^{\wedge 2} \left((\vec{0},1)\wedge (\xi\vec{m}_{ba},0)\right) \nonumber \\
        &\Leftrightarrow  \left[R_{e_0} \pi(\overline{g}_a) R_{e_3}\right]^{\wedge 2} \left((\vec{0},1)\wedge (\xi\vec{m}_{ab},0)\right) = \left[R_{e_0} \pi(\overline{g}_b) R_{e_3}\right]^{\wedge 2}\left((\vec{0},1)\wedge (\xi\vec{m}_{ba},0)\right)  \nonumber \\
        &\Leftrightarrow \pi\left((\overline{g}_a^{-1})^\dagger i\sigma_3 \right)^{\wedge 2}  \left((\vec{0},1)\wedge (\xi\vec{m}_{ab},0)\right) = \pi\left((\overline{g}_b^{-1})^\dagger i\sigma_3\right)^{\wedge 2}  \left((\vec{0},1)\wedge (\xi\vec{m}_{ba},0)\right) \nonumber\,.
    \end{align}
    \item The polyhedral pseudo-complex contains both types of polyhedra. Then at least some of the faces $ab$ will be described by the heterochronal equations
    \begin{equation}
    \begin{gathered}
    (g_a^T)^{-1} \ket{+_{ab}}=\frac{\lambda_{ba}}{\lambda_{ab}}(g_b^T)^{-1} \ket{\tau_{ba}}\,, \\
    \overline{g}_a \ket{+_{ab}}=\frac{\overline{\lambda}_{ab}}{\overline{\lambda}_{ba}}\overline{g}_b \tau \sigma_3 \ket{\tau_{ba}}\,.
    \end{gathered}
    \end{equation}
    A reflection of the polytope induces different transformations on $g_a$ and $g_b$, since these elements are associated to normals of the type $(1,\vec{0})$ and $(\vec{0},1)$, respectively. The individual transformations are as before, and given one solution $\{g_a,g_b,\lambda_{ab},\lambda_{ba}\}$  to the equations, a second one can be constructed as $\{(g_a^{-1})^\dagger, (g_b^{-1})^\dagger i\sigma_3,-i \tau \overline{\lambda}_{ba}, \overline{\lambda}_{ab}\}$.
    Regarding the bivector equation
    \begin{equation*}
        \pi(\overline{g}_a)^{\wedge 2}  \left((1,\vec{0})\wedge (0,\vec{n}_{ab})\right) = \pi(\overline{g}_b)^{\wedge 2} \left((\vec{0},1)\wedge (\xi\vec{n}_{ba}^\tau,0)\right)\,,
    \end{equation*}
    this transformation leads to
    \begin{align*}
        &\pi\left((\overline{g}_a^{-1})^{\dagger}\right)^{\wedge 2}  \left((1,\vec{0})\wedge (0,\vec{n}_{ab})\right)  = \pi\left((\overline{g}_b^{-1})^{\dagger} i\sigma_3\right)^{\wedge 2} \left((\vec{0},1)\wedge (\xi\vec{n}_{ba}^\tau,0)\right)\, \\
        &\Leftrightarrow \left[R_{e_0} \pi(\overline{g}_a) R_{e_0} \right]^{\wedge 2}  \left((1,\vec{0})\wedge (0,\vec{n}_{ab})\right)  = \left[ R_{e_0} \pi(\overline{g}_b) R_{e_3} \right]^{\wedge 2} \left((\vec{0},1)\wedge (\xi\vec{n}_{ba}^\tau,0)\right)\,,
    \end{align*}
    as had been found in \cite{Kaminski2018a}\footnote{The existence of a second solution for heterochronal interfaces eluded the authors in an earlier version of this paper, leading to the erroneous claim that the Cosine Problem would be absent in this case.}.
    
\end{enumerate}
We thus recover, for the Conrady-Hnybida extension, the well-known fact that any spin-foam critical point induces a second one, geometrically related to the first via a reflection of its underlying geometry. This second solution is the root of the so-called ``Cosine Problem", first described in the simplicial achronal analysis of \cite{Barrett2009a,Barrett2010}: as we mentioned before, in case the spin-foam model is derived from a simplicial triangulation, the critical point equations determine a unique simplex up to isometries; a given set of achronal boundary data then induces two solutions related by a reflection, and the asymptotic expansion of the spin-foam amplitude was shown to contain two terms with the generic form
\begin{equation}
    A_v \sim e^{i S_R} + e^{-i S_R} \sim \cos S_R \,,
\end{equation}
where $S_R$ denotes the Regge action, rather than the expected single term.
In the context of the CH extension, and taking into account the previous discussion, one may immediately see from the expression for the action at a critical point \eqref{eq:action} that a given pair of critical points still induces complex-conjugated Regge actions. Whether this leads to a cosine term in the asymptotic expansion of the amplitude, however, cannot strictly be discerned at this stage, as this would require an explicit expression for the asymptotics, and in particular for the functions that multiply the exponential of the action.

\subsection{The action at the critical points}

We would like to conclude our analysis with a description of the general form of the asymptotic expansion of equation \eqref{eq:asympt}. As we have argued in section \ref{section:closure_t}, the critical points of parachronal actions are not isolated, rather being elements of some critical surface in the integration domain, characterised by all possible values of $\Im \beta_{ab}$. Moreover, the prefactor of equation \eqref{eq:pref_p} contains branch cuts, and indeed a branch point of half order coinciding with the critical points of the action at $\Re \beta_{ab}=0$. Unfortunately, results concerning the uniform asymptotic expansion of generic higher-dimensional integrals with such structures are scarce throughout the literature. An analysis of a one-dimensional integral where the critical point and branch point coalesce can be found in \cite{Ciarkowski1989}, and one could use the results therein to formulate an heuristic expectation of the relative contribution of different critical configurations; but such an heuristic understanding would not be sufficiently satisfactory. That parachronal interfaces in particular seem to present such a more complex structure when compared to all other cases seems to indicate, once again, that their implementation in the spin-foam model might benefit from further refinement.

In the absence of the possibility of having an explicit expression for the asymptotic expansion, it is still illuminating to evaluate the heterochronal and parachronal actions figuring in the exponential part of \eqref{eq:asympt} at the geometrical critical points. Disregarding all other possible cases of interfaces, the full action associated to the $\Delta^*$ vertex reads
\begin{equation*}
    S_v=i \gamma \left[\sum^{\text{ hetero.}}_{ab} j_{ab}\, \ln\left| \frac{\lambda_{ba}}{\lambda_{ab}}\right|^2 +  \sum^{\text{ para.}}_{ab} s_{ab}\, 2\,\text{arg}\, \frac{\alpha_{ba}}{\alpha_{ab}} \right] + i \left[\sum^{\text{ hetero.}}_{ab} j_{ab}\, 2\,\text{arg}\, \frac{\lambda_{ab}}{\lambda_{ba}} +   \sum^{\text{ para.}}_{ab} s_{ab}\, \ln\left| \frac{\alpha_{ba}}{\alpha_{ab}}\right|^2  \right]\,,
\end{equation*}
and equivalently, referring to the geometrical discussion at the beginning of the present section,
\begin{align}
\label{eq:action}
    S_v=i \gamma \Biggl[\sum^{\text{ hetero.}}_{ab} j_{ab} \, \Re \theta_{N_a^{(1,\vec{0})}, N_b^{(\vec{0},1)}} &+   \sum^{\text{ para.}}_{ab} s_{ab} \, \theta_{N_a^{(\vec{0},1)}, N_b^{(\vec{0},1)}} \Biggr] \nonumber \\
    &+ i \Biggl[\sum^{\text{ hetero.}}_{ab} j_{ab}\, 2\,\text{arg}\, \frac{\lambda_{ab}}{\lambda_{ba}} +   \sum^{\text{ para.}}_{ab} s_{ab} \, \ln\left| \frac{\alpha_{ba}}{\alpha_{ab}}\right|^2  \Biggr]\,.
\end{align}
The reader may recognise in the first line of the previous equation the boundary area-Regge action associated to a single polytope\footnote{Note that the presence of $i$ in the action is necessary for the amplitude to be Lorentzian, since we did not include the imaginary unit explicitly in the argument of the exponential appearing in \eqref{eq:asympt}.}. There one sums over the dihedral angles and areas of polyhedra sharing a face. We remark again that, for the case of parachronal interfaces, the angle $\theta_{N_a^{(\vec{0},1)}, N_b^{(\vec{0},1)}}$ need not be strictly dihedral, insofar as the normals may not all point outward to the reconstructed polyhedra; this may be considered one more obstacle for the parachronal model. The second line is composed of an additional term that does not seem to have a direct geometrical meaning, and in that sense it may be considered a purely-quantum contribution to the asymptotics. The parachronal action once more agrees with the approximation in \cite{Liu2019}, and the heterochronal action recovers the result of \cite{Kaminski2018a}.

To complete the stationary phase analysis (to leading order) it remains to compute the factor in front of the exponentiated action for each critical point. In previous works \cite{Barrett2009a,Kaminski2018a} this final step is straightforward because the integrand satisfies certain regularity conditions, and in particular $k$-differentiability, where $k$ is a positive even integer. Hormander's theorem \cite{Hörmander2003} can be used, which generalizes the stationary phase approximation to higher dimensions. The pre-factor is given by the inverse square-root of the determinant of the Hessian matrix of the action, which gives the typical scaling behaviour of $\lambda^{-\frac{1}{2}}$ per integration variable, $\lambda$ being the uniform scaling parameter.

Unfortunately, due to the presence of branch points in the integrand, this same theorem cannot be applied to spin foam vertex amplitudes containing parachronal faces. Results for such cases exist in one dimension \cite{miller2006applied, Ciarkowski1989}, where the integrand contains branch points which are allowed to coalesce with critical points of the exponential function; compared to the more regular case, these integrals have a different scaling behaviour of $\lambda^{-\frac{1}{4}}$. However, to our knowledge, there is no theorem for general higher dimensional integrals. The authors of \cite{Liu2019} propose to overcome this issue by performing the stationary phase analysis iteratively, i.e. one variable at a time, effectively reducing the problem to a series of one-dimensional stationary phase analyses. To do so, it is necessary to permute the order of integrations, and this is straightforwardly possible if the integral absolutely converges. The authors argue that this is the case if one restricts the integration to a compact region around a critical point, and they have recently argued \cite{Hongguang_git} that the critical points of the iterative and variational analysis agree. This is indeed encouraging, and while we consider this proposal to be reasonable, we believe the matter requires a more thorough investigation: it must be shown that restricting the integration to compact regions around the critical points is a legitimate procedure, in particular in light of additional non-geometric critical points. Moreover, the iterative evaluation of each integral must be actually carried out to derive an expression for the asymptotic amplitude, and with it its overall scaling behaviour. We hope that the comparatively more concise expressions derived in this paper can facilitate such an analysis, leading eventually to a complete understanding of the asymptotic behaviour of the parachronal amplitude. 

\section{Discussion and outlook} \label{sec:Discussion}

In this article we have revisited the asymptotic expansion of the Lorentzian EPRL spin-foam model in the Conrady-Hnybida extension \cite{Conrady2010}. This extension lifts the original restriction of having only space-like polyhedra and polygons in the cellular decomposition, allowing instead also for time-like polyhedra with both time- and space-like faces. The asymptotic analysis for the CH extension has been studied previously in great detail \cite{Kaminski2018a,Liu2019}, and the present work both complements and expands the existing literature. In the following we briefly list the various results of the present article:
\begin{itemize}
    \item \textit{New parametrization of time-like faces}:
    Our key new result is a different definition of the $\text{SU}(1,1)$ coherent states for parachronal faces. Instead of the Conrady-Hnybida construction, where these states have real eigenvalues and are labelled by $\lambda = \sqrt{s^2 + \frac{1}{4}}$, we allow for generalised eigenstates with complex eigenvalues $\lambda = -\frac{i}{2} -s$. Thanks to this prescription, the respective inner products of the coherent states greatly simplify, such that we find a closed and exact formula for the integrand entering the asymptotic expansion. In this manner we avoid having to resort to an approximation for the parachronal amplitudes, as was done in \cite{Liu2019}.

    Nevertheless, we would like to emphasise that we understand this alternate prescription as a different parametrization of the \textit{same} model: we still choose coherent states that minimise the variance of the relevant Casimir operator, and simply expand the amplitude with a completeness relation constructed from these states. That the resulting amplitudes should still describe the same physics is supported by the form of the critical points of the amplitudes in our parametrization, which agree with previous results \cite{Liu2019}.
    
    \item \textit{Minkowski's theorem in Lorentzian signature}:
    A necessary ingredient for extending the asymptotic analysis of the Lorentzian EPRL model to arbitrary cell complexes is a Lorentzian version of Minkowski's theorem for three-dimensional polyhedra. To our best knowledge, a generalisation of Minkoski's theorem to polyhedra in Minkowski space did not exist in the literature. We have thus included in appendix \ref{appendix:geometry} a discussion of familiar geometrical concepts (e.g. orthogonality, half-spaces and angles\footnote{See also \cite{Sorkin:2019llw,Asante:2021zzh} for alternative definitions of Lorentzian dihedral angles.}) in Minkowski space, and presented a straightforward proof for the existence and uniqueness of Minkowskian polyhedra given a closure condition. We have also shown, as a complement, that such Minkoskian polyhedra are always rigid when convex. A proof for an analogous Minkowski theorem regarding higher-dimensional Minkowskian polytopes is still missing, but not strictly necessary for studying the spin-foam framework in 4 dimensions.

    As we have argued, the closure relations (for geometric critical points) obtained asymptotically induce 3-dimensional polyhedra, while the bivector equations encode how these are glued together into a 4-dimensional object. We remark once more that the faces of glued polyhedra must have the same area but not necessarily the same polyhedral shape, in line with previous results in Riemannian \cite{Bahr:2015gxa,Bahr:2016hwc,Dona:2017dvf} and Lorentzian (space-like polyhedra) models \cite{Dona2020}.

    \item \textit{Light-like vectors in closure condition and general polyhedra}:
    When taking into account parachronal interfaces in the asymptotic analysis, we find closure conditions \eqref{eq:closure1_p}, \eqref{eq:closure2_p} that, in addition to the expected time- and space-like vectors (on which the boundary coherent states are peaked) involve also null vectors. The presence of these null vectors hinders a reconstruction of polyhedra via Minkowski's theorem, and thus it constitutes an obstacle to recovering discrete general relativity from the EPRL model with the CH extension. However, the contribution of null vectors to the closure constraints is controlled by the complex factors $\beta_{ab}$, which are coefficients of spinors expanded in the $l_{ab}^\pm$ basis as in \eqref{eq:expand}. A restriction to geometrical configurations can thus be made by requiring of all $\beta_{ab}$ to vanish. This must be done by analysing both $\Re \beta_{ab}$ and $\Im \beta_{ab}$, which we now discuss. 
    
    In  \cite{Liu2019} the authors consider restricting the causal character of the model to time-like tetrahedra with one orthochronal and three parachronal faces, and show that in this case $\Re \beta_{ab} = 0$. Note that such a restriction forbids ``sliced triangulations'' as in causal dynamical triangulations \cite{Ambjorn:2012jv}. In contrast, we propose to lift these causal restrictions and instead assume that all $\tau_{ab}$ for parachronal faces are identical; as we have shown, this is sufficient to obtain $\Re \beta_{ab} = 0$. This approach allows us to generalise the argument of \cite{Liu2019} to polyhedra of valence higher than four and, crucially, to not restrict the causal character of the cellular decomposition. The choice of $\tau_{ab}$ for parachronal faces does not affect the geometry of the reconstructed polyhedra, and it does not enter into the formula for the associated action. 

    Regarding $\Im \beta_{ab}$, the authors of \cite{Liu2019} set it to zero, by arguing that the integrand of their amplitude does not depend on this factor and is thus considered ``gauge''. In our case the integrand \eqref{eq:pref_p} explicitly depends on $\Im \beta_{ab}$, yet in both approaches there is no reason to expect of configurations with $\Im \beta_{ab}=0$ to be sub-dominant. Indeed, multiple configurations with non-vanishing $\Im \beta_{ab}$, not necessarily allowing for reconstructing polyhedra, should still contribute to the amplitude.

    The impossibility of naturally restricting the dominant configurations to those that allow for reconstructing geometry indicate that the Lorentzian EPRL model in the Conrady-Hnybida extension is not sufficiently constrained. This is supplemented by the unexpected possibility that the choice of $\tau_{ab}$ may heavily influence the existence of additional critical points, despite its value \textit{not influencing} the geometrical (and thus physical) content of the coherent states entering the model. We leave the investigation of these observations for future research.
    
    \item \textit{Asymptotic formula and branch points}:
    We have not presented a closed formula for the asymptotic expansion. In agreement with the analysis in \cite{Liu2019}, the function \eqref{eq:pref_p_exp} appearing in the integrand for parachronal faces is such that its branch points of order $1/2$ coincide with the critical points of the amplitude at $\beta_{ab}=0$. On top of this there is the problem that the generic (not necessarily geometrical) critical configurations may not be isolated, but may rather constitute a critical hypersurface. Unfortunately, the stationary phase analysis of generic multidimensional integrals with multiple branch cuts and branch points coalescing with critical points is barely explored in the literature, and the present system requires a careful and thorough analysis. While there are heuristic indications suggesting that for geometric critical points the branch points might coalesce with roots in the integrand, we cannot comment further on this at this stage. In \cite{Liu2019}, the authors propose to perform the stationary phase analysis iteratively. For this to be viable, one must permute the order of integrations. To do so, the authors assume that the integral can be restricted to compact regions around the geometric critical points, in which the integral is absolutely convergent. The stationary phase approximation would then be performed one variable at a time. Recently \cite{Hongguang_git} the authors confirmed that,  in the context of their analysis \cite{Liu2019}, critical points in the iterative and variational procedures agree, which is encouraging. The explicit computations implementing this suggestion have  however not been carried out as of yet. We leave this intriguing question for another time.
    
    \item \textit{Dihedral angles at parachronal faces:}
    One further issue we must address is the one of dihedral angles at parachronal faces, i.e. the angle at time-like faces shared by time-like polyhedra. Since both polyhedra are time-like, their normals are space-like. Thus, these normals lie in an Euclidean plane and are related to one another by an $\text{SO}(2)$ rotation. The complex factors $\alpha_{ab}$ and $\alpha_{ba}$, figuring in formula \eqref{eq:action} for the action at the critical points, relate to the reconstructed normals to the polyhedra. There is however no way to ensure that these reconstructed normals are all outward or inward pointing, meaning that the angles obtained from $\alpha_{ab}$ and $\alpha_{ba}$ may not be always interpreted as dihedral. 

    The problem is absent for all other types of interfaces. For those cases the normals lie in a Minkowski plane, and the action at the critical points depends solely on the real part of the angle, which remains the same if a given normal is inverted.
\end{itemize}

Having listed our findings, we now turn towards a more speculative outlook for some of the points made above. Evidently, the mismatch between the asymptotic behaviour of the EPRL model and our semi-classical expectations, brought about by the extension to parachronal interfaces, constitutes a serious problem. It is interesting to note that the authors of \cite{Kaminski2018a} hinted at the possibility that such an issue could occur: for time-like tetrahedra and time-like interfaces, the causal character of rotation and boost generators can lead to cases in which the simplicity constraints cannot be satisfied. It is moreover the case, as pointed out originally by Conrady \cite{Conrady2010a}, that the master constraint, which is supposed to enforce simplicity, fails to do so already at the classical level; the problem is ultimately due to the properties of the spectrum of the $\text{SU}(1,1)$ Casimir operator. It would be interesting to investigate whether this is the cause for the observed non-geometrical critical points and how the model can be modified to suppress them.

%However, even if we restrict the model for parachronal faces by hand to just the geometrical critical points given by $\beta_{ab} = 0$, the asymptotic analysis is not complete. While the action evaluated on the critical points is known, the pre-factor determining the scaling behaviour of the amplitude is not. This difficulty is rooted in the existence of branch cuts in the integrand coalescing with the geometrical critical points; stationary phase analysis for such cases for more than one-dimensional integrals are scarce in the literature and thus requires careful analysis. We hope that the simpler form of the amplitude thanks to the parametrization presented in this paper can help facilitate such an analysis. {\color{red} The authors of \cite{Liu2019} propose to perform this analysis iteratively, such that one computes a one-dimensional integral after another. This idea requires the convergence of the integral (in a compact region around critical points, if only these critical points contribute) and has not been performed yet. In the light of additional (non-geometric) critical points, it will be interesting to investigate this proposition further.} Moreover, once the scaling behaviour is known it will be interesting to see whether geometric and non-geometric critical points have different scaling behaviour, determining their contribution in the semi-classical limit, which may provide further motivation for refining the construction of spin foam models with time-like polyhedra and polygons.

In a nutshell, spin-foams with time-like building blocks in the extended EPRL model lead to several promising results and intriguing features, yet also to a considerable number of open questions that must be addressed. Particularly concerning is the generic appearance of non-Regge-like critical points in addition to the anticipated Regge-like ones corresponding to the boundary data. This fact also obstructs the direct connection to effective spin-foam models \cite{Asante:2020qpa,Asante:2021zzh} (without further assumptions), which crucially rely on the asymptotic formula of spin-foams. Nevertheless, since Regge-like critical points are permitted in the asymptotic limit, the effective models could efficiently explore the ``desired'' part of the spin-foam path integral for large triangulations, and may be modified in the future once the correct scaling behaviour of the amplitudes is known. Last but not least, we note that better understanding these extended Lorentzian models might allow us to make more direct contact with other approaches to quantum gravity that heavily rely on the causal structure built into their construction, such as causal dynamical triangulations \cite{Ambjorn:2012jv} or causal set theory \cite{Surya:2019ndm}.

\section*{Acknowledgements}

The authors would like to thank Hongguang Liu for valuable input on an earlier version of this paper, and Leonardo Garc\'{i}a-Heveling for useful criticism of the appendix on polyhedral geometry in Minkowski space. J.D.S. benefited from insightful discussions with Maximilian H. Ruep on the subject of generalised eigenstates. S.St. would like to thank Benjamin Bahr and Bianca Dittrich for enlightening discussions on the results of this work and their interpretation.

J.D.S. and S.St. are funded by the Deutsche Forschungsgemeinschaft (DFG, German Research Foundation) - Projektnummer/project-number 422809950.

%\newpage
\appendix
\section{Generalized Eigenstates of a non-comapct \texorpdfstring{$\text{SU}(1,1)$}{SU(1,1)} generator}
\label{appendix:generalized}

The extension of the EPRL spin-foam model to time-like polyhedra obtained in \cite{Conrady2010} requires one to make use of certain generalised eigenstates \cite{Lindblad1970a} of a non-compact generator of $\text{SU}(1,1)$. With the purpose of slightly modifying these states in our analysis, we will start by reviewing the representation theory of $\text{SU}(1,1)$, and later clarify some points regarding the construction of such generalised states.

\subsection{Unitary irreducible representations of \texorpdfstring{$\text{SU}(1,1)$}{SU(1,1)}}

The algebra of $\text{SU}(1,1)$ is spanned by the generators $\{J_3,K_1,K_2\}$, defined with the standard Pauli matrices as $\{\sigma^3/2, i \sigma^1/2,i \sigma^2/2\}$ in the fundamental representation. Denoting for convenience $\vec{F}=\left(J_3,K_1,K_2\right)$, the Casimir element of the algebra reads $F^2=J_3^2-K_1^2-K_2^2$. There are two families of unitary irreducible representations labelled by eigenvalues of $Q$ and $J_3$, called the \textit{continuous} and \textit{discrete series}. In the case of the discrete series, the representation space $\mathcal{D}^{k, \alpha}$ is spanned by eigenstates
\begin{equation}
\begin{gathered}
F^2 \ket{k\, m}=k(k-1) \ket{k\,m}\,,\\
J_3  \ket{k\,m} = m  \ket{k\,m} \,,
\end{gathered}
\end{equation}
and it is labelled by a positive half-integer $k\in \mathbb{N}/2$ and by a sign $\alpha=\pm$. Possible $J_3$ eigenvalues are determined by $\alpha$ as
\begin{equation}
m=\alpha k,\, \alpha(k+1), \, \alpha(k+2), \, ... \,.
\end{equation}
For the continuous series, the Hilbert space $\mathcal{C}_s^\epsilon$ is spanned by the eigenstates
\begin{equation}
\begin{gathered}
F^2 \ket{s\, m}=j(j+1) \ket{s\,m}\,, \quad j=-\frac{1}{2}+is\,, \quad s \in \mathbb{R}_0^+\\
J_3  \ket{s\,m} = m  \ket{s\,m} \,,
\end{gathered}
\end{equation}
and it is labelled by a continuous value $s$ and by $\epsilon\in \{0,\frac{1}{2}\}$. The possible $J_3$ eigenvalues $m$ depend on $\epsilon$ as
\begin{equation}
m=\epsilon,\, \epsilon\pm1,\, \epsilon \pm2,\, ...\,.
\end{equation}
All the eigenstates described above, for both families of representations, are orthonormal. We refer the reader to \cite{Conrady2011} for the matrix elements of the generators in these bases. Finally, for both series of representations, there is an outer automorphism of the algebra $(J_3,K_1,K_2) \mapsto (-J_3, K_1, -K_2)=P F^i P^{-1}$, which for the continuous series can be realized as
\begin{equation}
    P \ket{s\,m}=e^{i \pi m} \ket{s,\,-m}\,.
\end{equation}

\subsection{Generalized eigenstates of \texorpdfstring{$K_1$}{K1} with complex eigenvalues}
\label{appendix:generalized_2}

Rather than constructing the representations of $\text{SU}(1,1)$ using the eigenstates of $J_3$, one may choose to define them from eigenstates of $K_1$ or $K_2$. This, however, introduces some additional subtleties due to the non-compactness of the generators, as they have no eigenstates in $\mathcal{C}_s^\epsilon$. Such representations were constructed in \cite{Lindblad1970a} using generalized eigenstates of $K_1$ through the machinery of Gelfland triples $\mathcal{D} \subset \mathcal{C}_s^\epsilon \subset \mathcal{D}'$, where $\mathcal{D}$ is a certain dense set in $\mathcal{C}_s^\epsilon$. The crux of the construction lies in the extension of $K_1$ to the containing space $\mathcal{D}'$ where the generalized eigenstates are defined, such that there exists a well-defined scalar product
\begin{equation}
\braket{f, \phi}=\overline{\braket{\phi, f}}
\end{equation}
between functions $\phi \in \mathcal D$ and distributions $f \in \mathcal{D}'$. One then shows that a nuclear spectral theorem holds, that is, one has eigenstates in $\mathcal{D}'$,
\begin{equation}
\label{eq:K1}
\begin{gathered}
    K_1 \ket{j\, \lambda \sigma}=\lambda \ket{j\, \lambda\, \sigma}\,, \\
    P \ket{j\, \lambda \sigma}=(-1)^\sigma \ket{j\, \lambda\, \sigma}\,,
\end{gathered}
\end{equation}
where $\lambda \in \mathbb{R}$ and $\sigma \in \{0,1\}$, such that they satisfy orthogonality and completeness\footnote{Note that, as is always the case with generalized states, an expression like $\sum_\sigma \int \diff \lambda\, \ket{j\, \lambda \, \sigma} \bra{j\, \lambda \, \sigma}=\mathbbm{1}_j$ is strictly meaningless, and has to be understood formally.}:
\begin{equation}
\label{eq:gen_orth}
\sum_m \braket{j\, \lambda \, \sigma | j\,m}\braket{j\, m | j\, \lambda' \, \sigma}=\delta(\lambda-\lambda')\,,
\end{equation}
\begin{equation}
\label{eq:gen_comp}
\sum_\sigma \int_\mathbb{R} \diff \lambda\, \braket{j\, m' | j\, \lambda \, \sigma} \braket{j\, \lambda \, \sigma | j\,m}=\delta_{m m'}\,.
\end{equation}
We now slightly add to this construction by making the straightforward observation that the eigenvalues in \eqref{eq:K1} can actually be taken to be complex. The author of \cite{Lindblad1970a} had already noted that the differential equations characterizing $\braket{j\, \lambda \, \sigma | j\,m}$ admit complex values of $\lambda$, but remarked that equations \eqref{eq:gen_orth} and \eqref{eq:gen_comp} would not hold in that case. It is however straightforward to show, by using the integral expression for the functions $\braket{j\, \lambda \, \sigma | j\,m}$ in \cite{Lindblad1970a}, that states associated to complex eigenvalues still satisfy the relations
\begin{equation}
\sum_m \braket{j\, \overline{\lambda} \, \sigma | j\,m}\braket{j\, m | j\, \lambda' \, \sigma}=\delta(\lambda-\lambda')\,,
\end{equation}
\begin{equation}
\sum_\sigma \int_{\mathbb{R}+i \alpha} \diff \lambda\, \braket{j\, m' | j\, \lambda \, \sigma} \braket{j\, \overline{\lambda} \, \sigma | j\,m}=\delta_{m m'}\,,
\end{equation}
where now $\lambda$ is constrained to have a constant imaginary part $i \alpha$. Quite remarkably, that these relations must now involve the complex conjugated eigenvalue follows from the self-adjointness of $K_1$ in $\mathcal{C}_s^\epsilon$. Denoting the distribution associated to a generalized eigenstate by $f_{\lambda}$, we find
\begin{align*}
    \braket{K_1\phi, \psi}_{\mathcal{C}_s^\epsilon}&=\sum_\sigma \int_{\mathbb{R}+i \alpha} \diff \lambda\;  \overline{K_1 f_{\overline{\lambda}}(\phi)} f_\lambda(\psi)\\
    &=\sum_\sigma \int_{\mathbb{R}+i \alpha} \diff \lambda \; \lambda \; \overline{ f_{\overline{\lambda}}(\phi)} f_\lambda(\psi) \\
    &=\braket{\phi, K_1 \psi}_{\mathcal{C}_s^\epsilon}\,,
\end{align*}
so that $K_1$ is self-adjoint, as necessary for consistency. Finally, a completeness relation for coherent states is immediately implied from the previous relations, since
\begin{align}
\int \diff g \; D^j(g)\ket{j \lambda \sigma}\bra{j \overline{\lambda'} \sigma} D^j(g)^\dagger&=\sum_{m, m', n, n'} \ket{j m}\bra{j m'} \int \diff g \; D^j_{ m n}(g) \overline{D^j_{ m' n'}(g)}\; \braket{j n|j \lambda \sigma} \braket{j \overline{\lambda'} \sigma | j n'} \nonumber \\
&= \sum_{m} \ket{j m}\bra{j m}\; \sum_n \braket{j \overline{\lambda'} \sigma|j n} \braket{j n|j \lambda \sigma} \mu_\epsilon(s)^{-1} \nonumber \\
&= \frac{\mathbbm{1}_j \, \delta(\lambda-\lambda')}{\mu_\epsilon(s)}\,,
\end{align}
where in the second line we used the orthogonality of $\text{SU}(1,1)$ matrix coefficients (see section \ref{appendix:principal_series}).

A formula for the matrix coefficients $D^{j}_{m \lambda \sigma}(v)=\braket{j,m| D^{j}(v)|j,\lambda, \sigma}$ in the eigenbasis of $K_1$ can be found in \cite{Lindblad1970}, and reads
\begin{equation*}
\label{coef11cont}
\begin{gathered}
D^{j}_{m \lambda \sigma}(v)=S^j_{m \lambda \sigma}\left(T^j_{m \lambda} F^j_{m \lambda}(v) -(-1)^\sigma T^j_{-m \lambda} F^j_{-m \lambda}(\bar{v})\right) \,, \\
S^j_{m \lambda \sigma}=\sqrt{\frac{\Gamma(m-j)}{\Gamma(m+j+1)}}\frac{2^{j-1}\Gamma(-j-i \lambda)}{i^\sigma\sin\left(\frac{\pi}{2}(-j+i \lambda+\sigma)\right)}\,, \quad T^j_{m\lambda}=\frac{1}{\Gamma(-m-j)\Gamma(m-i\lambda+1)}\,,
\end{gathered}
\end{equation*}
\begin{align}
F^j_{m \lambda}(v)=2^{-m}&(v_1+v_2)^{m-i\lambda}(v_1-v_2)^{m+ i\lambda} \cdot \nonumber \\
&\cdot {}_2F_1\left(-j+m,m+j+1,m-i \lambda+1,\frac{(\bar{v_1}-\bar{v_2})({v_1}+{v_2})}{2}\right)\,,
\end{align}
where $_2 F_1$ denotes the hypergeometric function.

\subsection{A useful expression for the coefficients at \texorpdfstring{$\lambda=\overline{ij}$}{conjugate eigenvalues}}
\label{appendix:generalized_3}

We close this section with a derivation of an expression for the coefficients $\braket{j,m|D^j(g)|j,\overline{ij},\sigma}$. According to \cite{Lindblad1970a}, there exists a ladder operator
\begin{equation}
\label{eq:ladder}
F^+ \ket{j,\lambda,\sigma}=i(-j-1+i \lambda) \ket{j,\lambda+i,\hat{\sigma}}\,,
\end{equation}
taking $\lambda$ to $\lambda+i$ and $\sigma$ to $\hat{\sigma}=(\sigma+1)_{\text{mod}\,2}$, with $F^+=J_3+K_2$. Thus we have that
\begin{equation}
\ket{j,\overline{ij},\sigma}=\frac{i}{2j+1} F^+ \ket{j,ij,\hat{\sigma}}\,,
\end{equation}
and the matrix coefficients then read
\begin{align}
    \braket{j,m|D^j(g)|j,\overline{ij},\sigma}&=\frac{i}{2j+1} \braket{j,m|D^j(g) F^+|j,ij,\hat{\sigma}} \nonumber\\
    &=\frac{i}{2j+1} \braket{j,m|D^j(g) (D'^j(J_3)+D'^j(K_2))|j,ij,\hat{\sigma}} \nonumber\\
    &=\frac{1}{2j+1}\frac{\diff}{\diff t} \Bigr \rvert_{t=0} \left[D^j_{m,ij,\hat{\sigma}}(g e^{i t J_3}) + D^j_{m,i j,\hat{\sigma}}(g e^{i t K_2})\right]\,.
\end{align}
In turn, the coefficients $D^j_{m,\lambda,\sigma}(g)$ considerably simplify at $\lambda=i j$, as indeed we find
\begin{equation}
    D^j_{m,ij,\sigma}(g)=A^j_{m,\sigma} \left(\frac{g_1+g_2}{\sqrt{2}}\right)^{j+m} \left(\frac{\overline{ g_1+g_2}}{\sqrt{2}}\right)^{j-m}\,,
\end{equation}
with $A^j_{m,\sigma}=S^j_{m,ij,\sigma}\left(T^j_{m,ij}-(-1)^\sigma T^j_{-m,ij}\right)$. Hence the dual coefficients\footnote{We name these coefficients ``dual'' because they are calculated at $\overline{\lambda}$.} can be written as
\begin{align}
    D^j_{m,\overline{ij},\sigma}(g)=\frac{i}{2j+1} A^j_{m,\hat{\sigma}} \left(\frac{g_1+g_2}{\sqrt{2}}\right)^{j+m} &\left(\frac{\overline{ g_1+g_2}}{\sqrt{2}}\right)^{j-m} \cdot \nonumber \\
   & \cdot \left[(j+m)\frac{g_1-g_2}{g_1+g_2}-(j-m)\frac{\overline{g_1-g_2}}{\overline{g_1+g_2}} \right]\,.
\end{align}

\section{Unitary irreducible representations of \texorpdfstring{$\text{SL}(2,\mathbb{C})$}{SL(2,C)}}
\label{appendix:principal_series}
We collect in this section general results concerning the representation theory of the special linear group. Recall that the algebra $\mathfrak{sl}(2,\mathbb{C})$ is spanned by the generators $J^i=\sigma^i/2$ and $K^i=i \sigma^i/2$ in the fundamental representation, with the commutators
\begin{equation}
    [J^i,J^j]=i \epsilon^{ijk} J^k\,, \quad [K^i,K^j]=-i \epsilon^{ijk} K^k\,, \quad [J^i,K^j]=i \epsilon^{ijk} K^k\,.
\end{equation}
The representations of $\text{SL}(2,\mathbb{C})$ are constructed on the space $\mathcal{D}_\chi$ of homogeneous functions of two complex variables \cite{Gelfand1963},
\begin{equation}
\begin{gathered}
F: \mathbb{C}^2 \setminus \{{0}\} \rightarrow \mathbb{C} \\
F(\alpha z_1, \alpha z_2)=\alpha^{n_1 -1} \bar{\alpha}^{n_2 -1} F(z_1, z_2) , \, \alpha \in \mathbb{C}\, ,
\end{gathered}
\end{equation}
characterized by two parameters $(n_1,n_2) \in \mathbb{C}^2$. Here we focus on the so-called \textit{principal series representations}, where one restricts to the case $n_1=\overline{n}_2$. For convenience it is customary to redefine $n_1= (-n +i \rho)/2$, $n_2= (n+ i \rho)/2$\footnote{This is the notation used in \cite{ruhl1970the}, but other choices do exist in the literature \cite{Gelfand1963,Naimark2014}}, with $n \in \mathbb{Z}$ and $\rho \in \mathbb{R}$, and to collect both labels in $\chi=(n, \rho)$. Such representations act on $\mathcal{D}_\chi$ by the group's natural action on $\mathbb{C}^2$, that is
\begin{equation}
\label{rep}
\begin{gathered}
D^\chi: \text{SL}(2,\mathbb{C}) \rightarrow U(\mathcal{D}_\chi) \\
D^\chi(g) F(\mathbf{z}) = F(g^T \mathbf{z}) \,,
\end{gathered}
\end{equation}
and $D^\chi$ is indeed unitary and irreducible under the inner product
\begin{equation}
\begin{gathered}
\label{eq:innerproduct}
\langle F_1, F_2 \rangle =  \int_{\mathbb{C}P} \bar{F_1}(\mathbf{z}) F_2(\mathbf{z}) \, \omega\,, \\
\omega =\frac{i}{2} (z_2 \diff z_1 -z_1 \diff z_2) \wedge (\bar{z_2} \diff \bar{z_1} -\bar{z_1} \diff \bar{z_2}) \, ,
\end{gathered}
\end{equation}
where $\omega$ is the standard volume form on $\mathbb{C}^2$. The integral is to be calculated over a section $(z_1,z_2) \mapsto (z_1/z_2,1)$ of the bundle $\mathbb{C}^{2*} \rightarrow \mathbb{C}P$.

It turns out that not all representations labeled by $\chi$ are truly independent,  as the representations $\chi$ and $-\chi$ of the principal series can be shown to be equivalent. There thus exists an intertwining map
\begin{equation}
\label{eq:intertwiner}
    \begin{gathered}
        \mathcal{A}:\; \mathcal{D}_{\chi} \rightarrow \mathcal{D}_{-\chi} \\
        \mathcal{A} D^{-\chi}(g)=D^{\chi}(g) \mathcal{A}\,,
    \end{gathered}
\end{equation}
preserving the inner product. According to \cite{Gelfand1963}, this intertwiner can be used to construct a bilinear form $(\cdot,\cdot)$ in $\mathcal{D}_{\chi}$ as
\begin{equation}
    (F_1,F_2)= \braket{\mathcal{J}F_1,F_2}\,,
\end{equation}
where we define $\mathcal{J}F=\overline{\mathcal{A}F}$.

In order to make calculations more manageable, one might like to introduce orthonormal bases for the functions in $\mathcal{D}_\chi$. Two particularly useful realizations of these functions are provided by the so-called canonical and pseudo-bases, obtained from representations of the compact and non-compact subgroups $\text{SU}(2)$ and $\text{SU}(1,1)$, respectively. We shall go over their construction in the subsequent sections, following \cite{ruhl1970the}.

\subsection{The canonical basis}

Due to the homogeneity property of $F \in \mathcal{D}_{\chi}$, such functions can be uniquely characterised by their values on $S^3=\{\mathbf{z} \in \mathbb{C}^2 \; | \;|z_1|^2 + |z_2|^2=1\}$ as
\begin{equation}
F(z_1,z_2)=(|z_1|^2 + |z_2|^2)^{i \rho/2 -1} F\left(\frac{z_1}{\sqrt{|z_1|^2 + |z_2|^2}}, \frac{z_2}{\sqrt{|z_1|^2 + |z_2|^2}}\right) \,,
\end{equation}
and this allows us to realize $F$ on the unitary group. Indeed, using the well-known diffeomorphism between the sphere and $\text{SU}(2)$\footnote{We choose to associate $z_1, z_2$ to the conjugates of $u_1,u_2$ rather than the non-conjugated variables so that in equation \eqref{cov} we get $m=n/2$ rather than its symmetric.},
\begin{equation}
u=\begin{pmatrix}
u_1 && u_2 \\
-\bar{u_2} && \bar{u_1}
\end{pmatrix}\in \text{SU}(2) \,; \quad -\bar{u_2}= \frac{z_2}{\sqrt{|z_1|^2 + |z_2|^2}}\,,  \bar{u_1}= \frac{z_1}{\sqrt{|z_1|^2 + |z_2|^2}} \,,
\end{equation}
$F$ may just as well be understood as a function $f \in \mathcal{C}^\infty(\text{SU}(2))$ such that $f(u(\mathbf{z}))=F(\mathbf{z}/|\mathbf{z}|)$. Of course, $F$ still has to satisfy the homogeneity condition on the sphere under multiplication of the arguments by a norm-1 complex number,
\begin{equation}
\begin{gathered}
F(e^{i \omega} z_1,e^{i \omega} z_2 )=e^{-i \omega n} F(z_1,z_2) \\
f(\gamma u)= e^{i \omega n} f(u)\,, \quad \gamma=
\begin{pmatrix}
e^{i \omega} && 0 \\
0 && e^{-i \omega}
\end{pmatrix}\,,
\end{gathered}
\end{equation}
with $\omega \in \mathbb{R}$. To functions $f \in \mathcal{C}^\infty(\text{SU}(2))$ satisfying this transformation property we will call \textit{covariant}, as in \cite{ruhl1970the}. It then turns out that one has a Plancherel theorem identifying the Hilbert space of square integrable covariant functions  and $\mathcal{D}_\chi$:
\begin{equation}
\begin{gathered}
L^2(\text{SU}(2))_\text{cov} \simeq \mathcal{D}_\chi \,, \\
\int_{\text{SU}(2)} |f(u)|^2 \,\diff u = \int_{\mathbb{C}P} |F(z_1,z_2)|^2 \, \omega \,.
\end{gathered}
\end{equation}
Consequently, one may equivalently describe the Hilbert space of $\text{SL}(2,\mathbb{C})$ principal series representations in terms of $\text{SU}(2)$ states. Since, by the Peter-Weyl theorem for compact groups, the set
\begin{equation}
\{\sqrt{\text{dim}(D^j)}D^j_{m m'} \; | \; j \in \Lambda\}\,,
\end{equation}
where $\Lambda$ denotes the set of unitary irreducible representations of SU(2) and $D^j_{m m' }$ the matrix elements of the representation, is an orthonormal basis for $L^2(\text{SU}(2))$, one may restrict this set to covariant representation functions to construct a basis for $\mathcal{D}_\chi$. Noting that $\gamma$ is an element of the one-parameter subgroup associated with the $J_3$ generator of $\text{SU}(2)$, and considering the matrix elements of the fundamental irreducible SU(2) representations under $J_3$ eigenstates,
\begin{align}
\label{cov}
D^j_{m m'}(\gamma u) &= \sum_lD^j_{m l}( e^{2 i \omega J_3}) D^j_{l m'}(u) \nonumber \\
&= e^{2 i m\omega} D^j_{m m'}( u)
\end{align}
we find that $D^j_{m m'}$ is covariant when $m=n/2$. A general unit-norm representation function of $\text{SL}(2,\mathbb{C})$ can therefore be written as
\begin{equation}
\label{eq:func2}
F_{j,m}^{\chi}(\mathbf{z})=\sqrt{2j+1}\,(|z_1|^2 + |z_2|^2)^{i \rho/2 -1}D^j_{\frac{n}{2} m}(u(\mathbf{z}))\,,
\end{equation}
and this defines the canonical basis. For convenience we may also introduce a Dirac notation for this basis and the inner product \eqref{eq:innerproduct},
\begin{equation}
F_{j,m}^{\chi} \rightarrow \ket{\chi; j\, m}\,,
\end{equation}
allowing us to write a simple resolution of identity
\begin{equation}
\mathbbm{1} _{(n,\rho)}=\sum_{j=\frac{n}{2}}^\infty \sum_{m=-j}^{m=j} \ket{\chi; j \,m}\bra{\chi; j \,m}\,.
\end{equation}
%Finally, we may write explicitly an arbitrary function in $\mathcal{D}_\chi$. General $\text{SU}(2)$ matrix elements are given by \cite{barrett_lorentzian_2010}
%\begin{align}
%\label{coef2}
%D^j_{m m'}(u)=&\sqrt{\frac{\Gamma(j-m+1) \Gamma(j+m+1)}{\Gamma(j-m'+1) \Gamma(j+m'+1)}} \sum_{n} \binom{m+m'}{n} \binom{j-m'}{j+m-n}  \;\cdot \nonumber \\
%\cdot \;&  u_1^n \, u_2^{j+m-n} \,(-\overline{u}_2)^{j+m'-n} \, \overline{u}_1^{n-m-m'} \,,
%\end{align}
%and in this basis a general unit-norm SL$(2,\mathbb{C})$ representation function will have the form
%\begin{equation}
%\label{func2}
%F_{m}^j(\mathbf{z})=\sqrt{2j+1}\,(|z_1|^2 + |z_2|^2)^{i \rho/2 -1}D^j_{\frac{n}{2} m}(u(\mathbf{z}))\,.
%\end{equation}

\subsection{The pseudo-basis}

In complete analogy to the previous SU(2) case, homogeneous functions in $\mathcal{D}_\chi$ are uniquely defined by their values on the hyperboloids (or pseudo-spheres) $H_\pm^3=\{\mathbf{z} \in \mathbb{C}^2 \; | \;|z_1|^2 - |z_2|^2=\tau\,,\tau=\pm1\}$ through
\begin{align}
F(z_1,z_2)=\sum_\tau \Theta\left(\tau(|z_1|^2 - |z_2|^2)\right) & \left(\tau(|z_1|^2 - |z_2|^2)\right)^{i \rho/2 -1} \cdot \nonumber \\
& \cdot F\left(\frac{z_1}{\sqrt{\tau(|z_1|^2 - |z_2|^2)}}, \frac{z_2}{\sqrt{\tau(|z_1|^2 - |z_2|^2)}}\right) \,,
\end{align}
where $\Theta$ is the Heaviside function. Just as before, $F$ can be understood as a function $f \in \mathcal{C}^\infty(\text{SU}(1,1))$ through an association of the hyperboloids with the non-compact group. This correspondence depends on $\tau$ as follows:
\begin{equation}
\begin{gathered}
v=\begin{pmatrix}
v_1 && v_2 \\
\bar{v_2} && \bar{v_1}
\end{pmatrix}\in \text{SU}(1,1) \,; \\
\tau=1: \quad \bar{v_2}=\frac{z_2}{\sqrt{|z_1|^2 - |z_2|^2}}\,, \bar{v_1}=\frac{z_1}{\sqrt{|z_1|^2 - |z_2|^2}} \\
\tau=-1: \quad v_1=\frac{z_2}{\sqrt{|z_2|^2 - |z_1|^2}}\,, v_2=\frac{z_1}{\sqrt{|z_2|^2 - |z_1|^2}} \,.
\end{gathered}
\end{equation}
One may once more prove a Plancherel theorem \cite{ruhl1970the}
\begin{equation}
\begin{gathered}
\label{ruhlplancherel}
L^2(\text{SU}(1,1))_\text{cov} \oplus L^2(\text{SU}(1,1))_\text{cov} \simeq \mathcal{D}_\chi \,, \\
\sum_\tau \int_{\text{SU}(1,1)} |f_\tau(v)|^2 \,\diff v =  \int_{\mathbb{C}P} |F(z_1,z_2)|^2 \, \omega \,, \\
F(\mathbf{z}/|\mathbf{z}|)=\sum_\tau \Theta\left(\tau(|z_1|^2 - |z_2|^2)\right)f_\tau( v_\tau(\mathbf{z}))\,,
\end{gathered}
\end{equation}

where now, unlike the case for $\text{SU}(2)$, the space $\mathcal{D}_\chi$ is isomorphic to \textit{two} copies of $L^2(\text{SU}(1,1))_\text{cov}$, labelled by $\tau$. Crucially, both $f_\tau$ will still need to satisfy the covariance condition $f_\tau(\gamma v)=e^{i \omega n \tau} f_\tau(v)$.

As in the case in the previous subsection, one would like to have a description of $\mathcal{D}_\chi$ in terms of unitary irreducible representations of $\text{SU}(1,1)$. It turns out that, although the Peter-Weyl theorem is not applicable due to the non-compactness of the group, one still has another Plancherel theorem relating exactly those representations with $L^2(\text{SU}(1,1))$ \cite{HarishChandra1952,Sciarrino1967}. Indeed there exists an isomorphism involving both the discrete and continuous series,
\begin{equation}
\label{plan11}
\begin{gathered}
\bigoplus_k \mathcal{D}^+_k \bigoplus_k \mathcal{D}^-_k \bigoplus_\epsilon \int^\oplus \diff s \; \mathcal{C}^\epsilon_s \simeq L^2(\text{SU}(1,1))\,, \\
\sum_{m m'} \left[\sum_\epsilon \int_0^\infty \diff s\;  |\psi^{s,\epsilon}_{m m'}|^2 + \sum_k (|\psi^{k+}_{m m'}|^2+|\psi^{k-}_{m m'}|^2)\right] = \int_{\text{SU}(1,1)} |f(v)|^2 \,\diff v\,,
\end{gathered}
\end{equation}
where the various $\psi$ are defined as follows:
\begin{equation}
\begin{gathered}
\psi^{s,\epsilon}_{m m'}=\int_{\text{SU}(1,1)}\diff v\, \sqrt{\mu_\epsilon(s)}\; \overline{D}^{s,\epsilon}_{m m'}(v) \, f(v) \; \diff v\,, \quad \mu_\epsilon(s)=\begin{cases}
2s \tanh(\pi s)\,,\quad \epsilon=0 \\
2s \coth(\pi s)\,,\quad \epsilon=\frac{1}{2}
\end{cases} \\
\psi^{k,\alpha}_{m m'}= \int_{\text{SU}(1,1)}\diff v\, \sqrt{2k-1}\; \overline{D}^{k,\alpha}_{m m'}(v)\, f(v) \,, \quad \alpha=\pm1\,.
\end{gathered}
\end{equation}
Using the above map, one may now describe $\mathcal{D}_\chi$ in terms of unitary irreducible representations of $\text{SU}(1,1)$. According to the previous equations, the space of homogeneous functions should be isomorphic to two copies of the Hilbert space on the left-hand side of equation \eqref{plan11}, constrained to satisfy covariance. Through a similar argument as the one used in \eqref{cov}, one may check that the representation functions of both the continuous and discrete series are constrained to $m=\tau \frac{n}{2}$ and that, among the discrete series representations, only those labelled by $\alpha=\pm$ contribute to the expansion of $f_\pm$. We may thus unequivocally set $\alpha=\tau$. A general unit-norm representation function of $\text{SL}(2,\mathbb{C})$ reads therefore,
\begin{equation}
\label{eq:func11}
\begin{gathered}
F^{\chi,\tau}_{s,\epsilon, m}(\mathbf{z})=  \sqrt{\mu_\epsilon(s)} \; \Theta\left(\tau(|z_1|^2 - |z_2|^2)\right)\left(\tau(|z_1|^2 - |z_2|^2)\right)^{i \rho/2 -1} D^{s,\epsilon}_{\frac{\tau n}{2}, m}( v_\tau(\mathbf{z}))\,, \\
F^{\chi, \tau}_{k,m}(\mathbf{z})=\sqrt{2k-1}\;\Theta\left(\tau(|z_1|^2 - |z_2|^2)\right)\left(\tau(|z_1|^2 - |z_2|^2)\right)^{i \rho/2 -1}D^{k,\tau}_{\frac{\tau n}{2},m}( v_\tau(\mathbf{z}))\,,
\end{gathered}
\end{equation}
for the continuous and discrete series, respectively. This defines the pseudo-basis of $\mathcal{D}_\chi$. Finally, introducing once more the Dirac notation
\begin{align}
&F^{\chi,\tau}_{s,\epsilon, m} \rightarrow \ket{\chi,\tau; s,\epsilon,m}\,, \\
&F^{\chi,\tau}_{k, m} \rightarrow \ket{\chi,\tau; k,m}\,,
\end{align}
we may write a resolution of identity as
\begin{equation}
\mathbbm{1} _{(n,\rho)}=\sum_\tau \biggl[\int_0^\infty \diff s \sum_{\substack{\pm m=\epsilon \\\epsilon-\frac{n}{2}\in \mathbb{Z}}} \ket{\chi,\tau; s,\epsilon,m}\bra{\chi,\tau; s,\epsilon,m} +  \sum_{k-\frac{n}{2}\in \mathbb{Z}}\sum_{m=\tau k}^{\tau \infty} \ket{\chi,\tau; k,m}  \bra{\chi,\tau; k,m} \biggr]\,.
\end{equation}

\section{Geometrical description of \texorpdfstring{$\text{SU}(2)$}{SU(2)} and \texorpdfstring{$\text{SU}(1,1)$}{SU(1,1)}}
\label{appendix:geometrygroups}
We would like to make some remarks regarding the relationship of the unitary groups with the spaces of unit-norm vectors in 3-dimensional subspaces of $\mathbb{R}^{3,1}$: the sphere $S^2$ and the one- and two-sheeted hyperboloids $H^\text{sp}$ and $H^\pm$. These are the so-called \textit{surfaces of transitivity} \cite{gelʹfand2016generalized}, or \textit{homogeneous spaces} \cite{gelʹfand2016generalized}, of those groups. For later explicit computations it will be useful to have at hand parameterizations of both $\text{SU}(2)$ and $\text{SU}(1,1)$, which we collect here. Elements $u$ of $\text{SU}(2)$ can be parameterized by \cite{Carmeli_2000}
\begin{equation}
    u=e^{i J_3 \phi} e^{i J_2 \theta} e^{i J_3 \psi}\,, \quad 0\leq \phi,\psi \leq 2\pi\,, \; 0\leq \theta \leq \pi\,,
\end{equation}
while elements $v$ of $\text{SU}(1,1)$ will be parametrized in one of two ways, by either \cite{Lindblad1970}
\begin{equation}
    v=e^{i J_3 \phi} e^{i K_2 t} e^{i K_1 u}\,, \quad 0\leq \phi \leq 2\pi\,, \; -\infty \leq u,t \leq \infty\,,
\end{equation}
or \cite{Inomata_1992}
\begin{equation}
v=e^{i J_3 \phi} e^{i K_2 t} e^{i J_3 \psi}\,, \quad 0\leq \phi,\psi \leq 2\pi\,, \; 0 \leq u,t \leq \infty\,,
\end{equation}
depending on convenience.

\subsection{The Hopf fibration of the sphere}
\label{hopfsubsection}
As implied above, the group SU(2) is diffeomorphic to the sphere $S^3$, and it is also the double-cover of the rotation group SO$(3) \simeq S^2$, on which it acts transitively with U(1) as stabiliser. That it is so can be seen by first considering the $\mathfrak{su}(2)$ Lie algebra and the group's adjoint action on it,
\begin{equation}
\label{adjoint}
\begin{gathered}
\text{Ad}: \text{SU}(2) \times \mathfrak{su}(2) \rightarrow \mathfrak{su}(2) \\
(g, \, X) \mapsto g^{-1}Xg \,,
\end{gathered}
\end{equation}
which is clearly an isometry under the $\mathfrak{su}(2)$ inner product $\braket{X \, | \, Y}=2 \, \text{Tr}(X Y)$. Given that $\mathfrak{su}(2)$ is naturally isomorphic to $\mathbb{R}^3$, one may then make use of this algebra automorphism to generate all points on the sphere $S^2$, which is precisely the set of all 3-vectors of fixed norm.
To this end, we shall define a map from the set $\mathcal{B}\left(\mathfrak{su}(2)\right)$ of orthonormal bases of the algebra to the sphere. Denoting by $\ket{\pm}$ the usual $(1,0), \, (0,1)$ eigenvectors of $J_3$, such a mapping may be written as
\begin{equation}
\begin{gathered}
h_\pm: \mathcal{B}\left(\mathfrak{su}(2)\right) \rightarrow S^2 \subset \mathbb{R}^3 \\
 \{X_i\} \mapsto
\braket{\pm \,|\, X_i \cdot \pm} \hat{e}^i \,,
\end{gathered}
\end{equation}
where we make use of the standard inner product on $\mathbb{C}^2$.
Note in particular that the canonical basis $\{J_i\}$ has as image the vector $\pm \frac{1}{2}(0,\,0,\,1)\in S^2$. It is not hard to see, through the $\mathfrak{su}(2) \simeq \mathbb{R}^3$ isomorphism, that $h_\pm$ is injective up to transformations of the canonical basis preserving $J_3$.

Now, since the automorphism \eqref{adjoint} is an orientation-preserving isometry, we may use it to generate other orthonormal bases from the canonical one, and this establishes a mapping
\begin{equation}
\begin{gathered}
\pi_\pm: \text{SU}(2) \rightarrow S^2 \subset \mathbb{R}^3 \\
g \mapsto h_\pm \circ \, \text{Ad}\left(g,\{J_i\}\right) \,, \\
\end{gathered}
\end{equation}
projecting from $\text{SU}(2)\simeq S^3$ to $S^2$. As remarked before $\pi_\pm (ge^{i J_3 \psi})=\pi_\pm (g)$, and indeed the pre-images of the map are unique up to a $\text{U}(1)$ circle. Therefore, restricting the domain of $\pi_\pm$ to the subgroup $\text{SU}(2)/\text{U}(1)$, we find an injective map
\begin{equation}
\begin{gathered}
\pi_\pm: \text{SU}(2)/\text{U}(1) \rightarrow S^2 \subset \mathbb{R}^3 \\
g \mapsto h_\pm \circ \text{Ad}(g,\{J_i\})\,, \\
g=e^{i J_3 \phi} e^{i J_2 \theta} \,,
\end{gathered}
\end{equation}
\begin{align}
\label{su2vec}
h_\pm \circ \text{Ad}\left(g,\{J_i\}\right)&= \braket{\pm \,|\, g^\dagger J_i g \cdot \pm} \hat{e}^i \nonumber\\
&=\pm \frac{1}{2}\left(-\sin{\theta} \cos{\phi},\,\sin{\theta} \sin{\phi},\, \cos{\theta}\right)\,,
\end{align}
which can clearly be seen to be smooth and surjective. We have thus recovered the well-known Hopf fibration establishing the diffeomorphism $\text{SU}(2)/\text{U}(1) \simeq S^2$.

\subsection{The 2-sheeted hyperboloid analogue}

One may follow exactly the same arguments as above to construct a diffeomorphism from a subgroup of $\text{SU}(1,1)$ to the two-sheeted hyperboloid $H^\pm$. To do so, we consider again the adjoint action, this time for $\text{SU}(1,1)$,
\begin{equation}
\label{adjoint11}
\begin{gathered}
\text{Ad}: \text{SU}(1,1) \times \mathfrak{su}(1,1) \rightarrow \mathfrak{su}(1,1) \\
(g, \, X) \mapsto g^{-1}Xg=\sigma_3 g^\dagger {\sigma_3}Xg \,,
\end{gathered}
\end{equation}
where ${\sigma_3}=\text{diag}(1,-1)$. Again, this action is an isometry\footnote{Note that the defining property of $\text{SU}(1,1)$ is that $g^\dagger {\sigma_3} g = {\sigma_3}$.} under the $\mathfrak{su}(1,1)$ inner product $\braket{X \, | \, Y}=2\, \text{Tr}\left(X Y\right)$.
As before, different bases of $\mathfrak{su}(1,1)\simeq \mathbb{R}^{2,1}$ will be related by orientation-preserving isometries, i.e. by rotations. Once more we define the correspondence between algebra bases and fixed-norm \textit{time-like} vectors in $\mathbb{R}^{2,1}$,
\begin{equation}
\begin{gathered}
h_\pm: \mathcal{B}\left(\mathfrak{su}(1,1)\right) \rightarrow H^\pm \subset \mathbb{R}^{2,1} \\
\{X_i\} \mapsto
\braket{\pm \,|\, X_i \cdot \pm} \hat{e}^i \,,
\end{gathered}
\end{equation}
and the same construction as above establishes the mapping
\begin{equation}
\begin{gathered}
\pi_\pm: \text{SU}(1,1)/\text{U}(1) \rightarrow H^\pm \subset \mathbb{R}^{2,1} \\
g \mapsto h_\pm \circ \text{Ad}\left(g,\{F_i\}\right)\,, \\
g=e^{i J_3 \phi} e^{i K_2 t} \,, \quad 0 \leq \phi < 2\pi,\,\, 0\leq t<  \infty\,, \\
\end{gathered}
\end{equation}
\begin{align}
\label{eq:su11discvec}
 h_\pm \circ \text{Ad}\left(g,\{F_i\}\right)&= \braket{\pm \,|\,{\sigma_3 g^\dagger {\sigma_3} \, F_i g \,} \cdot \pm} \hat{e}^i \nonumber\\
 &= \pm \braket{\pm \,|\, g^\dagger {\sigma_3} \, F_i g \,\cdot \pm} \hat{e}^i \nonumber \\
 &=\pm \frac{1}{2}\left( \cosh{t},\, \sinh{t} \cos{\phi},\,-\sinh{t} \sin{\phi}\right)\,.
 \end{align}
 Note that the hyperboloids are oriented along $\hat{e}^1$. We point out that the inner product appearing in the second line of the previous equation can be writen as $\pm\braket{g \cdot \pm, \, F_ig \cdot \pm}_{\sigma_3}$ using the $\text{SU}(1,1)$-invariant inner product in $\mathbb{C}^2$ defined as $\braket{u,\,v}_{\sigma_3}=u^\dagger {\sigma_3} \, v$.
\subsection{The 1-sheeted hyperboloid analogue}

 We may also establish the relationship between the 1-sheeted hyperboloid $H^\text{sp}$ in $\mathbb{R}^{3,1}$ and the subgroup $\text{SU}(1,1)/(\mathbb{Z}_2 \,e^{i u K_1})$. We still use the adjoint action from equation \eqref{adjoint11}, but this time we make a different assignment by using the eigenstates of $K_1$,
 \begin{equation}
 \ket{l^+}=\frac{1}{\sqrt{2}}\left(\ket{+}+\ket{-}\right)\,, \quad  \ket{l^-}=\frac{1}{\sqrt{2}}\left(\ket{+}-\ket{-}\right)\,,
 \end{equation}
 such that the mapping to fixed-norm \textit{space-like} vectors is given by
\begin{equation}
\begin{gathered}
h_\pm: \mathcal{B}\left(\mathfrak{su}(1,1)\right) \rightarrow H^\pm \subset \mathbb{R}^{2,1} \\
\{X_i\} \mapsto
\braket{l^\pm \,|\, X_i \cdot l^\pm} \hat{e}^i \,.
\end{gathered}
\end{equation}
The diffeomorphism between the hyperboloid and the subgroup takes the form
\begin{equation}
\begin{gathered}
\pi_\pm: \text{SU}(1,1)/(\mathbb{Z}_2 \, e^{i u K_1}) \rightarrow H^\text{sp} \subset \mathbb{R}^{2,1} \\
g \mapsto h_\pm \circ \text{Ad}\left(g,\{F_i\}\right)\,, \\
g=e^{i J_3 \phi} e^{i K_2 t} \,,\quad 0 \leq \phi < 2\pi,\,\, - \infty < t<  \infty\,, \\
\end{gathered}
\end{equation}
\begin{align}
\label{eq:su11contvec}
h_\pm \circ \text{Ad}\left(g,\{F_i\}\right)&= \braket{l^\pm \,|\, g^\dagger {\sigma_3} \, F_i g \,{\sigma_3} \cdot l^\pm} \hat{e}^i \nonumber\\
&= \braket{l^\pm \,|\, g^\dagger {\sigma_3} \, F_i g \,\cdot l^\mp} \hat{e}^i \nonumber \\
&=\pm\frac{i}{2}\left(-\sinh{t},\,-\cosh{t} \cos{\phi},\,\cosh{t} \sin{\phi}\right)\,.
\end{align}
Once more the hyperboloid is oriented along $\hat{e}^1$, and again we note that the inner product in the second line of the previous equation can be understood in terms of the $\text{SU}(1,1)$ invariant inner product as $\braket{g \cdot l^\pm, \, F_ig \cdot l^\mp}_{\sigma_3}.$

\section{Convex geometry in Minkowski space-time}
\label{appendix:geometry}

Here we understand Minkowski space-time $\mathbb{R}^{3,1}$ to be the affine space modelled on $\mathbb{R}^4$, together with the Lorentzian metric $\eta=\text{diag}(1,-\vec{1})$. In stark contrast with Euclidean space, where one has a constant vanishing Gaussian curvature, the difficulty in describing geometric objects like polytopes in Minkowskian space is essentially due to the non-definite Lorentzian metric, which endows the space with four subsets of constant curvature (these being the space-like and time-like hyperboloids, as well as the light-like cone). Thus, while one may define polytopes in Minkowskian space by restricting them to those constant curvature spaces (essentially constructing the so-called \textit{hyperbolic polyhedra}), the usually useful geometrical measures of angles, areas and volumes are not immediately available in the whole of Minkowski space. Nonetheless, Minkowski space-time is still an affine space in the affine geometry sense, and certain geometrical notions are therefore readily available to us, namely all those that can be described by a vector space structure without the use of a metric (this includes, for example, having the notion of parallelism, and being able to compare colinear line segments). To these we may add those geometrical properties deriving from the metric which hold globally, of which orthogonality is perhaps the clearest example.

\subsection{Basic algebra of Minkowskian 3-space geometry} \label{app:Minkowski}

In the following we make a couple of simple observations about convex geometry in $\mathbb{R}^{2,1}$. We will always exclude light-like vectors from our analysis below\footnote{Light-like lines overlap the notions of orthogonality and parallelism, and are therefore unsuitable for our purposes.}.
\begin{enumerate}
    \item Minkowskian triangles and tetrahedra

    The convex hull of any three points not all colinear is a \textit{Minkowskian triangle} if:
    \begin{enumerate}
        \item Every edge of the triangle is non-null,
        \item The triangle is not contained in a null plane\footnote{A null hypersurface is defined to be orthogonal to a null vector; the induced metric on such hypersurfaces may be degenerate.}.
    \end{enumerate}

    The convex hull of any four points not all colinear is a \textit{Minkowskian tetrahedron} if:
    \begin{enumerate}
        \item Every edge of the tetrahedron is non-null,
        \item Every face is non-null.
    \end{enumerate}

    \item Orthogonal projections

    The orthogonal projection of a vector $v$ onto $u$ is given by $\text{Proj}_u v=\norm{u}^{-2}\braket{v,u}$. Let $\{u, u^\perp_i\}$ be a basis with $\braket{u,u_i^\perp}=0$, and $v=au+b^i u^\perp_i$. Then $\text{Proj}_u v=a \norm{u}^{-2} \braket{u,u}=a$.

    \item Half-spaces

    Let $v$ be a unit vector (i.e. $\norm{v}^2=\pm1$). The set $H_v=\left\{x \; \big|\, \braket{x,v}\norm{v}^2 \leq 0 \right\}$ defines a half-space through the origin,  orthogonal to $v$, and opposite to $v$ in the sense that $v \notin H$.
    The set $H_v$ can be translated in the direction of $v$ by a positive amount $h$, defining the translated half-space $H^h_v=\left\{x \; \big|\, \braket{x,v}\norm{v}^2 \leq h \right\}$. Notice that the half-space so defined will always contain the origin.

    \item Height of triangles and tetrahedra

    Consider the triangle $ABC$, and let $v$ be the outward unit normal vector to the opposite edge $AC$ to $B$. The height $h_b$ of the triangle $ABC$ from the vertex $B$ is defined to be the unique positive number such that $\vec{B}+h_b v$ lies on the line $AC$. In terms of the remaining edges, the triangle height is given by $h_b=|\text{Proj}_v \vec{AB}|=|\text{Proj}_v \vec{CB}|$.

    These definitions extend in the obvious manner to tetrahedra.
\begin{figure}[h]
    \centering
    	{\includegraphics[valign=c,scale=0.75]{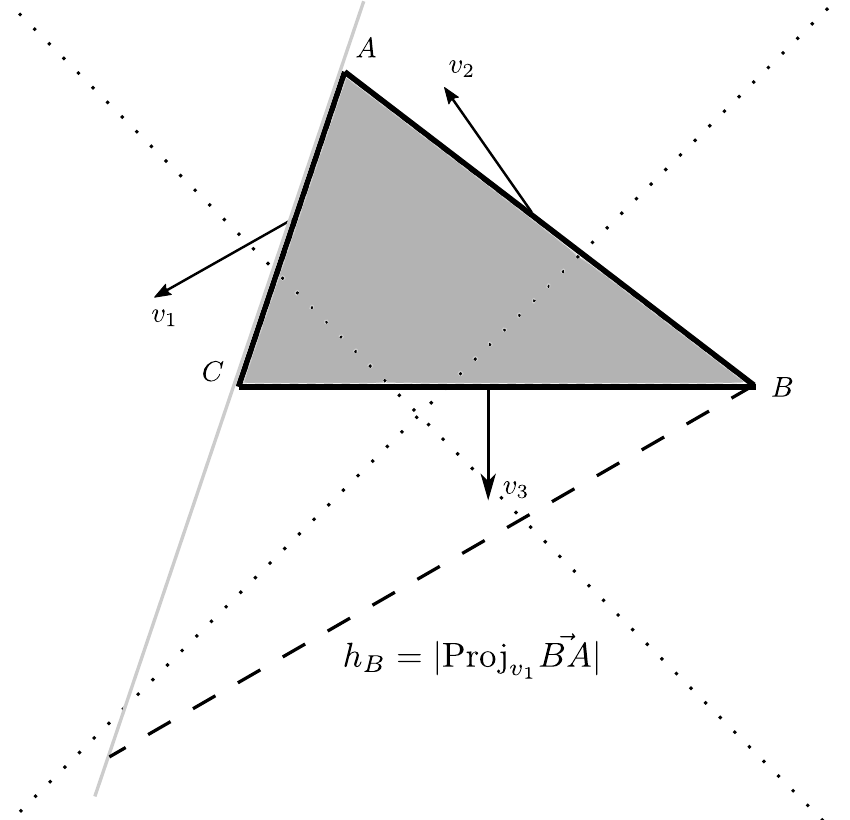}}
    	\caption{An example of a Minkowskian triangle in $\mathbb{R}^{1,1}$ and respective normal vectors. Also pictured is the triangle height $h_B$ from the vertex $B$.}
\end{figure}

    \item Squared areas and volumes of triangles and tetrahedra

    Consider a tetrahedron ABCD, and let $a,b$ and $c$ be edge vectors of the tetrahedron with the same base at vertex $D$ and end-points at A, B and C, respectively. We define the tetrahedron volume by $V^2_{abc}=\frac{1}{(3!)^2}\norm{\star(a \wedge b \wedge c)}^2$, and the area of the triangle with sides $a,b$ by $A^2_{ab}=\frac{1}{(2!)^2}\norm{\star(a\wedge b)}^2$. Notice that this is the same definition one may have for areas and volumes in euclidean space, but now we allow for negative squared areas and volumes, depending on the signature of the metric of the space.

    \item Signs of squared areas

    Given the above definition, the signs of the squared areas of triangles depend on their causal character. Indeed, let $a,b$ be edge vectors of a time-like triangle, both with base $A$. Since the face is time-like, the span of its edges must contain both space- and time-like vectors. Thus the quadratic $f(x)=||a x + b||^2$ must change sign, implying its discriminant $\Delta$ must be positive. But we also have that $\Delta=-4||a\wedge b||^2$, and hence $A_{ab}^2< 0$. If one considers a space-like triangle, the same argument implies that the polynomial must not change sign, and this shows that $A_{ab}^2> 0$.

    We may thus claim that time-like triangles are characterised by negative squared areas, while space-like triangles have positive squared areas.

    \item Orthogonal vectors to triangles

    Consider again the triangle $ABC$ and the edges $a,b$, both with the same base point. Then the vector $v=\frac{\star(a\wedge b)}{\sqrt{4|A_{ab}^2|}}$ is a unit vector orthogonal to the triangle. Orthogonality follows from the properties of the Hodge star: let $\omega$ be the volume form induced by the metric. Then $\braket{\star(a\wedge b),b}\omega = \braket{a\wedge b,\star b} \omega =a\wedge b\wedge b =0$, and analogously $ \braket{\star(a\wedge b),a}\omega=0$. Since $\omega$ is non-degenerate by definition, orthogonality holds.

    \item Squared volume formula for a tetrahedron in terms of boundary areas

    Let $a,b,c$ be edge vectors for a tetrahedron $ABCD$ as before, all having a common base point at D. Recall that $V^2_{abc}=\frac{1}{(3!)^2}\norm{\star(a \wedge b \wedge c)}^2$. Define the orthogonal unit vector to the face $ACD$ by $v=\frac{\star(a\wedge c)}{\sqrt{4|A_{ac}^2|}}$, pointing out of the tetrahedron, and the tetrahedron height from the vertex opposite to the same face by $h_b=|\text{Proj}_v b|$. Then we may write $b=-h_b v+\alpha a + \beta c$, for some numbers $\alpha,\beta$. The following holds:
    \begin{align*}
    V^2_{abc}&=\frac{1}{(3!)^2}\norm{\star(a \wedge c \wedge (h_bv))}^2 \\
    &=\frac{1}{(3!)^2}\det \begin{pmatrix}
        \braket{a,a} & \braket{a,c} & h_b\braket{a,v} \\
        \braket{c,a} & \braket{c,c} & h_b\braket{c,v} \\
        h_b\braket{v,a} & h_b\braket{v,c} & h_b^2\braket{v,v}
    \end{pmatrix} \\
    &=\frac{1}{(3!)^2} h_b^2\norm{v}^2 A_{ac}^2=\frac{1}{(3!)^2} h_b^2 |A_{ac}^2|\,,
    \end{align*}
    since $v \perp a,b$.
    We may conclude that Minkowskian tetrahedra, defined in this manner, have positive squared volumes independently of their causal character.

    \item Areas and volumes of triangles and tetrahedra

    Because we have shown that squared volumes are always positive, we may define the volume of a tetrahedron $P$ by $V_P=\sqrt{V_P^2}$. For areas of triangles $Q$, we take $A_Q=\sqrt{|A_Q^2|}$.

    \item Minkowskian polygons and polyhedra

    Since we have a notion of orthogonal half-spaces, we define polyhedra as follows: a convex Minkowskian polyhedron is a finite and bounded intersection of half-spaces containing the origin,
    \begin{equation}
    P=\bigcap_{\{ v,h\}} H_v^h\,,
    \end{equation}
    with each half-space being characterised by an orthogonal non-null vector $v$ and a positive height $h$ from the origin along $v$. We further require that every edge of $P$ is non-null.

    Every face of such a polyhedron is a Minkowskian polygon.
    
    \begin{figure}[h]
    	\centering
    	\includegraphics[valign=c,scale=0.75]{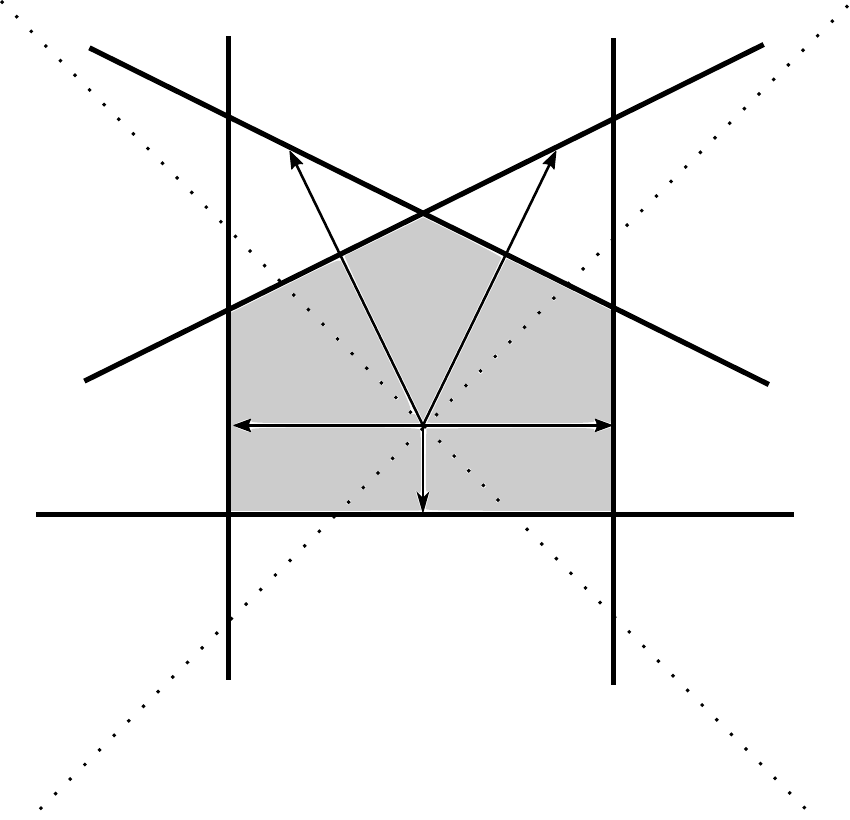}
    	\caption{An example of a Minkowskian polygon, defined as the intersection of half-spaces in Minkowski space-time $\mathbb{R}^{1,1}$ containing the origin. Each vector is orthogonal to the line it touches.}
    \end{figure}

    \item Additivity of areas and volumes

    We have so far only defined areas of triangles and volumes of tetrahedra. In order to extend the definitions to more general polygons and polyhedra, we need to determine to which extent areas and volumes may be added to each other. Since our definitions have relied on the interior product and the Hodge star, both linear maps (or equivalently on the wedge product), areas and volumes are naturally additive \textit{in their own subspace of definition}. That is to say, although there is no sense in which one might add the area of a time-like triangle and a space-like triangle, our definitions naturally allow for summing areas of parallel triangles. Volumes, on the other hand, are defined on the whole three-dimensional space, are always positive, and may freely be added to each other. Note that then our tetrahedron is a bit of a strange geometrical object, in the sense that it has individual face areas and a total volume, while not having a total area in general.

    \item Areas of polygons

    Consider a general convex polygon, and place a vertex in its interior. Now join every vertex on the boundary of the polygon to the interior vertex, obtaining a triangulation of the polygon. The polygonal area will be defined as the sum of the triangle areas.

    \item Volumes of polyhedra

    Analogously to what we did previously, given a general convex polyhedron, consider its triangulation by tetrahedra: triangulate first every face as above, obtaining triangular faces $f$, and then join every vertex to a new vertex in the interior of the polyhedron\footnote{It is well-known that in more than two dimensions there exist non-convex polyhedra not admitting a triangulation by simplices.}. The total volume of the polyhedron is well-defined as the sum of the volumes of the individual tetrahedra, and according to the previous discussion it is given by
    \begin{equation}
    \label{volume}
    V=\frac{1}{3!}\sum_f h_f A_f\,,
    \end{equation}
    where $A_f$ is the area of the face $f$ and $h_f$ is the tetrahedral height from the interior vertex to the plane defined by the face $f$.

    \item Closure condition for polyhedra

    We show this for a tetrahedron, as the generalization to other polyhedra should be clear from the previous discussion. Let $a,b,c$ be edge vectors of the tetrahedron, all with the same base point, and consider the four vectors normal and outward-pointing to its faces: $\star(b \wedge a), \star(c\wedge b), \star(a\wedge c), \star[(c-b)\wedge (a-b)]$. Then it is immediate that 
    \begin{equation*}
    \star(b \wedge a)+ \star(c\wedge b)+ \star(a\wedge c)+ \star[(c-b)\wedge (a-b)]=0\,.
    \end{equation*}
    Thus, for any convex polyhedron,
    \begin{equation}
        \sum_f v_f A_f=0\,,
    \end{equation}
    where $v_f$ is the unit vector orthogonal to the face $f$ and outward-pointing.
\end{enumerate}

\subsection{Angles in the Minkowski plane}

In order to discuss some further properties of polytopes in Minkowski space-time we will need the notion of angle between two arbitrary non-null vectors. When the plane defined by the two vectors of interest is entirely space-like one may make use of Euclidean  angles, which we define in the usual way as
\begin{equation}
\cos{\theta_{uv}}=\frac{\braket{u,v}}{||u|| \, ||v||}.
\end{equation}
However, it might be that the plane spanned by the two vectors has the metric structure of $\mathbb{R}^{1,1}$, and this requires a more careful discussion. We will use the description found in \cite{Alexandrov2003}, since it allows for keeping the property of angle additivity and constructing an analogous Schläfli identity to the Euclidean case \cite{SuarezPeiro2000}.

\begin{figure}
\centering
\includegraphics[valign=c,scale=0.65]{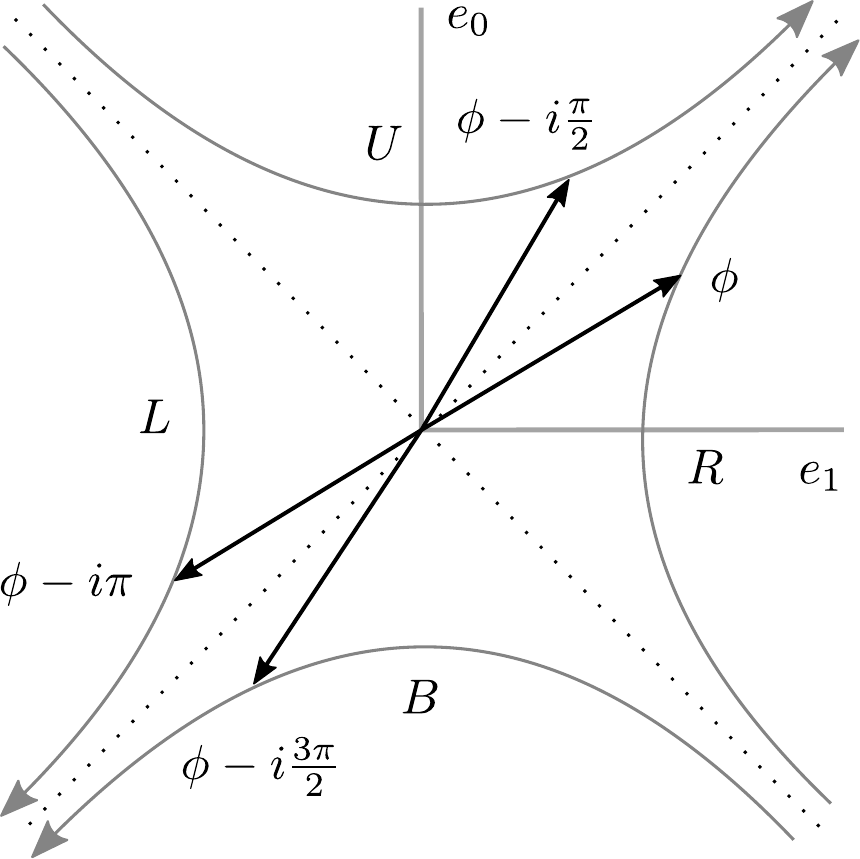}
\caption{\label{anglespic}Standard orientation of hyperbolas in $\mathbb{R}^{1,1}$. Just as Euclidean  orthogonal vectors are separated by a $\frac{\pi}{2}$ angle, so too are Minkowski ones separated by $-i \frac{\pi}{2}$.}
\end{figure}

Let $R,U,L,B$ denote the right, upper, left and bottom regions in $\mathbb{R}^{1,1}$ separated by the light cone, respectively. We will orient the hyperbolas of those regions as in Figure \ref{anglespic}. Consider the ordered pair of vectors $(u,v)$, positively oriented (\textit{i.e} agreeing with the orientation of the hyperbolas), and assume both lie in the same region, having thus the same causal character. As usual, the ordered angle $\theta_{uv}$ from $u$ to $v$ is given by
\begin{equation}
\cosh{\theta_{uv}}=\frac{\braket{u,v}}{||u|| \, ||v||}\,,
\end{equation}
and the sign of $\theta_{uv}$ is fixed by the causal character: $\theta_{uv}$ is taken to be positive for space-like vectors, and negative for time-like ones. We also define $\theta_{uv}=-\theta_{vu}$, i.e. the angle for negatively-oriented vectors is the symmetric of the angle for positively-oriented ones.

In order to extend this definition to angles between vectors of different causal character, we analytically continue the angles $\theta$ by complementing them with an imaginary part $i\varphi$. Thus, for any non-null vector $u$, we define its angle $\theta_{1u}=\phi_{1u} -i \varphi_{1u}$ with the vector $e_1$ through the formula
\begin{align}
\label{eq:angles}
\cosh\left(\theta_{1u} \right)&=\frac{\braket{u,e_1}}{||u|| \, ||e_1||} \\
&=\cosh \phi_{1u} \cos \varphi_{1u} - i \sinh \phi_{1u} \sin \varphi_{1u} \nonumber\,,
\end{align}
restricting $\varphi_{1u}$ to values in $\{0,\frac{\pi}{2}, \pi, \frac{3\pi}{2} \}$, depending on whether $u$ lies in $R,U,L$ or $B$, respectively. We take the norm to be such that $||\cdot|| = \sqrt {||\cdot||^2} \in \mathbb{R}^+_0 \cup i  \mathbb{R}^+_0$. In this manner we are led to think of the four arms of the light-cone as localised at one of the four values of $\varphi$, and of $\phi$ as the angle between vectors lying in the same region. Since under this definition angles are additive \cite{Alexandrov2003}, we may then define angles between positively oriented arbitrary vectors $(u,v)$ as $\theta_{1v}-\theta_{1u}$. Note that in the case when $u$ is of the same causal character as $e_1$ the left-hand side reduces to $\pm \cosh \phi_{1u}$, while if it is of a different character it reduces to $\pm \sinh \phi_{1u}$.

\subsection{Uniqueness and existence of Minkowskian polyhedra}
Having established the above definitions and properties, we now turn to formulating an analogous Minkowski theorem for Minkowskian polyhedra. We cite a famous result by Alexandrov \cite[Theorem 1 of section 6.3]{Alexandrov2005}:
\begin{theorem}[Alexandrov]
\label{Alex}
Let the word ``face" stand for a vertex, edge or proper face of a polyhedron, and define two faces to be parallel if they are contained in parallel support planes. If for all pairs of parallel faces of two convex Euclidean  polyhedra neither face can be placed strictly inside the other by parallel translation, then the polyhedra are translates of one another.
%Let $f(Q)$ be a monotone function of a polygon $Q$, i.e., we have $f(Q_1)>f(Q_2)$ whenever $Q_2$ can be placed inside $Q_1$ by translation. If two convex polyhedra satisfy $f(Q_1)=f(Q_2)$ for each pair of parallel faces $Q_ 1$ and $Q_2$ , then the polyhedra are translates of one another.
\end{theorem}
Alexandrov proves this theorem in the context of Euclidean  3-space, but under our definition of convex polyhedra it still holds for Minkowski space-time.
\begin{corollary}
	If for all pairs of parallel faces of two convex polyhedra in $\mathbb{R}^{2,1}$ Minkowski space-time neither face can be placed strictly inside the other by parallel translation, then the polyhedra are translates of one another.
\end{corollary}
\begin{proof}
	Note that there is an identity mapping $\mathbb{R}^{2,1} \rightarrow \mathbb{R}^3$, and that a set is convex in $\mathbb{R}^{2,1}$ if and only if it is convex in Euclidean  space. Suppose that we are given polyhedra $P,P'$ in the conditions of the theorem, and consider their Euclidean  image. These images satisfy the requirements of Theorem \ref{Alex}, and thus they are translates of each other. But then, under the identity mapping, so too are the original polyhedra $P,P'$.
\end{proof}
We then have an immediate corollary on congruence of polyhedra depending on their face areas:
\begin{corollary}
\label{cor2}
Let $P,P'$ be two convex polyhedra in either $\mathbb{R}^{2,1}$ or $\mathbb{R}^{3}$, defined as above. Denote by $Q$ a face of $P$, and by $Q'$ a face of $P'$. Moreover, let $\{v_{Q},A_{Q}\}, \{v_{Q'},A_{Q'}\}$ be the sets of outward-pointing orthogonal vectors to their faces, none of them light-like, as well as their respective areas. If and only if both sets are the same, then $P$ and $P'$ are translates of each other. That is, convex polyhedra in Minkowski space-time are uniquely characterised by their face areas and normals, up to translations.
\end{corollary}

\begin{proof}
    Since both sets of vectors are the same, the polyhedra share pairwise-parallel faces $(Q,Q')$. Consider the function $f(Q)=A_Q=\sqrt{|A_Q^2|}$, which is well-defined and monotonic on parallel polygons independently of their causal character. By virtue of the monotonicity of $f$, parallel faces of equal area cannot be placed strictly inside each other by a translation; Theorem \ref{Alex} then implies the result.
\end{proof}

On the other hand, we may also prove existence of polyhedra given some boundary data satisfying the closure condition, essentially repeating the proof due to Minkowski \cite{Alexandrov2005}, thus establishing both uniqueness and existence:
\begin{theorem}[Minkowski's theorem]
\label{theorem:minkowski}
Let $\{v_f,A_f\}$ be a set consisting of unit vectors $v_f$ in either $\mathbb{R}^3$ or $\mathbb{R}^{2,1}$ and positive numbers $A_f$. Suppose such vectors are not all co-planar, no vector is light-like, and the following holds
\begin{equation}
\label{alex_closure}
     \sum_f v_f A_f=0\,.
\end{equation}
Then there exists a unique convex polyhedron in $\mathbb{R}^3$ or $\mathbb{R}^{2,1}$, respectively, such that $A_f$ is the area of its face $f$ and $v_f$ points orthogonally outward to the face.
\end{theorem}
\begin{proof}
    Assume we are given $F$ such vectors and $F$ such numbers. Consider all sets $h$ containing $F$ non-negative numbers $h_f$, to be understood as distances from the origin. Equation \eqref{alex_closure} implies that the vectors $v_f$ cannot all point towards the same half-space, and thus, for every $h$, the intersection of the hyperspaces $H_{v_f}^{h_f}$ in the relevant space $\mathbb{R}^ 3$ or $\mathbb{R}^{2,1}$ defines convex polyhedra $P_h$. Let $\tilde{A}_f$ be the area associated to the hyperplane boundary of $H_{v_f}^{h_f}$, and set it to zero if that hyperplane does not define a face of $P_h$. Among all $P_h$, consider those satisfying the constraint
    \begin{equation}
    \label{constraint}
    \sum_f h_f A_f=1\,.
    \end{equation}
    We now show that there exists a polyhedron which maximises the volume under the above condition. First, note that all $h_f$ are bounded from above by $h_f\leq1/A_f$, and the constraint is a closed condition. Thus the set of admissible $h_f$ is compact. Hence the volume of $P_h$ attains a maximum in the domain of the constraint at some $P_h^*$. Equation \eqref{alex_closure} implies that if $P_h^*$ satisfies \eqref{constraint} so does a translation of itself; we may then assume that $P_h^*$ contains the origin, and therefore $h_f^*>0$. Using Lagrange multipliers, the extrema are found at
    \begin{equation}
    \begin{gathered}
    \frac{\partial}{\partial h_f}\left( V(P_f)+ \lambda \left(\sum_{f'} h_{f'} A_{f'} -1\right)\right)=0 \nonumber \\
    \Rightarrow \frac{1}{3!}\tilde{A}_f + \lambda A_f=0\,,
    \end{gathered}
    \end{equation}
     where we used the volume formula \eqref{volume}. Thus we have a polyhedron $P_h$ such that the vectors $v_f$ are orthogonal to its faces, and upon a suitable rescaling its areas are given by the $A_f$. By Corollary \ref{cor2}, this polyhedron is unique up to translations.
\end{proof}

\subsection{Rigidity of Minkowski polytopes}

We now turn to the question of whether convex polytopes in space-time are rigid. We will call such a polytope rigid if every continuous displacement of its vertices leaving the lengths of its edges invariant and preserving its combinatorics amounts to an orthogonal transformation with respect to the space-time metric, i.e. a congruence. That every convex polytope in Minkowski space of dimension $\geq 3$ is rigid, much like their Euclidean  counterparts, can be shown straightforwardly by making use of a result from \cite[Lemma 9]{Alexandrov2003}.

\begin{lemma}
	Let $P(t)$ be a smooth family of convex orientable polyhedra such that, for each $t$, each edge of $P(t)$ is non-null, its length is invariant, and each face carries a non-degenerate metric. Let $f_1,f_2$ denote the faces adjacent to an edge $e$. Denote by $n_{f_i}^e$ the unit vector which lies in $f_i$, is orthogonal to $e$, and points inside $f_i$, and by $m_{f_i}$ the outward-pointing normal unit vector to $f$. The dihedral angle (i.e. the angle between $m_{f_1}$ and $m_{f_2}$) at an edge $\theta_e(t)$ then satisfies the velocity equation
	\begin{equation}
	\label{angularvelocity}
	\frac{\diff \theta_e}{\diff t}=\sum_i\braket{\frac{\diff m_{f_i}}{\diff t}, n_{f_i}^e}\,.
	\end{equation}
\end{lemma}
Using this angular velocity equation, we now show that polyhedral corners, i.e. non-compact polyhedrons with a single vertex, satisfy a closure condition.
\begin{lemma}
	Let $P$ be a polyhedral corner, and denote by $e$ the unit edge vectors with base point at the vertex. Then $P$ satisfies the closure condition
	\begin{equation}
	\sum_e  \epsilon(e) e	\frac{\diff \theta_e}{\diff t}=0\,,
	\end{equation}
	where $\epsilon(e)$ is the sign of the squared norm of $e$.
\end{lemma}
\begin{proof}
	This follows directly from equation \eqref{angularvelocity}. We start by summing the inner product appearing in that equation over every face $f$ of $P$
	\begin{align*}
	\sum_f \sum_{e \in \partial f} \epsilon(e) e \braket{\frac{\diff m_{f}}{\diff t}, n_{f}^e}&=\sum_f \sum_{e \in \partial f}\epsilon(e) \left[\braket{\frac{\diff m_{f}}{\diff t}, n_{f}^e \wedge e} + n_f^e \braket{\frac{\diff m_{f}}{\diff t}, e} \right] \\
	&=\sum_f \braket{\frac{\diff m_{f}}{\diff t}, \sum_{e \subset f} \epsilon(e) n_{f}^e\wedge e}\,,
	\end{align*}
	where the inner product between 1- and 2-vectos is short-hand for the interior product.  In the second equality we use the fact that $\diff m_{f}/\diff t$ must be orthogonal to $m_f$ and tangent to the arc of rotation $\theta_e(t)$, and thus normal to $e$. Focusing now on a single face, note that associated to the induced metric at the face there is a Hodge star $\overline{\star}$ and a volume-form $\overline{\omega}$. Denoting by $e_1,e_2$ the two edges incident to $f$, we then have
	\begin{align*}
	\sum_{e \in \partial f} \epsilon(e) n^e_f \wedge e &= \epsilon(e_1) n^{e_1}_f\wedge e_1+ \epsilon(e_2) n^{e_2}_f\wedge e_2 \\
	&=\pm \left(\epsilon(e_1) \overline{\star} e_1 \wedge e_1- \epsilon(e_2) \overline{\star} e_2\wedge e_2\right) \\
	&=\pm \left( \epsilon(e_1)\braket{e_1,e_1}-\epsilon(e_2)\braket{e_2,e_2}\right) \overline{\omega}=0\,,\\
	\end{align*}
	and the sign indeterminacy is due to the possible orientations of $\overline{\star}$. We thus have that the left-hand side of the previous series of equations above vanishes, and interchanging the summations proves the result.
\end{proof}
We need one more ingredient for the proof of rigidity, in some sense associated to the combinatorics of polyhedral corners.
\begin{lemma}
	\label{0or4}
	Consider a family of polyhedral corners $P(t)$, as above. Label the edges of $P(t)$ according to the sign of $\epsilon(e)\frac{\diff \theta_e}{\diff t}$, leaving an edge unlabelled if the associated angular velocity vanishes. Then either there are 4 or more sign changes as one goes around the vertex, or all edges are unlabelled.
\end{lemma}
\begin{proof}
	Clearly there must be an even number of sign changes, because in going around a vertex one always returns to the initial sign. We therefore need to show that 0 and 2 sign changes are impossible.

	First, due to the closure condition, one cannot have 0 sign changes, as otherwise the linear combination $\sum_e  \epsilon(e) e	\frac{\diff \theta_e}{\diff t}$ would not vanish, contradicting the previous lemma.

	Secondly, suppose there are exactly 2 sign changes. Assume w.l.o.g. that $e_1,...,e_n$ are labelled with $+$, and $e_{n+1},...,e_m$ are labelled with $-$. Consider the half-space $H^+$ separating the two sets of edges and containing the ones labelled with a positive sign. Then clearly $\sum_{i=n+1}^m \epsilon(e_i) e_i	\frac{\diff \theta_{e_i}}{\diff t}$ also lies in $H^+$, and thus the full linear combination $\sum_e  \epsilon(e) e	\frac{\diff \theta_e}{\diff t}$ does too, contradicting once more the lemma above.
\end{proof}

We now state a theorem on the rigidity of convex polyhedra in Minkowski space-time.

\begin{theorem}
	\label{rigid}
	Let $P$ be a convex polytope in Minkowski space-time of dimension $d \geq 3$, defined as above. Then every smooth deformation of $P$ preserving the length of its edges and such that the metric at each face remains non-degenerate is an isometry of the Minkowski metric, i.e. a congruence of $P$. That is, space-time polyhedra are rigid.
\end{theorem}
\begin{proof}
	If we prove this statement first for 3 dimensions, since higher dimensional objects are characteristically more constrained as noted in \cite{Alexandrov2005}.

	Consider a smooth family $P(t)$ of convex 3 dimensional polyhedra containing $P$, and label the polyhedron edges according to the rules of Lemma \ref{0or4}. By the same lemma, at each vertex there can be either no labels at all or at least 4 sign changes. By Cauchy's combinatorial lemma \cite[Section 2.1]{Alexandrov2005}, in a planar graph where edges are either not labelled or labelled with a sign there cannot be at least 4 sign changes around a vertex. All convex polyhedra induce graphs through its vertices and edges which can be embedded on an Euclidean  sphere, and thus those graphs are planar. It thus follows that no edge of $P$ can be labelled, and therefore all angular velocities vanish. This is enough to establish rigidity, since then every smooth deformation must preserve the dihedral angles.

	For a given 4 dimensional polytope, consider its intersection with an Euclidean  3-sphere around one of its vertices. The resulting intersection is a 3 dimensional convex polyhedron, and it is rigid from the discussion above. Thus the polytope itself is rigid, and the argument can be extended to higher dimensions.
\end{proof}

\printbibliography

\end{document}